\documentclass[english]{article}
\usepackage{custom_tex}

\title{SymmPI: Predictive Inference for Data \\ 
with Group Symmetries}
\author{Edgar Dobriban and Mengxin Yu\footnote{
The authors are with the
Department of Statistics and Data Science,  
University of Pennsylvania,
E-mail addresses: \texttt{dobriban@wharton.upenn.edu},
\texttt{mengxiny@wharton.upenn.edu}. The authors have contributed equally.
}}

\begin{document}
\maketitle
\begin{abstract}
Quantifying the uncertainty of predictions is a core problem in modern statistics.
Methods for predictive inference have been developed under a variety of assumptions, often---for instance, in standard conformal prediction---relying on the invariance of the distribution of the data under special groups of transformations such as permutation groups.
Moreover, many existing methods for predictive inference aim to predict unobserved outcomes in sequences of feature-outcome observations.
Meanwhile, there is interest in
predictive inference under more general observation models (e.g., for partially observed features) and for data satisfying more general distributional symmetries (e.g., rotationally invariant observations in physics).
Here we propose  SymmPI, a methodology for predictive inference 
when data distributions have general group symmetries in arbitrary observation models.
Our methods leverage the novel notion of \emph{distributionally equivariant} transformations, which process the data while preserving their distributional invariances.
We show that SymmPI has valid coverage under distributional invariance and characterize its performance under distribution shift, recovering recent results as special cases.
We apply SymmPI to predict unobserved values associated to vertices in a network, where the distribution is unchanged under relabelings that keep the network structure unchanged.
In several simulations in a two-layer hierarchical model, and in an empirical data analysis example, SymmPI performs favorably compared to existing methods.
\end{abstract}

\tableofcontents

\section{Introduction}

Prediction is one of the most important problems in modern statistical learning.
Since unobserved data cannot always be predicted with certainty, 
quantifying the uncertainty of predictions is a crucial statistical problem, studied in the areas of 
\emph{predictive inference} and \emph{conformal prediction} \citep[e.g.,][]{geisser2017predictive,vovk2005algorithmic}.
Numerous predictive inference methods have been developed under both parametric and nonparametric conditions \citep[e.g.,][etc]{Wilks1941,Wald1943,vovk1999machine,vovk2005algorithmic,lei2013distribution,lei2014classification,Chernozhukov2018,romano2019conformalized}; see the related work section for more examples.

Among these, conformal prediction (or inference) has been gaining increasing attention recently because it can lead to
prediction sets with
finite-sample coverage guarantees
under reasonable conditions on the data, such as the exchangeability of the datapoints.
Moreover, this exchangeability condition is preserved under natural permutation-equivariant maps \citep[see e.g.,][]{dean1990linear,kuchibhotla2020exchangeability}. 
This implies that residuals constructed from statistical learning methods that are invariant with respect to the data---such as $M$-estimators---remain exchangeable, and can be used for conformal inference. 
Conformal prediction has been applied and extended
to a wide range of statistical machine learning problems, including 
non-parametric density estimation and regression 
\citep{lei2013distribution,lei2014distribution,lei2018distribution},
quantile regression \citep{romano2019conformalized},
survival analysis \citep{candes2023conformalized,gui2022conformalized}, etc.

At the same time, predictive inference methods have been developed under assumptions different from exchangeable independent and identically distributed data, 
including for datapoints $Z_n$ in
sequential observation models called
\emph{online compression models} 
leading to data $Z_1,\ldots, Z_{n-1} \in \mZ_0$ for some space $\mZ_0$,
and for non-sequential models called \emph{one-off-structures} \citep{vovk2005algorithmic}. 
These are closely related to the classical statistical notion of \emph{conditional ancillarity}
\citep{cox2006principles,dobriban2023joint}.
Methods have also been developed under more concrete assumptions such as 
hierarchical exchangeability \citep{lee2023distribution}, 
exchangeable network data \citep{lunde2023conformal},
and invariance under a finite subgroup of the permutations of the datapoints \citep{Chernozhukov2018}.

However, at the moment, there 
are no {unified} predictive inference methods 
(1) for \emph{arbitrary unobserved functions} of \emph{arbitrary data} {$Z\in \mathcal{Z}$} {beyond exchangeability} (e.g., a network, a function, etc, {where $\mathcal{Z}$ is a measurable space}.) whose distributional symmetries are characterized by an  \emph{arbitrary---possibly infinite and continuous---group} (e.g., a rotation group that arises for rotationally invariant data); 
and 
(2) that enable \emph{processing the data} in flexible ways to keep the distributional symmetries, similar to what is possible in the special case of conformal prediction.
In this paper, we develop such methods. More specifically:
\begin{enumerate}
\item We consider datasets whose distributions are invariant under general---compact Hausdorff topological---groups {$\mathcal{G}$}. We argue that invariance under such groups  {(i.e., for any $g\in \mathcal{G},$ it holds that $\rho(g)Z=_d Z$, where $Z\in \mathcal{Z}$ and $\rho(g)$ is an action of $\mathcal{G}$ on $\mathcal{Z}$)}, 
is of broad interest
and includes in particular invariance under all finite groups (e.g., exchangeability, hierarchical exchangeability, cyclic shifts). 
It also includes invariance under continuous groups such as rotations (and combinations such as rotations and translations), which are of broad interest in the physical sciences. 
For example many quantities are rotationally invariant in physics \citep{gross1996role,robinson2011symmetry,schwichtenberg2018physics}.
Continuous (e.g., rotational) invariances are also common for image data.
The datasets we consider are not restricted in any other way, and in particular we are not limited to sequential observations $Z_1,\ldots, Z_n$.

\item {In \Cref{sec:ex}, we introduce four examples that we will study through the paper:
(a) Exchangeable data; 
(b) Network-structured data;
(c) Rotationally invariant data; and
(d) Two-layer hierarchical data.
The first one is included to show how our methods recover existing ones from conformal prediction.
For (b), we study prediction sets on networks, where random variables associated with a network are assumed to have a distribution invariant under any transformation that keeps the network structure unchanged (i.e., to the automorphism group of the graph).
For (c), we focus on rotationally invariant observations that are made in a particular coordinate system, such as for a variety of physical quantities. 
We return to (d) {in point 5} below.}

\item For our general theory, we introduce the key notion of \emph{distributional equivariance} of transformations  $V:\mathcal{Z}\rightarrow \tilde{\mathcal{Z}}$ 
of the data
for a measurable space  $\tilde{\mathcal{Z}}$.
This means that 
$V(\rho(G)V(Z))=_d \tilde{\rho}(G)V(Z)$, where
$G$ follows the uniform distribution $U$ over $\mathcal{G}$
and $\tilde{\rho}$ is an action of $\mathcal{G}$ over $\tilde{\mathcal{Z}}$.
We show that it is enough to preserve \emph{distributional invariance}. We explain how this allows us to process data to extract meaningful features that enable constructing accurate prediction sets, and---for instance---adapting to data heterogeneity in a 
two-layer hierarchical model
example.
We also allow \emph{arbitrary group actions} on the input and output spaces, not limited, e.g., to permutation actions.
We allow the observed component of the data be 
determined by an arbitrary function  that we call the \emph{observation function}.
For instance, this can include any part of the features in a supervised learning setting.
In particular, we are not limited to predicting an outcome $Y_n$ after observing 
feature-label pairs 
$(X_1,Y_1), \ldots, (X_{n-1},Y_{n-1})$ and features $X_n$.

\begin{algorithm}[h]
\caption{SymmPI: Predictive Inference for Data with Distributional Symmetry}\label{alg:symm_pi}
\begin{algorithmic}[1]
\Require Data $z$ satisfying distributional invariance under group $\mG$.
Distributionally equivariant map $V: \mZ \to \tilde\mZ$. 
Test function $\psi: \tilde\mZ \to \R$.
Observation function $\o: \mZ \to \mO$.
Coverage level $1-\alpha \in [0,1]$.
Observed data $\zo = \o(z)$.

\For{$z \in \mZ$}
\State Compute $\tilde z = V(z)$
    \State Compute $t_{\tz}
 =
Q_{1-\alpha}
\left(\psi(\trG  \tz),\, G\sim U \right)$ \Comment{{$Q_{1-\alpha}$ denotes the $1-\alpha$th quantile}}
\EndFor

\State \textbf{return} $T^{\mathrm{SymmPI}}(\zo) = \{z\in \mZ: \psi(V(z))\le t_{V(z)},\, \o(z)= \zo \}$
\end{algorithmic}
\end{algorithm}
\item 
We propose \OurMethod, a method for predictive inference for data with distributional symmetries in the above setting {with arbitrary test function $\psi$ and observation function $\o$}. {
High level pseudo-code is presented in \Cref{alg:symm_pi}. To implement the algorithm and construct the prediction set, we primarily need to calculate $t_{V(z)}$, which represents the $1-\alpha$ quantile of $\psi(\trG V(z))$, by sampling $G\sim U$, for $z\in\mZ$. 
This calculation assumes that the data $Z$ exhibits distributional invariance, and $V$ is distributionally equivariant under the group $\mathcal{G}$. 
Prediction sets for (a) exchangeable data\footnote{When $\psi(\tilde{z})=\tilde{z}_{n+1}$ and $\o(z)=(z_1,\cdots,z_n)$, by letting $V$ be a non-conformity score applied coordinatewisely, we recover the standard conformal prediction, see section \ref{cpex} for more details.}, (b) network data, and (c) rotationally invariant data mentioned in point 2 above are can be developed
by implementing \Cref{alg:symm_pi} according to their distributionally invariant structure,
presented in detail in \Cref{sec:ex} and \Cref{sec:ex_re}. 
As an illustration, we will present a detailed discussion and algorithm for example (d) in point 5 below.} 

{On the theory side}, we show that {the prediction set $T^{\mathrm{SymmPI}}(\zo)$ returned by} \OurMethod~has coverage greater than or equal to the nominal level, and not much more than that, under distributional invariance.
We bound the over-coverage in terms of a group-theoretical quantity (the number of orbits of the action of the group).
We further bound the impact 
of distribution shift, i.e., the lack of distributional invariance, 
on the coverage.
Finally, we introduce a \emph{non-symmetric version of \OurMethod} for the 
distribution shift case where the processing algorithm is not distributionally equivariant, and provide associated coverage guarantees.
In the special case of conformal prediction, we recover recent results of \cite{barber2023conformal}.
\begin{algorithm}[h]
\caption{Two-layer hierarchical model example:
When observing data from multiple potentially heterogeneous sub-populations, \OurMethod~can adapt to heterogeneity across sub-populations, unlike conformal prediction (which treats all datapoints symmetrically). As shown experimentally in \Cref{sec:twolayer-detail}, this improves precision. 
}\label{alg:symm_pi_unsup}
\begin{algorithmic}[1]
\Require 
Observed data $(z_i^{(k)})_{(i,k)\in [M]\times [K]\setminus {(M,K)}}$ from $K$ sub-populations/tasks/branches. {Let $\psi(\tilde{z})=\tilde{z}_M^{(K)}.$}
Constant $c \ge 0$; default $c=2$.
Coverage level: $1-\alpha \in [0,1]$. \\ {\Comment{Below, we construct $V$}}
\For{$k \in [K]$} 
\State Compute $\bar{z}_k = \frac{1}{M}\sum_{j=1}^M z_{j}^{(k)}$ \Comment{Within-branch means}
\State Let $\hat\sigma_k^2 = 1$ if $M = 1$, else $\hat\sigma_k^2 = \frac{1}{M-1}\sum_{j=1}^M(z_j^{(k)}-\bar{z}_k )^2$  \Comment{Within-branch variances}
\EndFor
\State Compute $\bar{z} = \frac{1}{K}\sum_{k=1}^K \bar z_{k}$
\For{$k \in [K]$} \Comment{Grand mean across branches}
\State Define $A_{k,c} = \{|\bar{z}_k-\bar{z}|\le c\hat\sigma_k/\sqrt{M}\}$ \Comment{Check closeness of branch means to grand mean}
\State Find $\tilde{z}_{i}^{(k)} = |z_{i}^{(k)}-[\bar{z}I(A_{k,c})+\bar{z}_kI(A_{k,c}^{\complement})]|/\hat\sigma_k$, $i \in [M]$ \Comment{Center at branch/grand mean}
\EndFor
\State \textbf{return} Prediction set for $z_M^{(K)}$:  the values of $z_M^{(K)}$ satisfying $\tilde{z}_M^{(K)}\le Q_{1-\alpha}((\tilde{z}_{i}^{(k)})_{k=1,i=1}^{K,M})$, where $Q_{1-\alpha}$ denotes the $1-\alpha$-th quantile, {and $\tilde{z}=V(z)$}.
\end{algorithmic}
\end{algorithm}
\item { Continuing with Example (d) introduced in \Cref{sec:ex}, 
we study in detail hierarchical two-layer models with several sub-populations \citep{dunn2022distribution,sesia2023conformal,wang2024conformal},
also known as meta-learning with several tasks \citep{park2022pac}. 
We design a 
data processing architecture for this setting.

A simple instance of our method is shown in \Cref{alg:symm_pi_unsup}, 
where we consider scalar-valued data
$Z_{i}^{(k)}, i\in [M]$, $k\in [K]$ 
with a hierarchical structure specified as follows:
There are $K$ distributions $P_k$, $k\in [K]$ and independent samples 
$Z_{i}^{(k)}\sim P_k, i\in [M]$ of size $M$ are generated from each.
We observe all these except $Z_{M}^{(K)}$, and aim to predict the unobserved value.
Further,  
$P_k,k\in [K]$
are drawn independently from a common distribution $\mathcal{P}$.
Unlike conformal prediction---which treats all datapoints symmetrically---\OurMethod~can adaptively leverage all $K$
datasets $(Z_{i}^{(k)})_{i\in [M]}$, for $k=1,\ldots,K$, 
despite their potential heterogeneity.
\OurMethod~achieves this by adaptively centering the observations either at their sub-population mean (given a high heterogeneity) or at the grand mean (given a low heterogeneity).
We show that \OurMethod~performs favorably compared to prior methods, including standard conformal prediction and the algorithm from \cite{dunn2022distribution}.} 
\end{enumerate}

Our paper is structured as follows: In \Cref{sec:prelim}, we introduce preliminaries from group theory used in our work, and notions of distributional equivariance and invariance. In \Cref{sec:related-work}, we provide a detailed review of previous research relevant to our study. In \Cref{sec:SymmPI}, we introduce our novel approach, referred to as \OurMethod, and discuss its underlying theoretical principles and guarantees. Additionally, in \Cref{sec:twolayer}, we illustrate the practical application of \OurMethod~through a two-layer hierarchical model and substantiate its effectiveness through numerical experiments.
A software implementation of the methods used in this paper, along with the code necessary to reproduce our numerical results,
is available at \url{https://github.com/MaxineYu/Codes_SymmPI}.

{\bf Notation.} For a positive integer $m\ge 1$, the $m$-dimensional all-ones vector is denoted as $1_m = (1,1,\ldots,1)^\top \in \R^m$
and the all-zeros vector is 
$0_m = (0,\ldots,0)^\top \in \R^m$. 
We denote $[m]:=\{1,2,\ldots,m\}$, and for $j\in [m]$, the $j$-th standard basis vector by $e_j = (0,\ldots, 1, \ldots, 0)^\top$, where only the $j$-th entry equals unity, and all other entries equal zero. 
For two random objects $X,Y$, we denote by $X=_dY$ that they have the same distribution.
For a probability distribution $\Gamma$ and a random variable $X\sim \Gamma$, we may write the probability that $X$ belongs to a measurable set $A$ as $P(X\in A)$, $P_X(A)$, $P_{\Gamma}(X\in A)$, $P_{X\sim \Gamma}(X\in A)$, $\Gamma(X\in A)$, or $\Gamma(A)$. 
For a cumulative distribution function (c.d.f.)~$F$ on $\R$, and $\alpha\in [0,1]$, the $1-\alpha$-th population quantile is
$q_{1-\alpha}(F) = F^{-1}({1-\alpha}) = \inf\{x : F(x) \ge {1-\alpha}\}$, with $q_{1-\alpha}(F) = \infty$ if the set is empty.
The $1-\alpha$-quantile of the random variable $X$, for $\alpha\in[0,1]$, is $Q_{1-\alpha}(X)$.
For  $c\in \R^k$,
let $\delta_c$ be the point mass at $c$.
For a vector $v\in \R^m$, 
$Q_{1-\alpha}(v)=Q_{1-\alpha}(v_1,\ldots,v_m)$
denotes
$Q_{1-\alpha}(\sum_{i=1}^m \delta_{v_i}/m)$.

\section{Preliminaries}\label{sec:prelim}

We introduce our predictive inference method based on the notions of \emph{distributional equivariance} and \emph{invariance}.

\subsection{Review of Group Theoretic Background}

First we provide a self-contained review of some basic material from group theory that is required in our work.
We refer to e.g., \cite{giri1996group,diaconis1988group,nachbin1976haar,folland2016course,diestel2014joys,tornier2020haar} for additional details.
Readers may skip ahead to 
\Cref{defs-disteq}
 and refer back to this section as needed.

A \emph{group} $\mG$
is a set endowed with a binary operation ``$\cdot$" which\footnote{The sign ``$\cdot$" is dropped for brevity when no confusion arises, and we write 
for
$g,g'\in \mG$, 
$g g' = g\cdot g'$.} 
is associative, in the sense that for all
$g,g',g'' \in \mG$, 
$(g\cdot g') \cdot g'' = g\cdot (g' \cdot g'')$.
Further, a group has an identity element (or unit, or neutral element) denoted as $1_\mG$ or  $e_\mG$, 
such that 
for all
$g \in \mG$, 
$e_\mG g = g e_\mG = g$.
The subscript of the identity is dropped if no confusion can arise.
Finally, each group element $g$ has an inverse $g^{-1}$ such that $g\cdot g^{-1} = 1_\mG$.

A key example is the \emph{symmetric group}
$\S_n$ of permutations of $n\ge 1$ elements
$\S_n=\{\pi:[n]\to [n]\mid \pi\, \mathrm{ permutation}\}$, 
where the multiplication $\pi\cdot\pi'$ corresponds to the composition $\pi\circ\pi'$ of the permutations.
Moreover the identity element $1_{\mG}$ is the identity map with 
$1_{\mG}(x) = x$ for all $x\in[n]$, 
and the group inverse of any permutation $\pi$ is 
its functional inverse $\pi^{-1}$.
Other important groups are the group $\O(n)$ of orthogonal rotations and reflections of $\R^n$
and the special orthogonal group $\SO(n)$ of rotations of $\R^n$.

For a group $\mG$,
the map $\rho: \mG \times \mZ \to \mZ$ 
is an \emph{action} of $\mG$ on 
$\mZ$, 
if 
for all $g,g'\in \mG$ and $z\in \mZ$,
$\rho(gg',z)=\rho(g,\rho(g',z))$; 
and if 
for all $z\in \mZ$,
$\rho(e,z)=z$.
We denote $\rho(g,z):=\rho(g)z$ for both non-linear and linear actions $\rho$. 
This notation takes special meaning when $\rho$ acts linearly, in which case $\rho$ is called a \emph{representation} and we think of $\rho(g):\mZ\to \mZ$ as a linear map.
\begin{example}[Permutation action]\label{perma}
For any space $\mZ_0$,
the symmetric group $\S_n$ acts on $\mZ_0^n$ by the \emph{permutation action} $\rho$, which 
permutes the coordinates of the input, 
such that for all $g\in \S_n$
and $z\in \mZ_0^n$,
$g \cdot x  :=\rho(g) x  := (z_{g^{-1}(1)},\ldots,z_{g^{-1}(n)})^\top$.    
\end{example}

For a general group, the \emph{orbit} of $z$ under $\rho$ is the set $O_z= \{\rho(g)z, g\in \mG\}$, 
and includes the subset of $\mZ$ that can be reached by the action of $\mG$ on $z$.
For instance the orbit of $(1,2)^\top$ under $\S_2$ is $\{(1,2)^\top,(2,1)^\top\}$, while that of $(1,1)^\top$ is $\{(1,1)^\top\}$.

Certain groups are also topological spaces \citep{munkres2019topology}, with associated open sets. In that case,
one can construct the Borel sigma algebra generated by the open sets.
For certain groups---Specifically, for compact Hausdorff topological groups---there is a \emph{``uniform" (Haar) probability measure} $U$ 
over the group endowed with the Borel sigma algebra
\citep[see e.g.,][]{diestel2014joys}. 
For a finite group such as $\S_n$, this is the discrete uniform measure. 
In general, the Haar probability measure satisfies that for any $g\in \mG$, and $G\sim U$, 
we have $gG \sim U$.
We will only consider groups that have a Haar probability measure.

\subsection{Distributional Equivariance and Invariance}
\label{defs-disteq}

We consider a dataset $Z$, 
belonging to a measurable space\footnote{All spaces, sets, and functions will be measurable with respect to appropriate sigma-algebras, which will be kept implicit for simplicity.} $\mZ$, such as an Euclidean space. 
This will also be referred to as the \emph{complete dataset}, 
because we observe only part of $Z$; as explained later in \Cref{sec:predict_region}.
This dataset is completely general:
as special cases, it can 
represent
labeled observations
in a supervised learning setting, i.e., 
$Z = ((X_1,Y_1), \ldots, (X_n,Y_n))^\top$, 
or unlabeled observations in unsupervised learning, i.e.,
$Z = (X_1 \ldots, X_n)^\top$.

The complete data $Z$ has
an unknown distribution $P$ belonging to a set $\mP$ of probability distributions.
Consider a measurable map $V:\mZ\to\tmZ$
for some measurable space $\tmZ$, which we can think of as a transformation of the data.
This transformation can either be designed by hand, or learned in appropriate ways.
For instance, in a supervised learning setting where
$Z = ((X_1,Y_1), \ldots, (X_n,Y_n))^\top \in (\mX\times \mY)^n$, 
and for a predictor $\hmu:\mX\to \mY$ learned based on $Z$, 
we may have $V(Z) = (|Y_1-\hmu(X_1)|, \ldots, |Y_n-\hmu(X_n)|)^\top$.
Further examples and discussion will be provided in our illustrations in \Cref{sec:twolayer}.

We consider a known group $\mG$ that acts on the complete data space $\mZ$ by the action $\rho$ of $\mG$, 
and on the transformed data space $\tmZ$ 
by the action $\tir$ of $\mG$.
When
$Z = ((X_1,Y_1), \ldots, (X_n,Y_n))^\top$, these can be the standard permutation action of the symmetric group $\S_n$ from Example \ref{perma}, such that 
$\rho(g) Z = ((X_{g^{-1}(1)},Y_{g^{-1}(1)}), \ldots, (X_{g^{-1}(n)},Y_{g^{-1}(n)}))^\top$, and
$\rho(g) V(Z) = (|Y_{g^{-1}(1)}-\hmu(X_{g^{-1}(1)})|, \ldots, |Y_{g^{-1}(n)}-\hmu(X_{g^{-1}(n)})|)^\top$.

A key property of the map $V$ 
that we will use to construct prediction regions
is that $V$ respects the symmetry of the group, in a distributional sense.
This is formalized in our definition of \emph{distributional equivariance} given below.
For two random objects $X,Y$, recall that we denote by $X=_dY$ that they have the same distribution.


\begin{definition}[Distributional equivariance]\label{dedef}
We say that the map $V:\mZ\to \tmZ$ is $\mG$-\emph{distributionally equivariant} (with respect to the  actions $\rho,\tir$ on $\mZ,\tmZ$ respectively, 
and over the class $\mP$ of probabilities), 
when for 
all $P\in \mP$, 
for $Z\sim P$ and for  an 
independently drawn group element 
$G\sim U$ from the uniform probability distribution $U$ over $\mG$,
we have the equality in distribution
\beq\label{d}
V(\rG Z)=_d \trG V(Z).
\eeq
\end{definition}

This means that the distribution of a randomly chosen action of $\mG$ on $Z$, transformed by $V$, is equal to 
the distribution found by first  transforming $Z$ by $V$, 
and then randomly acting on it by $\mG$.
In a distributional sense, the random action of $\mG$ and the deterministic transform $V$ ``commute". 
This definition generalizes the classical notion of \emph{deterministic} $\mG$-equivariance, which requires that 
for \emph{all} $z\in\mZ$ and \emph{all} $g\in \mG$,
\beq\label{e}
V(\rg z)=  \trg V(z).
\eeq

Deterministic equivariance is widely studied in
the mathematical area of
representation theory \citep[e.g.,][]{fulton2013representation}.
Distributional equivariance only requires the equality of the distributions of 
$V(\rG Z)$ and $\trG V(Z)$ for \emph{random} $Z,G$, 
whereas 
deterministic equivariance requires \eqref{e} to hold for
for \emph{all} $z\in\mZ$ and \emph{all} $g\in \mG$.
Deterministic equivariance clearly implies the distributional version. 
We characterize these conditions in \Cref{disp}, 
showing 
that distributional equivariance is a strictly more general condition than the deterministic version. {Moreover, 
consider the following example to illustrate the difference between distributionally 
and deterministically equivariant maps.
For some positive integer $M$,
we consider the action of the group of cyclic shifts $z\mapsto z+a$ modulo $M$ on itself.
We show in \Cref{disp} that deterministic equivariance requires affine maps $z\mapsto z+a$  modulo $M$, for any $a$, while distributional equivariance is satisfied by \emph{all maps}}.

A key example of distributional equivariance is \emph{distributional invariance}, namely $V(\rG Z)$ $=_d$ $V(Z)$.
This condition states that after applying the map $V$, the distributions of the original data $Z$  and the data $\rG Z$ acted upon by the group are equal.
This is a special case of Definition \ref{dedef} 
{for an identity output representation $\tir(g)$} for all $g\in \mG$. 
For the identity map $V$,
taking any $g' \in \mG$, we can deduce from it that
$\rgp Z =_d \rgp\rG Z =_d \rG Z =_d Z$. 
Hence for the identity map $V$, distributional equivariance implies that $Z=_d \rgp Z$ for \emph{all} $g' \in \mG$. 
This latter condition has been widely studied; for instance in 
analyzing randomization tests \citep{dobriban2021consistency} and
data augmentation \citep{chen2020group,chatzipantazis2021learning}.
We analyze this condition further in \Cref{dip}.

Consider a fixed $z\in \mZ$ and let $O_z = \{\rg z: g\in \mG\} $ be the orbit of $z$ under $\mG$.
The distribution of $\rG z$ when $G\sim U$
can be viewed as 
a \emph{uniform distribution over the orbit} $O_z$ with the sigma-algebra generated by the intersection of $O_z$
with the sigma-algebra over $\mZ$.
Since this distribution is the same regardless of the distribution $P$ of $Z$, the data $Z$ is conditionally ancillary given the orbits, see also \cite{chen2020group,dobriban2023joint}. 
This shows that distributional invariance as in   
$\rG Z$ $=_d$ $Z$ is a form of conditional ancillarity. 
Conditional ancillarity is one of the most general conditions under which finite-sample valid predictive inference methods have been designed (see \Cref{sec:related-work}).

\subsection{Examples}
\label{sec:ex}
{In this paper, we will mainly discuss four examples. 
Since we will later consider only part of  $Z$ as observed, {to present the core concept without involving excessive technical complexity}, we will here let $Z$ contain $n+1$ observations, 
where the $(n+1)$st will later not be fully observed. 
\begin{itemize}
    \item {\bf Setting (a): Exchangeable data}.
Let $Z=(Z_1,\ldots,Z_{n+1})^\top$ have exchangeable coordinates, 
where each $Z_i$ belongs to some space $\mZ_0$.
This is captured by the permutation group $\S_{n+1}$
permuting the coordinates as 
$\rho(g)Z = (Z_{g^{-1}(1)},$ $\ldots,Z_{g^{-1}(n+1)})^\top$ for all $g\in \S_{n+1}$.
\item {\bf Setting (b): Network-structured data.}
Consider an undirected---possibly weighted---graph with vertex set $[n+1]$ and a symmetric adjacency matrix $A \in [0,\infty)^{(n+1)\times(n+1)}$. 
To each vertex $i\in[n+1]$,
we associate the random variable $Z_i$, and assume that 
\emph{the distribution of the random vector $Z$ is unchanged after relabeling the vertices subject to keeping its structure---as captured by $A$---unchanged}.
This corresponds to invariance under relabeling the vertices by permutations in the graph's automorphism group $\mG = \mathrm{Aut}(A) \subset \S_{n+1}$,
leaving the graph structure unchanged.  
\item {\bf Setting (c): Rotationally invariant data.}
Consider a dataset $Z_1, \ldots, Z_{n+1} \in \R^p$,
for some positive integer $p$,
of observations that are exchangeable and have a jointly rotation-invariant distribution.
 This can occur for rotationally invariant observations that are made in a particular coordinate system, such as for a variety of physical quantities. 
 Many of the fundamental laws of physics can derived from the principle that those laws are independent of coordinate systems, see e.g., \cite{gross1996role,robinson2011symmetry,schwichtenberg2018physics}.
Specifically, we consider $(Z_1^\top, \ldots$, $  Z_{n+1}^\top)^\top $ $=_d (Z_{\pi^{-1}(1)}^\top, \ldots, Z_{\pi^{-1}(n+1)}^\top)^\top$ for any permutation $\pi \in \S_{n+1}$, 
and moreover
 $(Z_1^\top, \ldots$, $Z_{n+1}^\top)^\top $ $=_d (OZ_1^\top, \ldots$, $OZ_{n+1}^\top)^\top$
 for all orthogonal matrices $O \in \O(p)$. 
Thus, the data distribution is invariant under the direct product $\S_{n+1}\times \O(p)$.
\item {\bf Setting (d): Two-layer hierarchical data.}
We consider data
$Z_{i}^{(k)}, i\in [M]$, $k\in [K]$ 
with a hierarchical structure specified as follows:
There are $K$ distributions $P_k$, $k\in [K]$ and independent samples 
$Z_{i}^{(k)}\sim P_k, i\in [M]$ of size $M$  are generated from each.
Further,  
$P_k,k\in [K]$
are drawn {independently} from a common distribution $\mathcal{P}$.
This setting can be useful in meta-learning, sketching, and for clustered data; we will revisit it in \Cref{sec:twolayer}.
The data distribution is invariant under permutations that map each sample $D_k = \{Z_{i}^{(k)}, i\in [M]\}$ into another sample $D_{k'}$, for $k,k'\in[K]$.
\end{itemize}
}

\section{Related Works}\label{sec:related-work}

There is a great deal of related work, and we can only review the most closely related ones.
The idea of prediction sets dates back at least to the pioneering works of \protect\cite{Wilks1941}, \protect\cite{Wald1943}, \protect\cite{scheffe1945non}, and \protect\cite{tukey1947non,tukey1948nonparametric}.
More recently conformal prediction has emerged as a prominent methodology for constructing prediction sets \protect\citep[see, e.g.,][]{saunders1999transduction,vovk1999machine,papadopoulos2002inductive,vovk2005algorithmic,Vovk2013, Chernozhukov2018,dunn2022distribution,lei2013distribution,lei2014distribution,lei2015conformal,lei2018distribution,angelopoulos2021gentle,guan2022prediction,guan2023localized, guan2023conformal,romano2020classification,bates2023testing,einbinder2022training,liang2022integrative,liang2023conformal}. 
Predictive inference methods
\citep[e.g.,][etc]{geisser2017predictive}  
have been developed under various assumptions
\protect\citep[see, e.g.,][]{Sadinle2019,bates2021distribution,Park2020pac,park2021pac,park2022pac,sesia2022conformal,qiu2022distribution,li2022pac,kaur2022idecode,si2023pac}. 
{In particular, recent works go beyond simple coverage, allowing conformal risk control \cite{angelopoulos2022conformal}, by viewing permutations as
acting on exchangeable random functions.} 

There are many works on predictive inference going beyond exchangeability.
Some of these involve invariance under specific permutation groups \citep[e.g.,][etc]{dunn2022distribution,sesia2022conformal,lee2023distribution}, and some are designed to work under various forms of distribution shift \citep{tibshirani2019conformal,park2021pac,park2022pac,qiu2022distribution,si2023pac}.

Online compression models \citep{vovk2005algorithmic} 
are a weaker condition than  exchangeability, 
and enable a generalization of conformal prediction. 
In online compression models, a sequence of observations $\sigma_1,\ldots,\sigma_n,\ldots$ is made of datapoints
$Z_1,\ldots,Z_n,\ldots$, where 
  for some space $\mZ_0$ 
and for all $n$,
$Z_n\in \mZ_0$.
It is assumed that for all $n$,
the conditional distribution of $(\sigma_{n-1},Z_n)$ 
given $\sigma_n$ is known.
A one-off structure 
\citep{vovk2022algorithmic}
is the special case of this in a non-sequential setting, and is closely related to the statistical concept of conditional ancillarity. 
Compared to this, our work can be viewed as the special case of 
one-off 1-structures with a trivial object space $X$, 
when we further have distributional invariance under a group, 
for which the summary statistics are the orbits of the group action.
As we discuss below and in \Cref{dip}, 
distributional invariance has the crucial advantage that there is a broad class of maps---distributionally equivariant ones, including equivariant neural nets---that preserve it; which enables processing the data in a flexible way. 
This does not generally hold under conditional ancillarity or one-off structures. 
Moreover, 
we give bounds on the coverage of our methods under distribution shift, and develop flexible non-symmetric versions of our method.

\cite{Chernozhukov2018} develop predictive inference methods assuming invariance under subgroups of permutation groups.
Compared to this, our work handles the broader class of compact topological groups, which are both technically more challenging, and are of interest in a broader class of applications.
Moreover, we have a more general observation model, focus on the notion of distributional equivariance to enable flexible data processing, and provide methods and guarantees under distribution shift.

Joint coverage regions \citep{dobriban2023joint} are a methodology aiming to unify prediction sets and confidence regions. They have been developed
for general observation models
under general 
conditional ancillarity.
Our focus here differs, as we introduce the notion of distributional equivariance to enable flexible data processing, as well as methods and guarantees applicable to distribution shift.

In a different line of work, invariance and equivariance have been widely studied in other aspects of statistics and
machine learning.
In statistics, this dates back at least to permutation tests \citep{eden1933validity,fisher1935design,pitman1937significance,pesarin2001multivariate,ernst2004permutation,pesarin2010permutation,pesarin2012review,good2006permutation,anderson2001permutation,kennedy1995randomization,hemerik2018exact}.
{In particular, recent work shows how to construct permutation tests using arbitrary distributions of permutations \citep{ramdas2023permutation}.}
Other key early work with general groups includes \cite{lehmann1949theory,hoeffding1952large}.
For more general discussions of invariance in statistics see \cite{eaton1989group,wijsman1990invariant,giri1996group}.
In machine learning, work with invariances
dates back at least to \cite{fukushima1980neocognitron,lecun1989backpropagation} with the development of convolutional neural nets (CNNs), which build translation equivariant layers via convolutions. 
These have been extended to discrete 
and continuous rotation invariance 
\citep{cohen2016group,WeilerHS18}
and to more general Lie groups \citep{pmlr-v119-finzi20a}. 
Alternative approaches include those based on invariant theory \citep{villar2021scalars,villar2022dimensionless,blum2022equivariant} and data augmentation  \citep{lyle2016analysis,chen2020group}.

\section{\OurMethod: Predictive Inference with Group Symmetries}
\label{sec:SymmPI}

\subsection{Constructing Prediction Regions}\label{sec:predict_region}
Here we introduce our \OurMethod~method for predictive inference
when the data has distributional symmetry or invariance.
Our key principle in constructing prediction regions is to leverage the interactions between distributional invariance and equivariance.
Specifically, 
if the full data satisfies the distributional invariance property $V(Z)=_d V(\rG Z)$ when $G\sim U$ and 
if $V$ is distributionally equivariant with respect to $\rho,\tir$ as per Definition \ref{dedef}, 
we have
$$V(Z)=_d V(\rG Z)=_d \trG V(Z).$$
Thus, 
$V(Z)$ is also distributionally invariant, and so
\emph{distributional invariance is preserved by distributionally equivariant maps}.
In the special case of permutation symmetry,
and for the special case of deterministic equivariance,
this simple and key observation has often been used
 in conformal prediction\footnote{For the special case of the symmetric group where $\mG = \S_n$, and for the permutation actions $\rho, \tir$, \cite{dean1990linear} have provided a sufficient condition for a transform $V$ to preserve exchangeability; distributional equivariance is equivalent to their condition in this special case, see Section \ref{deqs}.}
 \citep{vovk2005algorithmic}.
 Here, we aim to vastly extend its reach  
in order to be able to construct prediction sets for data with invariance under
arbitrary compact Hausdorff topological groups; motivated by the examples described above.

A bit more generally, 
\emph{distributional equivariance is preserved by composition}.
Suppose
that for some space $\bar \mZ$ and action $\bar\rho$ on $\bar \mZ$, a map
$h:\tilde \mZ \to \bar \mZ$ is distributionally equivariant with respect to input and output actions $\tir,\bar\rho$. 
This implies that when $G\sim U$, 
and for the random variable $\tilde Z = V(Z)$ over $\tmZ$,
we have
    $h(\trG  \tilde Z)=_d \bar\rho(G)  h(\tilde Z).$
 Hence, we find
 $h(V(\rG Z))=_d h(\trG V(Z)) =_d  \bar\rho(G)  h(V(Z))$,
 and thus $h\circ V$ is $\mG$-distributionally equivariant.
 It follows that we can compose arbitrary $\mG$-distributionally equivariant maps and preserve distributional invariance.\footnote{This is the key reason for which we focus on distributional invariance, as opposed to other forms of conditional ancillarity, to construct prediction sets. 
 For more general conditional ancillarity, this property does not need to hold, and this limits the types of data processing maps $V$ we can use; see \Cref{dip}.}
This property enables us to construct prediction sets based on processing the data in several equivariant steps, for instance
via equivariant neural nets.
We will argue that compositionality helps with expressivity, 
and will leverage this to 
design predictive inference methods that can adapt to heterogeneity, see Section \ref{sec:twolayer}.

Thus, we let $Z$ satisfy distributional invariance, and let $V$ be distributionally equivariant.
We do not observe $z$, but instead 
observe some function $\o(z)$ of $z$, where $\o:\mZ\to\mO$ is an \emph{observation function} for a space $\mO$.
For instance, when $z = (z_1, \ldots,z_n, z_{n+1})^\top$ consists of $n+1$ datapoints, 
in an unsupervised case,
for any $j\in[n]$ we 
can take the observation function 
$\o(z) = (z_1, \ldots,z_{j})^\top$
to be the first $j$ observations.
In a supervised case where  $z = ((x_1,y_1), \ldots, (x_{n+1},y_{n+1}))^\top$,  
we can take
for any $j\in[n]$ 
the observation function 
$\o(z) = ((x_1,y_1), \ldots, (x_{j},y_{j}), x_{j+1}, \ldots, x_{n+1})^\top$
to be the first $j$ labeled observations and the remaining features.
We are interested to predict the unobserved part of $Z$. 
Since the observed part does not necessarily uniquely determine the unobserved part, we aim to predict a \emph{set of possible values}.

We consider a test function $\psi:\tmZ\to \R$, such that we want to include small values of $\psi(V(Z))$ in our prediction set. 
This map generalizes the standard idea of a non-conformity score from conformal prediction \citep{vovk2005algorithmic}. 
For instance, in a supervised learning setting where
$V(Z) = (|Y_1-\hmu(X_1)|, \ldots, |Y_{n+1}-\hmu(X_{n+1})|)^\top$ and $\hmu:\mX\to \mY$ is a predictor, 
we can take 
$\psi(V(Z)) = |Y_{n+1}-\hmu(X_{n+1})|$ aiming to predict unobserved outcomes $Y_{n+1}$ that are close to the values $\hmu(X_{n+1})$ predicted via $\hmu$.
If $\hmu$ is an accurate predictor and $Y_{n+1}$ is tightly centered around $\hmu(X_{n+1})$, this may lead to informative prediction sets.

Given some coverage target $1-\alpha\in[0,1]$,
intuitively, we may want to choose a fixed threshold---or, critical value---$t$ such that
we have the \emph{coverage bound}
$P(\psi(V(Z)) \le t) \ge 1-\alpha$, and then set $\{z:\psi(V(z)) \le t\}$ as our prediction set.
However, it is not generally clear how to find a fixed threshold $t$.
Instead, we use the distributional equivariance of $V(Z)$, which implies that
 for any function $\psi:\tmZ\to \R$,
and any deterministically $\mG$-invariant $t:\tmZ\to \R$, for which $t_{\trg \tz} = t_{\tz}$ for all $g\in \mG$ and $\tz\in \tmZ$,
\begin{equation*}
    P_Z(\psi(V(Z))\le t_{V(Z)}) 
    = P_{G,Z}(\psi(\trG V(Z))\le t_{\trG V(Z)} )
    = P_{G,Z}(\psi(\trG V(Z))\le t_{V(Z)} ).
\end{equation*}
Motivated by this observation, 
for all $\tz \in \tmZ$,
we set $t_{\tz}$ as the $1-\alpha$-quantile of the random variable $\psi(\trG  \tz)$, where $G\sim U $:
\beq\label{t}
t_{\tz}
 =
Q_{1-\alpha}
\left(\psi(\trG  \tz),\, G\sim U \right). 
\eeq
Again, this generalizes the standard approach from conformal prediction, where the quantile is computed for the uniform distribution over the permutation group \citep{gammerman1998learning,vovk2005algorithmic}.
By definition, 
$P_{G}(\psi(\trG  \tz)\le t_{\tz}) \ge 1-\alpha$ holds for any $\tz$. 
\begin{figure}[H]
\centering
\includegraphics[height=1.9in,width=4.6in]{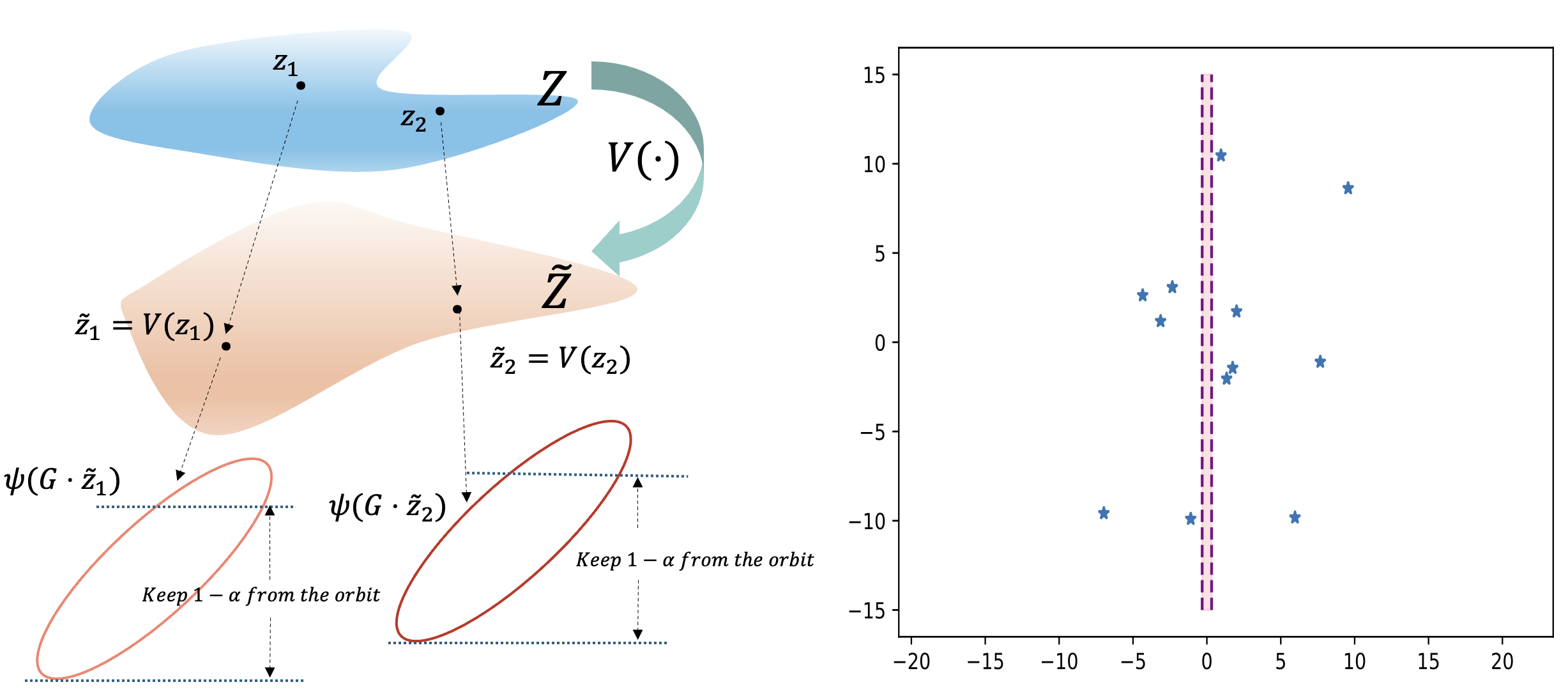}
	\caption{Left: A $1-\alpha$ prediction region extracted from the orbit of $\tilde{z}$, where $\tilde{z}=V(z)$.
 Right: A $95\%$-prediction region $T^{\mathrm\OurMethod}(z_{1:12})$ for $z_{13}$ defined in \eqref{rotation_example}. Here we let $z_i^\top, i\in [13]$ be i.i.d.~random vectors generated from $\N(0,30\cdot I_2)$, where $I_2$ is the $2\times 2$ identity matrix. The blue stars represent 12 observed values, while the white region is the $95\%$-prediction region, as defined in Equation \eqref{rotation_example}. 
 Thus, the pink region depicts the complement of the prediction region, which is  ``safe", in the sense of not having blue stars with 95\% probability. In contrast, the 95\% prediction region constructed via conformal prediction is the whole two-dimensional space. Consequently, using conformal prediction does not yield any ``safe" area; and is thus not informative here. 
 }
	\label{fig:pred_method_ri}
\end{figure}

To take into account the 
observation function $\o:\mZ\to \mO$,
we can simply intersect the prediction region with the 
set of valid observations $\{z: \o(z)= \zo\}$, defining the prediction set
\beq\label{T}
T^{\mathrm\OurMethod}(\zo) = \{z\in \mZ: \psi(V(z))\le t_{V(z)},\, \o(z)= \zo \} .
\eeq
This predictive inference method is applicable when the data has distributional invariance or symmetry, thus we call it \emph{\OurMethod}. 
See \Cref{fig:pred_method_ri} for an illustration, and \Cref{alg:symm_pi} for pseudocode.
This method predicts a set of plausible values for the \emph{full data} $z$. However, we are of course interested in a prediction set for the \emph{unobserved} component of $z$. 
Usually, we can write this unobserved component
as some function $m(z)$ of the data $z$, 
where $m:\mZ\to\mU$ for some space $\mU$; 
and moreover such that 
$z$ is in a one-to-one correspondence with $(\zo,m(z))$.
In that case,  
$T^{\mathrm\OurMethod}$ is equivalent to a prediction set for the unobserved component $m(z)$ of $z$.
For instance, when $z = (z_1, \ldots,z_n, z_{n+1})^\top$ consists of $n+1$ datapoints, 
and the observation function takes values
$\o(z) = (z_1, \ldots,z_{n})^\top$, then we can take $m(z) = z_{n+1}$.

 \subsection{Theoretical Properties}

In this section we study the theoretical properties of our method.
\subsubsection{Coverage Guarantee}
We aim to control the coverage probability $P(Z\in T^{\mathrm\OurMethod}(\Zo ))$, ensuring it is at least $1-\alpha$.
In order to achieve exact coverage $1-\alpha$, 
it is well-known that one may in general need to add a bit of randomization for discrete-valued data \citep{vovk2005algorithmic}.
We now show how this idea can be generalized to our setting.

\begin{definition}[Randomized \OurMethod~Prediction Set]\label{rps}
    For $\tz \in \tmZ$, 
let $F_{\tz}$ be the cumulative distribution function (c.d.f.)~of the random variable $\psi(\trG  \tz),\, G\sim U$, 
and $F'_{\tz}$ be the probability it places on individual points, i.e., for $x\in \R$, 
$F'_{\tz}(x) = F_{\tz}(x)-F^-_{\tz}(x)$, 
where
$F^-_{\tz}(x) = \lim_{y\to x, y<x} F_{\tz}(y) \ge 0$.
Let $\Delta_{\tz} = 0$ if $F_{\tz}'(t_{\tz}) = 0$, and otherwise let $\Delta_{\tz} \in (0,1]$ be 
\beq\label{del}
\Delta_{\tz} = 
\frac{1-\alpha-F_{\tz}^-(t_{\tz})}
{F_{\tz}'(t_{\tz}) }.
\eeq
Consider a random variable $U'\sim \mathrm{Unif}[0,1]$ independent of $Z$, and the randomized \OurMethod~prediction set
\begin{align}
\label{Tr}
T_r(\zo) =
\bigl(\{z: \psi(V(z)) < t_{V(z)}\} 
\cup \{z: 
\psi(V(z))=t_{V(z)},\,
U'  < \Delta_{V(z)}\}
\bigr)
\cap\{z:  \o(z)= \zo\}.
\end{align}
\end{definition}
Clearly, $T_r(\zo) \subset T^{\mathrm\OurMethod}(\zo)$.
Our first result, 
proved in \Cref{pfc1},
shows 
that
the randomized prediction set $T_r$ has coverage \emph{exactly} $1-\alpha$, 
and the deterministic prediction set
$T^{\mathrm\OurMethod}$ has at least $1-\alpha$ and at most a bit higher, depending on the\footnote{Our ``jump" function is superficially connected to the jump function used in conformal risk control \citep{angelopoulos2022conformal}. Indeed, our function measures the discontinuities in the cumulative distribution function of $\psi(\trG  \tz),\, G\sim U$ over the orbit, while their jump function measures the discontinuity of their deterministic loss function, without regard to the randomness.}
``jumps" $F'_{\tz}$
in the distribution of $\psi(\trG  \tz),\, G\sim U$; generalizing results from conformal prediction \citep{vovk2005algorithmic,lei2013distribution}.

\begin{theorem}[Coverage guarantee under distributional invariance]
\label{c1} 
For some group $\mG$ with a uniform probability measure $U$,
let the full data $Z\in\mZ$ 
satisfy the distributional invariance property $Z=_d \rG Z$ when $G\sim U$, for some action $\rho$ of the group $\mG$ on $\mZ$. 
Consider $\alpha\in[0,1]$,
a space $\tmZ$, 
a $\mG$-distributionally equivariant function $V: \mZ\to \tmZ$ as per Definition \ref{dedef},  
and a map $\psi: \tilde \mZ \to \R$.
Let the observed data be $\Zo  = \o(Z)$, 
for an observation function $\o:\mZ\to \mO$ and some space $\mO$.
Then the 
\textnormal{\OurMethod} prediction region from \eqref{T}, 
and the randomized prediction region from \eqref{Tr}
have valid coverage, lower bounded by $1-\alpha$, and also---with  $F'$ from Definition \ref{rps}---upper bounded as
\begin{align}\label{lb}
1-\alpha
=
P(Z\in T_r(\Zo ))
\le 
P(Z\in T^{\mathrm\OurMethod}(\Zo ))
\le
1-\alpha+\E[F'_{V(Z)}(t_{V(Z)})].
\end{align}
\end{theorem}

There are various conditions under which we can upper bound the slack $\E[F'_{V(Z)}(t_{V(Z)})]$ in the coverage error.
For instance, if 
$F'_{\tz}(x)  \le \tau$ for all $x\in\R$ and $\tz \in \tmZ$, 
then the coverage is at most $1-\alpha+\tau$.
To be more concrete,
consider the set $\mH_{\tz}=\{g \in \mG: \psi(\trg \tz)=\psi(\tz)\}$ 
consisting of the group elements that fix $\psi(\tz)$ under the action $\tir$.
As we show, the size of this set controls the jumps in $F_{\tz}$,
under the algebraic condition that $\mH_{\tz}$ is a \emph{subgroup} of $\mG$.
Recall that a set $\mH\subset \mG$ is a subgroup of $\mG$, if $\mH$ is also a group; this is denoted by $\mH \leq \mG$.

In particular, 
we will see in examples that often there is a set $\Omega \subset \tmZ$
such that 
$P(V(Z)\in \Omega) = 1$ and such that
for
every $\tz',\tz''\in \Omega$,
$\mH_{\tz'}=\mH_{\tz''}$. 
Then 
for $\mH =\{g: \psi(\trg \tz)=\psi(\tz),\, \text{ for all } \tz\in \Omega\}$, we have 
$\mH_{\tz'}=\mH$.
It readily follows that $\mH_{\tz'}=\mH$ is a subgroup of $\mG$. 
Recalling that for a finite set $A$, we let $|A|$ be the number of elements---cardinality---of $A$, we have the following result:

\begin{proposition}[Coverage upper bound]\label{ub}
If $\mG$ is finite, and 
if for all $g\in\mG$ the set   $\mH_{\tz}=\{g: \psi(\trg \tz)=\psi(\tz)\}$
is a subgroup of $\mG$, then 
$    P(Z\in T^{\mathrm\OurMethod}(\Zo ))\le 1-\alpha+\E|\mH_{V(Z)}|/|\mG|.$
In particular, if there is a set $\Omega \subset \tmZ$
such that 
$P(V(Z)\in \Omega) = 1$ and such that
for
every $\tz',\tz''\in \Omega$,
$\mH_{\tz'}=\mH_{\tz''}$,
then 
for $\mH =\{g: \psi(\trg \tz)=\psi(\tz),\, \text{ for all } \tz\in \Omega\}$, we have 
$$    P(Z\in T^{\mathrm\OurMethod}(\Zo ))\le 1-\alpha+|\mH|/|\mG|.$$
\end{proposition}
See \Cref{pub} for the proof.
As we will see below,
in many applications of interest,
$\mH_{\tz}$ are indeed subgroups of $\mG$ for all $\tz$, and often $\mH_{\tz}$ 
does not depend on $\tz$.
In particular, in this case
$|\mG|/|\mH|$ is the number of \emph{cosets} of the subgroup $\mH$ in $\mG$. Thus the above general result gives a group-theoretic characterization of the slack in the coverage error.

We also mention that a split, or split-data,  version of \OurMethod---inspired by inductive or split conformal prediction \citep{papadopoulos2002inductive}---is a special case of our method. 
Let the full data be given by $Z = (Z_{\mathrm{tr}}, Z')^\top$, where $Z_{\mathrm{tr}}$ consists of training data, and
$Z'$ consists of calibration and test data.
Suppose that for some group $\mG_0$, 
the distribution of $Z'$ conditional on $Z_{\mathrm{tr}}$ is $\mG_0$-invariant. Then, we can use the methodology described above, applied to $Z'$ and $\mG_0$ instead of $Z$ and $\mG$. 
This procedure has valid coverage even when 
the  $\mG_0$-distributionally equivariant
map $V$ is not fixed as above, 
but is 
learned using $Z_{\mathrm{tr}}$. 
The key advantage compared to full \OurMethod\,  is that we can fit $V$ once on $Z_{\mathrm{tr}}$, and then use it as a fixed predictor on $Z'$; which can improve computational efficiency compared to \OurMethod.
For $V$ to be a useful predictor, it is beneficial if $Z_{\mathrm{tr}}$ and $Z'$ have a similar structure. 
For instance, this could be the case if  $Z_{\mathrm{tr}}$ also satisfies $\mathcal{G}_0$-equivariance.

Finally, we mention that,
while our results only require $V$ to be distributionally equivariant, in practice there are often more known examples of deterministically  equivariant functions, and so we will typically still take $V$ to be deterministically equivariant. 
However, we believe that our theoretical contribution of introducing distributional equivariance is fundamental, because it reflects a broad condition under which distributional invariance is preserved under compositions. Thus it is a crucial notion for predictive inference methods based on symmetry.

\subsubsection{Extension: Distribution Shift}
\label{eds}

Next, we present an extension of our coverage result for the case of distribution shift. 
To present this result, we need to recall some additional notions.
For the subgroup $\mH \leq \mG$
the set 
$g \mH  = \{gh: h\in \mH\}$ is called a (left) \emph{coset} of $\mH$ in $\mG$.
The set of cosets is denoted as $\mG/\mH:=\{g \mH : g\in \mG\}$.
Then $\mG$
is partitioned into cosets $g\mH$, and
we obtain a set
$S$ 
of \emph{representatives} of cosets in $\mG/\mH$ by choosing an element of each coset.

First, we allow for a distribution shift away from distributional invariance, i.e., $Z\neq _d \rho(G)Z$. 
We provide a general coverage bound for this scenario.
For $g \in \mG$, define the map $\ell_{g}:\tmZ\to\R$ such that for all $\tz \in \tmZ$,
$\ell_{g}(\tz)=\psi(\trg \tz)$.
 Let $\mF=\{\ell_{g},\, g\in \mG\}$ be the set of all maps $\ell_{g}$, $g\in \mG$. 
Now,
$\mH =\{g: \ell_g =\ell_e\}$
is clearly a subgroup of $\mG$.
Hence, 
$\mG$
is partitioned into cosets $g\mH$, 
corresponding to distinct values of $\ell_g$.
Let $U(\mG/\mH)$ be the invariant probability measure over the cosets \cite[][Corollary 4, p. 140]{nachbin1976haar}; 
and identify $\mG/\mH$ with a measurably chosen set $S$ 
of representatives.
Let $G'\sim U(\mG/\mH)$, identified with a random variable over $S$.
For instance, for a finite $\mG$, 
we can identify 
$U(\mG/\mH)$ with the uniform distribution over
a set of representatives 
$\{g_1, \ldots, g_{|\mF|}\}$
of the cosets $\mG/\mH$. 
For any $\tz\in \tmZ$, 
define
$\nu(\tz) = \psi(\tz)-t_{\tz}$
and, with $\textrm{TV}$ denoting total variation distance,
\beq\label{D}
\Delta = 
\E_{G'\sim U(\mG/\mH)}
    \textrm{TV}_Z\Big(\nu(V(Z)),\nu(\tilde{\rho}(G')V(Z))\Big).
\eeq
See \Cref{pfthm_shift} for the proof of the following result.

\begin{theorem}[Coverage guarantee under distribution shift]\label{thm_shift}
Under the conditions of Theorem \ref{c1}, 
even if $Z$ 
does not necessarily satisfy distributional invariance, we have  with $\Delta$ from \eqref{D} that 
\begin{align}\label{ds}
-\Delta\le P\Big(Z\in T^{\mathrm\OurMethod}(\Zo )\Big)- (1-\alpha) \le \Delta+  \E[F'_{V(Z)}(t_{V(Z)})].
\end{align}
\end{theorem}

Theorem \ref{thm_shift} establishes that, 
even if $Z \neq_d \rho(G) Z$, 
we can derive coverage bounds similar to those found in Theorem \ref{c1}, up to a margin of $\Delta$ as given in equation \eqref{D}. This gap reduces to zero when the distributional invariance property holds (i.e., $Z=_d\rho(G)Z$), in which case $\Delta = 0$.

The aforementioned result relies on
symmetry in two ways. 
First, the quantile $t_{\tilde{z}}$ in equation \eqref{t} is chosen for the uniform probability distribution over the group; or equivalently over representatives of cosets.
Second, the function $V$ is required to be distributionally equivariant with respect to $\mathcal{G}$. 
In \Cref{non-symmetric_algorithm}, 
we introduce a novel algorithm
for scenarios where the aforementioned two symmetry properties do not hold. 
We also provide theoretical coverage guarantees for this algorithm.
Notably, our framework recovers the results presented in Theorems 2 and 3 of \cite{barber2023conformal} as a special case, in the context of studying conformal prediction with the group $\mathcal{G}=S_{n+1}$ and the function
$\psi$ with
$\psi(z)=z_{n+1}$ for all $z$. For more details, we refer to \Cref{non-symmetric_algorithm}.

\subsection{Computational Considerations}

In this section, we discuss computational considerations for our prediction regions. 
Given $V$ and $\psi$, we need to compute the set in \eqref{T}, i.e., 
$\obs^{-1}(\zo) \cap \{z: \psi(V(z))\le t_{V(z)}\}$, where $\obs^{-1}(\zo)$ denotes the preimage of $\zo$ under the map $\obs$,
 for a given $\zo$.
Often, the preimage of the observation map can be characterized in a convenient way; in many cases the observation map selects a subset of the coordinates of $z$, and so its preimage includes the set of observations with all possible values of the missing coordinates.

Moreover, we can write the second set of the above intersection as a preimage under $V$
in the form
$V^{-1} (B)$, where $B =  \{\tz \in\tmZ: \psi(\tz) \le t_{\tz}\}$.
A key step to compute $B$ is to compute the quantiles $t_{\tz}$
over the randomness of $G\sim U$, or---following notation from Section \ref{eds}---$G'\sim U(\mG/\mH)$.
When the number of equivalence classes $\mG/\mH$ is not large, we can calculate the quantile by enumeration. However, if the number of equivalence classes is large, a practical approach is to sample from each equivalence class to approximate the quantile 
Specifically, we can define
\beq\label{t2}
\tilde{t}_{\tz}
 :=
Q_{1-\alpha}
\left(\psi(\tz),\psi(\tilde{\rho}(G_1) \tz),\psi(\tilde{\rho}(G_2)  \tz),\cdots,\psi(\tilde{\rho}(G_M)  \tz)\right), 
\eeq
where $G_1,\cdots,G_M$ are sampled independently from $U(\mG/\mH)$.
  We then define the prediction set
  \begin{align*}
    \tilde{T}^{\mathrm\OurMethod}(\zo) := \{z\in \mZ: \psi(V(z))\le \tilde{t}_{V(z)},\, \o(z)= \zo \}.
    \end{align*}
This has the following coverage guarantee, with a proof deferred to \Cref{pfsample_g}.
\begin{proposition}\label{sample_g}
Under the conditions of Theorem \ref{c1}, we have 
$P(Z\in \tilde{T}^{\mathrm\OurMethod}(\Zo))\ge 1-\alpha.$
\end{proposition}

For the remaining problems, i.e., computing $A$ and finding its preimage under $V$, at the most general level, various approaches can be employed to generate suitable approximate solutions. 
At the highest level of generality, similar computational problems arise in standard conformal prediction, and no general computational approaches are known.
Our setting is similar, and thus computational approaches must be designed on a case by case basis. 
In many cases of interest, including all those presented in the paper, the computational problem simplifies and can be solved conveniently.
One approximate method involves systematically examining a grid of candidate values of $\tz$, and retaining those for which $ \psi(\tz) \le t_{\tz}$.
Further, if we can choose $V$ to be an invertible function whose inverse is convenient to compute, then the preimage under $V$ can be convenient to find.
Otherwise, in general, one may need to search over a grid on $\mZ$ (instead of $\tmZ$) to approximate $V^{-1}(A)$.

\subsection{Examples Revisited}\label{sec:ex_re}
We now revisit the examples from \Cref{sec:ex} in more detail.
\subsubsection{Setting (a): Exchangeable data} 
\label{cpex}
\begin{itemize}
    \item {\bf Data, group, and action.}
    Take
$\mZ = \mZ_0^{n+1}$, for some space $\mZ_0$,
and the group $\mG$ as the permutation group $\S_{n+1}$, acting 
on $\mZ$
by permuting the coordinates as 
$\rho(g)Z :=gZ =$ $(Z_{g^{-1}(1)},$ $\ldots,Z_{g^{-1}(n+1)})^\top$ for $Z=(Z_1,\ldots,Z_{n+1})^\top$.
Then {for any function $V:\mZ\to \mZ$ that applied coordinatewise,}
with for all $z\in\mZ$,
the distributional invariance condition reduces to the vector $Z$ having \emph{exchangeable components}.
\item {\bf Transformations.} If we choose $V$ more generally,
we can recover results 
from
conformal prediction \citep{gammerman1998learning,vovk2005algorithmic}.
Let 
$\tmZ = \R^{n+1}$,
and let $\S_{n+1}$ also act
on $\tmZ$
by the permutation action $\tir$.
In an unsupervised case,
we can take the observation function 
$\o(z) = (z_1, \ldots,z_{n})^\top$, for all $z$.
In a supervised case where  $z = ((x_1,y_1), \ldots, (x_{n+1},y_{n+1}))^\top$,  
we can take
the observation function 
$\o(z) = ((x_1,y_1), \ldots, (x_{n},y_{n}), x_{n+1})^\top$, for all $z$.\\
Next, we set 
$V:\mZ\to\tmZ$ as a permutation-equivariant map with respect to the permutation actions.
Here 
$s(z_j):=s(z_j;z):= [V(z)]_j$
are referred to as the \emph{non-conformity scores}. 
Considering the supervised case for concreteness,
we can take
$\psi(\tz)=\tz_{n+1}$ to be the last coordinate.
Since 
$z$ decomposes as 
$z = (\o(z),y_{n+1})$,
a prediction set for $z$ is equivalent to a prediction set for $y_{n+1}$. 
\item {\bf Properties of prediction sets.} Clearly, $T^{\mathrm\OurMethod}$ from \eqref{T} 
reduces to $z$ such that 
$s(z_{n+1})\le Q_{\beta'}(s(z_1),\ldots,s(z_{n}))$, where $\beta' = \lceil (n+1)(1-\alpha) \rceil/n$. 
This is identical 
to a standard conformal prediction set with non-conformity score $s$.
If $Z$ has exchangeable coordinates,
Theorem \ref{c1} recovers the classical conformal coverage lower bound  
$P(s(Z_{n+1})\le Q_{1-\alpha}(s(Z_1),\ldots,s(Z_{n+1}))) \ge 1-\alpha$
from \cite{gammerman1998learning,vovk2005algorithmic}.
Then, note that 
for $g \in \S_{n+1}$,
$(\trg \tz)_{n+1} = \tz_{g^{-1}(n+1)}$.
Hence, 
$\mH_{\tz}=\{g: \tz_{g^{-1}(n+1)} =\tz_{n+1}\}$.
If
all coordinates of $\tZ$ are distinct---which holds with probability one if
$V$ is injective and
$(Z_1,\ldots, Z_{n+1})^\top$ has a continuous distribution---then
$\mH_{\tz}=\mH = \{g: g(n+1) =n+1\}$ is the stabilizer of the $(n+1)$st element, the subgroup of permutations fixing the last coordinate.
In this case, 
$|\mH|/|\mG| = n!/(n+1)! = 1/(n+1)$,
and \Cref{ub} recovers the classical result
that the over-coverage of conformal prediction is at most $1/(n+1)$ \citep{lei2013distribution}.

\end{itemize}

\subsubsection{Setting (b): Network-structured data} 
\label{g}

\begin{itemize}
    \item {\bf Data, group, and action.} We discuss an illustration of our
methods for random variables whose symmetries are captured by a graph. 
We have $Z=(Z_1,\ldots,Z_{n+1})^\top \in  \mZ_0^{n+1}=:\mZ$ as in \Cref{cpex}.
However, now we have an undirected graph with 
adjacency matrix $A \in [0,\infty)^{(n+1)\times(n+1)}$, 
and $\mG$ is the automorphism group $\mG = \mathrm{Aut}(A) \subset \S_{n+1}$
acting on $\mZ$ by the action $\rho$ that permutes the coordinates of $Z$.
The elements of $\mG$ are
permutations $g$ such that, when viewed as linear maps $\R^{n+1}\to \R^{n+1}$, we have
$g A g^\top = A$.
This setting generalizes exchangeability, which can be recovered by taking---for instance---the identity matrix $A = I_{n+1}$.
\begin{figure}[H]
\centering
\includegraphics[height=1.8in,width=4.2in]{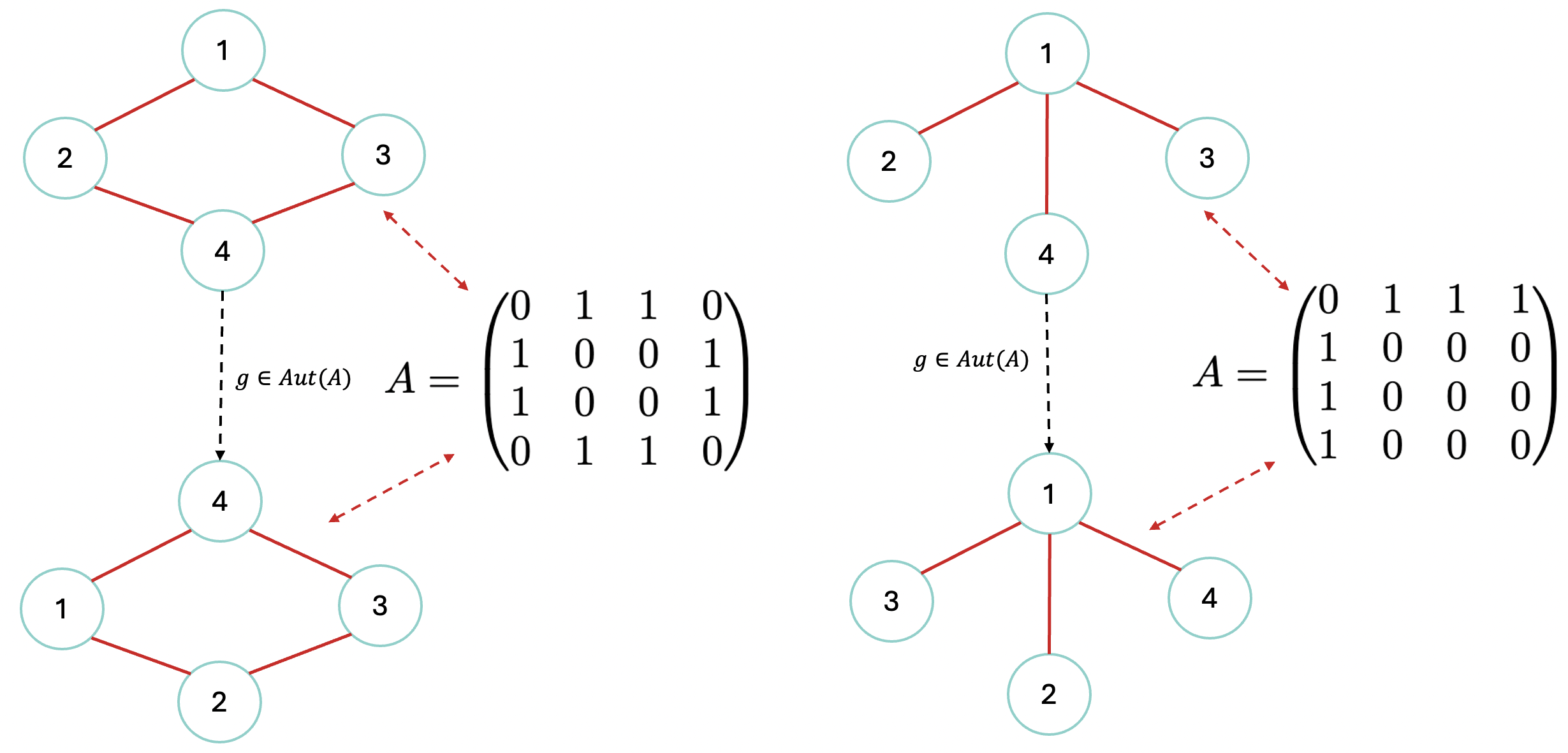}
	\caption{Examples of undirected graphs along with their adjacency matrices. Relabeling the vertices of a graph according to its automorphism group leaves the adjacency matrix unchanged.}
	\label{fig:pred_method_ri2}
\end{figure}
\item {\bf Transformations.} For some space $\tmZ$,
consider $V:\mZ\to\tmZ$, such that for some action $\tilde\rho$ of $\mG$, $V$ is distributionally equivariant. Then, based on Proposition \ref{diem} in the Appendix, 
for all $\tz$ in the image of $V$, 
we must have
the following equality of the sizes of sets  for all $g $:
$|V^{-1}(\tz)|=|V^{-1}(\trg \tz)|$; where $V^{-1}(c)$ denotes the preimage of the element $c\in \tmZ$ under $V$. 
This condition states that the number of elements mapping to any specific $\tz$ is the same as the number mapping to any other element $\trg \tz$ in the orbit of $\tz$. 

In contrast, deterministic equivariance requires that for all $z$ and $g \in \mathrm{Aut}(A)$, $V(g z) = \trg V(z)$. 
In machine learning, many graph neural net (GNN) architectures satisfying deterministic equivariance have been developed. 
A prominent example are message-passing graph neural networks (MPGNNs), see e.g., \cite{gilmer2017neural,xu2018representation}.
Here, 
for some depth $L$,
we define layers $z^0 := z$, and $z^1, \ldots, z^L $ sequentially.
For any $\ell\ge0$ and any $i\in[n+1]$,
the $i$-th coordinate
of $z^{\ell+1}$ is defined
by summing 
the values of a function $\lambda_1$
over the neighborhood  $N(i)$ of node $i$ in the adjacency matrix of the initial graph, and applying another function $\lambda_0$ as:
\begin{align}\label{u}
    z_{i}^{\ell+1}=\lambda_0\Bigg(z_i^{\ell},\sum_{j\in N(i)}\lambda_1(z_i^\ell,z_j^\ell)\Bigg):=F_\ell(z^{\ell})_i,
\end{align}
 The message passing neural network is
$\textrm{MPGNN}_L(z):=F_{L}\circ F_{L-1}\circ \ldots\circ F_1(z)$ for all $z\in \mZ$.
It is well known that any MPGNN is deterministically $\mG$-equivariant for $\tir=\rho$ being the permutation action, namely
$\textrm{MPGNN}_L(g\cdot z)= g\cdot \textrm{MPGNN}_L(z)$;
and hence also distributionally $\mG$-equivariant.
\item {\bf Properties of prediction sets.}
If it is feasible to enumerate the automorphism group, then we can use the prediction set from \eqref{T}. However, in general determining the automorphism of a graph is a hard problem; see e.g., \cite{babai2016graph} for a closely related problem.
We discuss a coarsening approach to partially overcome this in Section \ref{coarse}.

Consider now a simplified setting, where 
$\tmZ = \mZ$,
$V:\mZ\to \mZ$ is the identity transform, 
and  $\rho,\tir$ are the permutation actions. 
We can take the test function
$\psi$ such that 
$\psi(z)= z_{n+1}$ for all $z$,
in which case we are predicting $z_{n+1}$ in a graph. 
Let $\mH_{n+1}$  be the stabilizer subgroup of $n+1$ in $\mG$, i.e., $\mH_{n+1} = \{g\in \mG: g(n+1)=n+1\}$. 
Let $b = |\mG/\mH_{n+1}|$ be the number of equivalence classes of the quotient  $\mG/\mH_{n+1}$ and 
take a collection $g_1, \ldots, g_b \in \mG $ of representatives of the equivalence classes. 
By the orbit-stabilizer theorem \citep[][Proposition 6.8.4]{artin2018algebra}, $B = \{g_1(n+1), \ldots, g_b(n+1)\}\subset[n+1]$ is the orbit of $n+1$ under $\mG$ acting on $[n+1]$.
Then, 
for $\tz\in\tmZ$,
the quantile $t_{\tilde{z}}$ is presented as 
$Q_{1-\alpha}(\{\tilde{z}_j,$ $j\in B\})$. As a special case, when $\mathcal{G}=\S_{n+1}$, the orbit of $n+1$ is all of $[n+1]$.
Then the quantile reduces to the one used in standard conformal prediction.

There are many other choices of the test function $\psi$. 
Suppose for simplicity that $\mZ_0 = \R$.
Then, we can take 
for example
$\psi(x) = (e_{n+1} - 1_{B}/b)^\top x$ for all $x$, which 
measures the difference between the unknown value and the average of its orbit.

\begin{figure}[ht]
	\centering
	\includegraphics[width=1\textwidth]{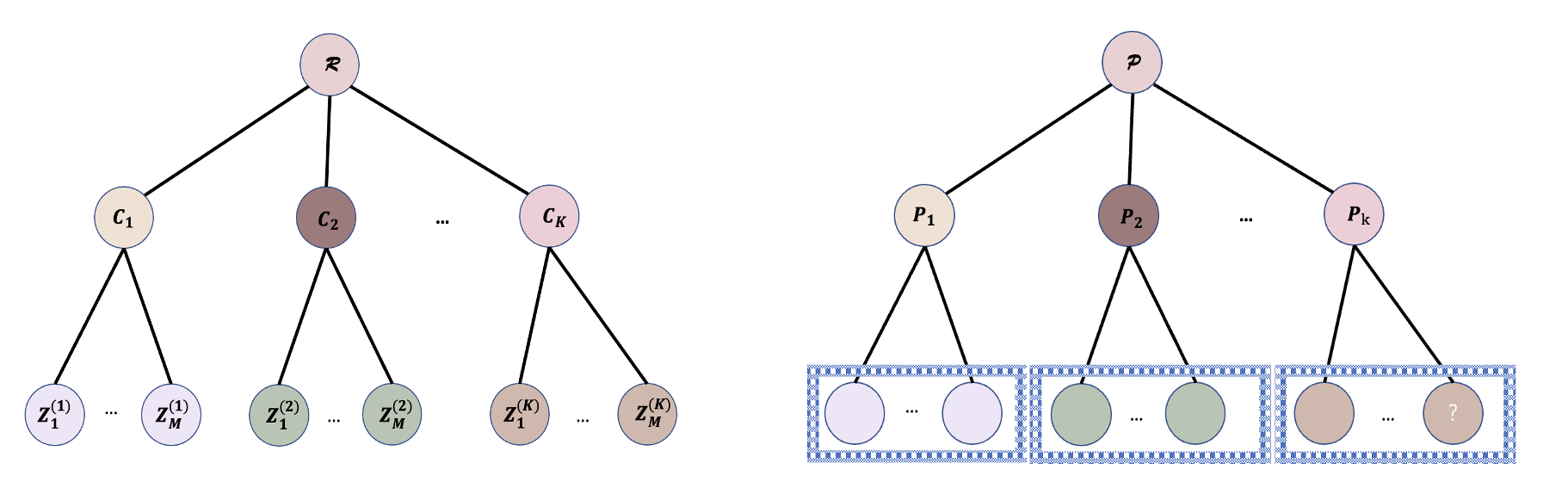}
	 \caption{Left: Example of a depth-two tree-structured graphical model.
       Right: Example of a two-layer hierarchical model.}
           \label{fig:tree2}
\end{figure}
\item  {\bf Tree-structured graphical model example.}
Next, we use a simple tree-structured graphical model example to showcase the above properties; see \Cref{fig:tree2} (left).
The \emph{rooted tree} $\Gamma$ has a \emph{root} with an associated random variable $R$. 
The root has $K\ge 1$ children, which are associated with random variables $C_1\ldots,C_K$ at its first layer; 
these define $K$ \emph{branches}. 
Each of the nodes $C_k$, $k\in[K]$ in the first layer
 has $M\ge 1$ children 
 with associated random variables
 $Z_{1}^{(k)},\ldots,Z_{M}^{(k)}$ in the second layer.
We assume that in the associated graph describing the symmetry, 
each node is connected precisely to its children.
Then, we assume that 
the joint distribution of
the random vector $T=(R,C_1\ldots,C_K,Z_{1}^{(1)},\ldots,Z_{M}^{(K)})^\top \in \mathcal{Z} $ 
associated with 
this depth-two tree
satisfies $g\cdot \Gamma=_d \Gamma$,
where $g$ is any element in the automorphism group $\mathcal{G}$, represented as a subgroup of $\S_{n+1}$.

We can consider the setting 
when $Z_{M}^{(K)}$, the last node of the last branch, is unobserved, and suppose for simplicity that $\mZ_0 = \R$.
We can let $\psi(z)=|e_{1+K+K\cdot M}^\top z|$, with $z\in \mathcal{Z}$. 
The orbit of $e_{1+K+K\cdot M}^\top$ 
is $\{ e_{1+K+i}^\top,i\in [K\cdot M]\}$. 
Therefore, 
the quantiles 
from \eqref{t}
are $t^{(1)}_{z}:=Q_{1-\alpha}( |z_{1}^{(1)}|,\ldots,|z_{M}^{(K)}|)$,
and the prediction set with $1-\alpha$ probability coverage reduces to 
\beqs
T^{\mathrm\OurMethod}(\zo) = \{z: |z_{M}^{(K)}|\le t^{(1)}_{z},\, \o(z)= \zo \} .
\eeqs
We can also aim to predict at a cluster after coarsening the graph, and we describe this in Section \ref{coarse}.
\end{itemize}

\subsubsection{Setting (c): Rotationally invariant data} 
\label{coex}
\begin{itemize}
    \item {\bf Data, group, and action.} 
Consider rotationally invariant data, 
in the same setting as in \Cref{cpex}, 
except letting
$\mZ_0 = \R^p$. 
The group is the direct product
$\mG = \S_{n+1} \times \O(p)$, 
acting on $z= (z_1^\top,\ldots, z_{n+1}^\top)^\top \in (\R^p)^{n+1} $ via the action 
$\rho(\pi,O)z = ((Oz_{\pi^{-1}(1)})^\top,\ldots, (Oz_{\pi^{-1}(n+1)})^\top)^\top$.
For instance, for $p=2$ or $p=3$, we may consider observations of celestial objects (e.g., coordinates of asteroids). 
The system of coordinates used to represent the data can be centered at the Earth, but the rotation of the system is arbitrary.
\item {\bf Transformations.} 
For simplicity, we will only consider $\tmZ = \mZ$
and the identity map $V$.
\item {\bf Properties of prediction sets.}
If we are interested in predicting the position of the $(n+1)$st object based on the positions of the first $n$, 
leveraging the inherent rotational invariance may increase precision. 
For the test function $\psi$,
we may aim to predict 
how close a new object could come to the
the trajectory of another celestial body of interest;  say, the path of a rocket.
For instance, 
consider the locus $z_{n+1,2}=0$ 
and suppose we aim to predict a region 
$|z_{n+1,1}|\ge C$ that contains the next observation at least $1-\alpha$ of the time.
Then we can take
$\psi(z)=-|z_{n+1,1}|$, 
and the prediction set for the $(n+1)$st observation becomes
\beq\label{rotation_example}
T^{\mathrm\OurMethod}(z_{1:n}) = \{z_{n+1}: -|z_{n+1,1}|\le Q_{1-\alpha} \left(-|(Oz_{\pi^{-1}(n+1)})_1|,\, O\sim U(\O(p)), \pi\sim U(\S_{n+1})  \right)\}.
\eeq
This can be further simplified since
for any $v$,
$(Ov)_1 = _d W^\top v/\|W\|_2$, where $W\sim \N(0,I_p)$ and $\|W\|_2$ is the Euclidean norm of $W$. 

For illustration, we present a two-dimensional toy example in Figure \ref{fig:pred_method_ri} (right).
In this case, the prediction set can be interpreted as  ``we are 95\% sure that a new observation
$Z_{n+1}$
will have $|Z_{n+1,1}|$ at least this large".
Such a region may be useful to determine the allowable range of motion of a rocket moving along the vertical axis---there is 95\% chance that the next celestial body is outside of the vertical strip.

\item {\bf Supervised case.}
Consider rotationally invariant data in the supervised case, where a datapoint satisfies $(X,Y)=_d (OX,Y)$. 
We can form the prediction set of those $y_{n+1}$ for which 
\begin{align} \label{rot_prediction}
|y_{n+1}-\hmu(x_{n+1})|\le Q_{1-\alpha} \left(|y_{\pi^{-1}(n+1)}-\hmu(Ox_{\pi^{-1}(n+1)})|,\, O\sim U(\O(p)), \pi\sim U(\S_{n+1})  \right).
\end{align}
For classification, given probabilistic predictor 
$\hat{p}:(y,x)\mapsto\hat{p}(y,x) \in [0,1]$,
we can use the score $-\hat p$, forming the set of $y_{n+1}$ such that 
$-\hat p(y_{n+1},x_{n+1})\le Q_{1-\alpha} (-\hat p(y_{\pi^{-1}(n+1)},Ox_{\pi^{-1}(n+1)}))$,
where 
$O\sim U(\O(p)), \pi\sim U(\S_{n+1})$.

\end{itemize}

\subsubsection{Setting (d): Two-layer Hierarchical Model}\label{sec:twolayer}
\begin{itemize}
    \item {\bf Data, group, and action.} 
We next study data with a two-layer hierarchical structure. 
Such a model can be useful in many applications, such as meta-learning \citep{fisch2021few,park2022pac}, sketching \citep{sesia2022conformal}, 
and clustered data \citep{dunn2022distribution,lee2023distribution}. 

For the first-layer nodes,
we draw 
distributions $P_k\sim \mathcal{P},k\in [K]$
independently from a distribution $\mathcal{P}$, which is associated with the zeroth layer.
These can be viewed as specifying distinct sub-populations from which data is collected.
From the perspective of meta-learning, they can be viewed as distinct but related tasks (e.g., prediction in various environments).
Different from the tree-structured graphical model example from \Cref{g}, 
here the zeroth- and first-layer nodes are
not observed.

The second-layer nodes (or leaves for simplicity) in the $k$-th branch are random variables 
$Z_{i}^{(k)}, i\in [M]$
drawn exchangeably 
from the distribution $P_k,k\in [K]$ \citep{dunn2022distribution,park2022pac,lee2023distribution}. An illustration is presented in the right panel of  \Cref{fig:tree2}. 
Our goal is to construct prediction sets 
for both unsupervised and supervised settings, given as follows: 

\begin{example}[Unsupervised Learning]\label{ul} We let $Z_{i}^{(k)}\in \R$, $i\in [M]$, $k\in [K]$.
\end{example}

\begin{example}[Supervised Learning]\label{sl} 
For some space $\mathcal{Z}_0$,
we let $Z_{i}^{(k)}=(X_i^{(k)},Y_i^{(k)}) \in \mathcal{Z}_0, i\in [M] $, $k\in[K]$, and suppose that 
$ Y_{i}^{(k)}= \mu_{P_k}(X_i^{(k)})+\epsilon_{i}^{(k)},i\in [M]$, where 
$\mu_{P_k}$ are maps that may depend on $P_k$,
and $\epsilon_{i}^{(k)},i\in [M]$ are i.i.d.~zero-mean random variables whose distributions may 
depend on $P_k,k\in [K]$. 
\end{example}

Let us consider the set $[KM] = \{1,\ldots, KM\}$
and 
for each $a\in[KM]$,
where 
$a = (k-1)\cdot M + i$ for a unique $i\in [M]$, $k\in [K]$,
associate with $a$
the random variable $Z_{i}^{(k)}\in \R$.
Thus the random variables  $Z_{i}^{(k)}\in \R$, $i\in [M]$ in the $k$-th branch are associated with the block $b_k=\{(k-1)\cdot M + i, i\in [M]\}$.
We let $\Lambda_{K,M} \subset \S_{[KM]}$ be the group of $KM$-permutations that map each block $b_k$
into some other block $b_{k'}$ in a bijective way.
Then for both the unsupervised and supervised cases, the distribution of the data is $\Lambda_{K,M}$-invariant.
\item {\bf Transformations and properties of prediction sets.}
We aim to build prediction sets for some unknown components of the last branch, 
hoping to improve prediction by 
pooling information both within and across branches. 
Since this setting provides a particularly fruitful example for our methods, we will dedicate 
all of \Cref{sec:twolayer-detail} to explain how we construct prediction sets. 
\end{itemize}

\section{Two-layer Hierarchical Model}\label{sec:twolayer-detail}

We now provide the details of using our methods in the 
two-layer hierarchical structure introduced in \Cref{sec:twolayer}.

\subsection{Methodological Considerations for Unsupervised Learning}\label{sec:meta_hetero}

In this section, we develop our methodology for building prediction sets in the unsupervised case from Example \ref{ul}.

{\bf Transformations.} For all $k\in[K]$, define
$\bar{Z}_k = \sum_{j=1}^M Z_{j}^{(k)}/M$, 
$\hat\sigma_k^2=1$ if $M=1$,
$$\hat\sigma_k^2=
\frac{\sum_{j=1}^M(Z_j^{(k)}-\bar{Z}_k )^2}{M-1}$$ 
otherwise,
and 
$\bar{Z} = \sum_{k=1}^K \bar Z_{k}/K$.
For some constant $c\ge 0$,
and for all $k\in[K]$,
define the events $A_{k,c} = \{|\bar{Z}_k-\bar{Z}|\le c\hat\sigma_k/\sqrt{M}\}$.
These capture the events that the means of the elements within the $k$-th branch are close to the grand mean.
Let $A_{k,c}^{\complement}$ be the complement of $A_{k,c}$,
and $I(A)$ be the indicator function of an event $A$, which equals $I(A)=1$ if $A$ happens, and $I(A)=0$ otherwise.
We also define, for all $k\in[K],i\in[M]$,
\beq\label{usls}
\tilde{Z}_{i}^{(k)}:=
\frac{|Z_{i}^{(k)}-[\bar{Z}I(A_{k,c})+\bar{Z}_kI(A_{k,c}^{\complement})]|}
{\hat\sigma_k},
\eeq
which are the standardized absolute deviations of $Z_{i}^{(k)}$ from the grand mean if the branch mean is close to that, or to branch mean itself otherwise.
Here $c$ is an absolute constant that can be 
set as $c=2$, as an approximate 97.5\% quantile of the standard Gaussian distribution. 
As explained later, it  can potentially also be optimized by minimizing a prediction loss. 
In Section \ref{sec:GNN_unsup} of the Appendix, we explain
how $(\tilde{z}_{i}^{(k)})_{k\in[K],i\in[M]}  =: V(z)$ can be obtained as functions of $z$
via a fixed message-passing graph neural network $V$ on a graph obtained
by constructing proxy statistics for the 
nodes 
in the zeroth- and first-layer nodes (also detailed below for the supervised case).
We construct the prediction region 
 \eqref{T} 
for $Z_{M}^{(K)}$ by using that $\tZ$ is $\Lambda_{K,M}$-distributionally invariant.

{\bf Properties of prediction sets.} To understand our procedure, 
suppose for a moment that 
$P_k,k\in [K]$  have finite expectations $\mu_k,k\in [K]$; but we emphasize that our method does not require this condition.
When some branch means $\bar Z_{k}\approx\mu_k,k\in [K]$ and $\bar{Z}\approx \sum_{j=1}^{K}\mu_j/K$ are very different---i.e., 
on the event $A_{k,c}^{\complement}$---our procedure centers observations within those branches by estimating the within-branch means $\mu_k,k\in [K]$. 
On the other hand, when $\bar Z_{k}\approx \bar{Z}$ holds for all branches $k\in [K]$---i.e., on $A_{k,c}$---we pool all observations
and mimic standard conformal inference.
Therefore, our procedure interpolates
prediction sets built using each individual branch and 
prediction sets built using full standard conformal inference.

We let $J_{K,M} = [M]\times [K]\setminus \{(M,K)\}$ denote the 
indices of the
 fully observed datapoints.
The following result formalizes the coverage guarantee of our result: 

\begin{proposition}\label{meta_unsup}
If the data $Z$ 
follows the two-layer hierarchical model introduced at the beginning of \Cref{sec:twolayer},
and the observation function has values $\obs(z)=(z_i^{(k)})_{(i,k)\in J_{K,M}}$ for all $z$, 
then the prediction set for $Z_M^{(K)}$ 
from  \eqref{T} 
with $\psi(\tz)=\tz_M^{(K)}$ and $t_{\tilde{z}}=Q_{1-\alpha}((\tilde{z}_{i}^{(k)})_{k=1,i=1}^{K,M})$
has coverage at least $1-\alpha$.
Moreover, if $\tZ$ has a continuous distribution, then the coverage is at most $1-\alpha+1/(KM)$.
\end{proposition}

The proof of Proposition \ref{meta_unsup} is deferred to \Cref{pfm}.
We compare our method with alternatives in Section \ref{sec:comp}.

\subsection{Methodology Considerations for Supervised Learning}\label{sup_learning}
In this section, 
we consider the two-layer hierarchical model
in a supervised learning setting. 

{\bf Transformations.} To ease the computational burden, 
we adopt split predictive inference. 
For every branch $k\in [K]$, we set the first $M'$---approximately half---datapoints to be the training sample,
and 
let $Z_{\tr}$ be the training data
gathered from all branches. 
We 
fit $\hat\mu(\cdot;Z_{\tr}):\R\to\R$ 
based on $Z_{\tr}$,
such that $x\mapsto \hat\mu(x;Z_{\tr})$
is an estimator of the regression function in the pooled data.
Further, for all branches, 
we fit\footnote{We first train $\hat\mu(\cdot;Z_{\tr})$ using all training data and train $\tilde{\mu}_k(\cdot;Z_{\tr}),k\in [K]$ to approximate the regression function of the residuals $(X_{k,i},Y_{k,i}-\hat\mu(X_{k,i});Z_{\tr}),i\in [M]$ from the $k$-th tree, $k\in[K]$.
We finally let $\hat\mu_k(\cdot;Z_{\tr})=\hat\mu(\cdot;Z_{\tr})+\tilde{\mu}_k(\cdot;Z_{\tr}).$} $
\hat\mu_k(\cdot;Z_{\tr}),k\in [K]$
based on $Z_{\tr}$,
as estimators of the within-branch regression functions $\mu_{P_k}$ from Example \ref{sl}.

Using $Z_{\tr}$,
we also fit 
pointwise confidence bands $x\mapsto \hat\sigma_k(x;Z_{\tr}),k\in [K]$
by estimating the pointwise  standard error of $x\mapsto\hat\mu_k(x;Z_{\tr}),k\in [K]$ over the randomness of the fitting process. 
Many classical estimators $\hat\mu_k(\cdot;Z_{\tr})$ have explicit expressions for standard error curves $\hat\sigma_k(\cdot;Z_{\tr})$, $k\in [K]$ 
including parametric models, non-parametric kernel regression, splines, etc. 
We emphasize that our method does not require any coverage properties for $\hsigma_k$.

Suppose that there are $M$ remaining datapoints in each branch, and without loss of generality, call them $(X_i^{(k)},Y_i^{(k)})$, $k\in[K]$, $i\in [M]$.
For $k\in[K]$, $i\in [M]$, let
$$\bar Z_i^{(k)} = (Y_i^{(k)}-\hat\mu_k(X_i^{(k)}; Z_{\tr}),
Y_i^{(k)}-\hat\mu(X_i^{(k)}; Z_{\tr}),\hat\sigma_k(X_i^{(k)}; Z_{\tr})).$$
We require that 
$\bar Z = (\bar Z_i^{(k)})_{k\in[K], i\in [M]}$ is  $\Lambda_{K,M}$-distributionally invariant.
This follows if $Z=(Z_{\tr},$ $(Z_i^{(k)})_{k\in[K], i\in [M]})$ is $\Lambda_{K,M+M'}$-distributionally invariant and the map 
$Z\mapsto \bar Z$ is  $\Lambda_{K,M+M'}$-distributionally equivariant.
To ensure this, we use the same algorithm for training $\hat\mu_k(\cdot;Z_{\tr})$, $\hat\sigma_k(\cdot;Z_{\tr})$ in all branches $k\in [K]$, and ensure they are all invariant to the order of the data within branch $k$.

In what follows, we suppress the dependence of  $\hat\mu, \hat\mu_k, \hat\sigma_k$ on the training data.

For all $k\in[K],i\in[M]$, define 
\begin{align}\label{numerator}
{\bar{Z}_i^{'(k)}}=\bigg[Y_{i}^{(k)}-\hat\mu(X_i^{(k)})I\bigg(\bigg|\frac{\hat\mu_k(X_i^{(k)})-\hat\mu(X_i^{(k)})}{\hat\sigma_k(X_i^{(k)})}\bigg|\le c\bigg)-\hat\mu_k(X_i^{(k)})I\bigg(\bigg|\frac{\hat\mu_k(X_i^{(k)})-\hat\mu(X_i^{(k)})}{\hat\sigma_k(X_i^{(k)})}\bigg|>c\bigg)\bigg].
\end{align}
We set $\hat\epsilon_k=1,$ for all $k\in [K]$ when $M=1$
and $\hat\epsilon_{k}=\sqrt{\sum_{i=1}^{M}(\bar{Z}_i^{'(k)})^2/(M-1)}$ otherwise. 

Let
$\tZ=(\tZ_1,\ldots, \tZ_K)^\top$, 
where for all $k\in[K]$,
$\tZ_k=(\tZ_{1}^{(k)},\ldots,\tZ_{M}^{(k)})^\top$
and for all $k\in[K],i\in[M]$,
we define the map $\bar Z\mapsto V(\bar Z) = (V_i^{(k)}(\bar Z))_{i=1,k=1}^{M,K}$
via 
\begin{align}\label{stat_1}
\tZ_i^{(k)}
:= V_i^{(k)}(\bar Z)
:=\frac{\bar{Z}_i^{'(k)}}{\hat\epsilon_{k}}.
\end{align}

We now argue that if the data 
$\bar Z_{i}^{(k)}, i\in [M] $, $k\in[K]$ are
$\Lambda_{K,M}$-distributionally invariant, 
then $\tZ$ also satisfies this property.
To see this, we define a special fixed message-passing graph neural net computing $V$ on 
the two-layer tree from \Cref{fig:tree2} (right), 
which captures the invariances of $\bar Z$,
mapping $\bar Z\mapsto \tZ=V(\bar Z)$.
Since message-passing GNNs are equivariant, it will follow that $\tZ$ is also $\Lambda_{K,M}$-distributionally invariant.
In fact, our MPGNN operates only on the subgraph of the two-layer tree excluding the zeroth-layer node.
The message passing GNN is defined in five steps:

\begin{itemize}
\item  Step 1: We 
initialize
the leaves $z_{i}^{(k,0)},i\in [M],k\in [K]$,
to have three channels: $z_i^{(k,0)}=(z_{i,1}^{(k,0)},z_{i,2}^{(k,0)},z_{i,3}^{(k,0)})^\top$, where
$z_{i,1}^{(k,0)}=y_{i}^{(k)}-\hat\mu_k(x_{i}^{(k)})$, $z_{i,2}^{(k,0)}=y_{i}^{(k)}-\hat\mu(x_{i}^{(k)})$,
$z_{i,3}^{(k,0)}=\hat\sigma_k(x_{i}^{(k)})$. 
We initialize the 
first-layer nodes
as all ones vectors with three channels: 
$p^{(k,0)}=(p_1^{(k,0)},p_2^{(k,0)}$, $p_3^{(k,0)})^\top
$ $=(1,1,1)^\top$, $k\in[K]$.

\item Step 2: 
Let the kernel $\mathbb{K}$
be defined by $\mathbb{K}(x)=I(|x|\le c)$ for all $x$.
We update 
each leaf individually 
as $z_{i}^{(k,1)}=V_1(z_{i}^{(k,0)})=(z_{i,1}^{(k,0)},z_{i,2}^{(k,0)},(z_{i,2}^{(k,0)}-z_{i,1}^{(k,0)})\cdot \mathbb{K}(|(z_{i,2}^{(k,0)}-z_{i,1}^{(k,0)})/z_{i,3}^{(k,0)}|))^\top,i\in[M],k\in [K]$. 
This corresponds to taking the map $\lambda_0$ in \eqref{u} to only depend on its first input. 
The updated third coordinate becomes
\begin{align*}
    &(z_{i,2}^{(k,0)}-z_{i,1}^{(k,0)})\cdot \mathbb{K}(|(z_{i,2}^{(k,0)}-z_{i,1}^{(k,0)})/z_{i,3}^{(k,0)}|)\\
    &\quad=[\hat\mu(x_{i}^{(k)})-\hat\mu_k(x_{i}^{(k)})]\cdot I(|\hat\mu(x_{i}^{(k)})-\hat\mu_k(x_{i}^{(k)})|/\hat\sigma_k(x_i^{(k)})\le c).
\end{align*}
We keep the values of the first-layer nodes unchanged: 
$p^{(k,1)}= p^{(k,0)}$, $k\in[K]$.

\item  Step 3: We update the leaves as $z_{i}^{(k,2)}=V_2(z_{i}^{(k,1)})=(|y_i^{(k)}-\hat\mu_k(x_i^{(k)})-z_{i,3}^{(k,1)}|,z_{i,2}^{(k,1)},z_{i,3}^{(k,1)})^\top$, 
$i\in[M],k\in [K].$
This also corresponds to a map $\lambda_0$ that only depends on its first input. 

\item Step 4: We update the first layer nodes 
as 
$p^{(k,2)}=(p_1^{(k,2)},1,1)^\top,k\in[K]$, where 
$p_{1}^{(k,2)}
=V_{32}(p^{(k,1)},\sum_{i=1}^{M}V_{31}(z_{i}^{(k,2)},p^{(k,1)}))$, and we have  
$V_{31}(x,y)=x^2/(M-1)$ and $V_{32}(x,y)=\sqrt{y}$. 
This corresponds to a standard message passing update in \eqref{u}.
Thus, after the update, we have $p^{(k,2)}=(\hat\epsilon_{k},1,1)^\top$, 
where 
$\hat\epsilon_{k}^2
=
\sum_{i=1}^{M}\Big(z_{i,1}^{(k,2)}\Big)^2/(M-1)$, $\hat\epsilon_{k}\ge 0$
for all  $k\in [K]$. In this step, we fix the values of the leaves.

\item Step 5: We update the leaves $z_{i}^{(k,3)}=V_{41}(z_{i}^{(k,2)},V_{42}(p^{(k,2)},z_{i}^{(k,2)}))$,
where
$V_{42}(p^{(k,2)},z_{i}^{(k,2)}) = p^{(k,2)}$
and 
$V_{41}(z_{i}^{(k,2)},p^{(k,2)})
= (z_{i,1}^{(k,2)}/p_1^{(k,2)}$, $z_{i,2}^{(k,2)}$, $z_{i,3}^{(k,2)})^\top$.
Thus,
$z_{i}^{(k,3)}=(|y_i^{(k)}-\hat\mu_k(x_i^{(k)})-z_{i,3}^{(k,1)}|/p_1^{(k,2)}$, $z_{i,2}^{(k,2)}$, $z_{i,3}^{(k,2)})^\top$, $i\in[M],k\in [K].$
 The first entry, $z_{i,1}^{(k,2)}$, becomes our statistic \eqref{stat_1}.

\end{itemize}
{\bf Properties of prediction sets.} We have the following coverage guarantee.

\begin{proposition}\label{meta_sup}
If the data 
$\bar Z_{i}^{(k)}, i\in [M] $, $k\in[K]$ are
$\Lambda_{K,M}$-distributionally invariant, 
and the observation function has values $\obs(z)=((z_i^{(k)})_{(i,k)\in J_{K,M}},x_{M}^{(K)})$ for all $z$,
then the prediction set for $y_{M}^{(K)}$ 
from  \eqref{T} 
with $\psi(\tz)=\tilde z_{M,1}^{(K)}$ 
and $t_{\tz}=Q_{1-\alpha}((\tilde{z}_{i,1}^{(k)})_{k=1,i=1}^{K,M})$
has coverage at least $1-\alpha$.
Moreover, if $\tilde Z$ has a continuous distribution, then the coverage is at most $1-\alpha+1/(KM)$.

\end{proposition}

The proof of Theorem \ref{meta_sup} is deferred to \Cref{pfm}.

\noindent\textbf{Remark}: In practice, similar to the unsupervised case, one could choose $c$ as a quantile of the standard normal distribution, or by minimizing a loss. 
For example, we could
compute 
the residual standard errors $\tilde \sigma_{k,\epsilon}^2$
of $\hat\mu_k$ on the training data,
and then 
minimize over $c$  the following empirical loss:
$$ \sum_{k,i=1}^{K,M}
\bigg[Y_{i}^{(k)}-\hat\mu(X_i^{(k)})I\bigg(\bigg|\frac{\hat\mu_k(X_i^{(k)})-\hat\mu(X_i^{(k)})}{\hat\sigma_k(X_i^{(k)})}\bigg|\le c\bigg)-\hat\mu_k(X_i^{(k)})I\bigg(\bigg|\frac{\hat\mu_k(X_i^{(k)})-\hat\mu(X_i^{(k)})}{\hat\sigma_k(X_i^{(k)})}\bigg|>c\bigg)\bigg]^2/{\tilde\sigma_{k,\epsilon}^2}.$$ 
Since the objective function is $\Lambda_{K,M}$-invariant, it is not hard to choose an approximate minimizer in an $\Lambda_{K,M}$-invariant way; by simple one-dimensional optimization. 
By our general theory, Proposition \ref{meta_sup} will still hold for this choice of $c$.

\subsection{Comparison with Other Methods}\label{sec:comp}
In this section, we discuss the performance of several
alternative benchmark methods 
and compare them with our proposed method. 
For simplicity, we assume 
that in the unsupervised case, 
$Z_{i}^{(k)}$ have equal variances for all $i\in [M]$, $k\in [K]$ and
in the supervised case, the same holds for 
$\epsilon_{i}^{(k)}$, $i\in [M]$, $k\in [K]$.
We provide a brief discussion of the scenario where the variances are different at the end of this section.\\

\noindent\textbf{Benchmark 1: Single Tree}. Since the random variables on the leaves are exchangeable within the same branch, one possible way to construct a prediction set is 
to use the $M-1$ observed calibration datapoints together with the $M$-th unobserved test datapoint 
in the last branch to construct a classical conformal prediction set
$$ T_1=\Big\{z_{M}^{(K)}:s(z_{M}^{(K)})\le 
Q_{1-\alpha}(s(z_{1}^{(K)}),s(z_{2}^{(K)}),\ldots, s(z_{M}^{(K)}))\Big\}.$$
Usually, for unsupervised learning, one would 
set $s(z)=|z|$ for all $z$,
and for supervised learning, one would 
set $s(z)=|y-\hat\mu_K(z)|$, where $\hat \mu_K$ is a function  fit
to the training data within the last branch.
Although this method yields a prediction set with valid coverage,
it may also exhibit a degree of conservatism
due to not using information from other branches.
This can result in the prediction set being large (or even including the entire space)
when the sample size within the branch is small.\\ 
\noindent\textbf{Benchmark 2: Split conformal prediction}. 
We compare our method with classical split conformal prediction. 
Denote by $\textrm{aver}(S)$ the average of a finite set $S\subset \R$.
In unsupervised learning, 
one version of the standard conformal prediction set contains $z_{M}^{(K)}$ such that
\begin{align*}
    |z_{M}^{(K)}-\textrm{aver}(z_i^{(k)},i\in[M],k\in [K])|\le 
    Q_{1-\alpha}(|z_{i}^{(k)}-\textrm{aver}(z_i^{(k)},i\in[M],k\in [K])|,i\in[M],k\in [K]).
\end{align*}
In a supervised learning setting, a classical conformal prediction set is given by $y_{M}^{(K)}$ such that
\begin{align*}
    |y_{M}^{(K)}-\hat \mu(x_M^{(K)})|\le Q_{1-\alpha}(|y_{i}^{(k)}-\hat \mu(x_i^{(k)})|,i\in[M],k\in [K]).
\end{align*}
where $\hat\mu$ is a regression function based on the training sample, as described in \Cref{sup_learning}.

Observe that 
in the unsupervised case 
one can view the problem as 
predicting the next observation from the distribution $P_K$. 
Defining $m = \textrm{aver}(z_i^{(k)},i\in[M],k\in [K]),i\in[M],k\in [K]$,
the length of the prediction set in unsupervised learning 
will be $2\cdot Q_{1-\alpha}(|z_{i}^{(k)}-m|)$, which---for large $M$---is close to the difference between the upper and lower $\alpha/2$-quantiles of the mixture distribution $\frac{1}{K}\sum_{k=1}^{K}P_{k}$ (assuming, without loss of generality, that this distribution is symmetric); rather than of $P_k$.

For supervised learning, the length of the prediction set will be $2\cdot Q_{1-\alpha}(|Y_i^{(k)}-\hat\mu(X_i^{(k)})|,i\in [M],k\in [K])=2\cdot Q_{1-\alpha}(|\epsilon_{i}^{(k)}+\mu_{P_k}(X_i^{(k)})-\hat\mu(X_i^{(k)})|,i\in [M],k\in [K])$. 
When $\mu_{P_k},k\in [K]$ 
differ a great deal
for different $P_k,k\in [K]$, 
training $\hat\mu$ by mixing their training datapoints will likely lead to a wider prediction set. \\

\noindent\textbf{Benchmark 3: Subsampling}.  We will also compare 
to a subsampling method proposed in \cite{dunn2022distribution},
which consists of uniformly sampling 
one observation from each of the first $K-1$ branches. 

The unobserved $Z_M^{(K)}$ 
is exchangeable with the sampled random variables, 
and thus a
standard conformal prediction set can be constructed using the sub-sample. 
If $K$ is sufficiently large,
we expect the length of the prediction set to be close 
to that of a set obtained via
standard conformal prediction using the full data.
Indeed, the latter method fits the quantiles of the mixture of the distributions of all branches. 
In addition, 
\cite{dunn2022distribution} introduced a repeated subsampling approach aimed at improving the stability of the prediction set. 
Recalling that we aim to predict the next observation from $P_K$, 
this method provides a valid prediction set, but it again effectively estimates the quantiles of the mixture distribution of all $P_k$s, $k\in [K]$.\\

We further discuss the advantages of our method compared 
with these benchmarks. 
Taking the supervised learning setting as example, 
when the distribution $P_K$ deviates from other distributions, our approach  provides a prediction set tailored to the last distribution $P_K$ that we are interested in---in contrast to benchmarks 2 and 3. 
In addition, we also leverage data
from other branches,
furnishing smaller prediction sets than benchmark 1 
when there are only few observations in the final branch.

When all $\mu_{P_k}, k\in [K]$, are close to each other, 
including confidence bands
enables us to 
achieve similar sized regions
to standard conformal prediction. 
As a result, our methodology
effectively offers the best of conformal prediction within the branch and using the full dataset. 
Thus, it furnishes an 
approach to predictive inference
that performs well under the heterogeneity of the distributions $P_k$, $k\in[K]$.

Finally, we mention that our approach 
can 
also handle the 
scenario where the residual variances are different
across branches.
In such cases, 
even if 
standard conformal prediction (Benchmark 2)
achieves satisfactory coverage probability,
this coverage
may be uneven across the branches due to their different variances. 
For instance, 
a branch with high variance 
requires wider
prediction sets
to prevent under-coverage; which is possible with our approach but not straightforward with standard conformal prediction. 

\subsection{Extension to Random Sample Sizes}

In this section, we study an extension of our methodology 
for settings with imbalanced observations across branches. 
We begin by introducing the probability models studied in this section.

{\bf Data, group, and action.}  We let a random vector $Z=(Z_1^\top, \ldots, Z_K^\top)^\top \in \mZ_0^* : = \cup_{j\ge 0}\mZ_0^j $ be generated from a joint distribution $\mathcal{P}$. 
We assume exchangeability across the $K$ components of $(Z_1^\top, \ldots, Z_K^\top)^\top$. 
 In addition, the dimensions of 
 $Z_i,i\in [K]$, denoted as 
 $N_i,i\in [K]$, are also random variables with $N_i\in \{0,1,\ldots,\}$, for any $i\in [K]$.
 Letting $\vec  N:=(N_1,\ldots,N_K), \vec n:=(n_1,\ldots,n_K)$, 
 for any $\vec n\in \mathbb{N}^K$, 
 we assume that 
  conditional on $\vec N=\vec n$,
 and for all $k\in [K]$,
 the 
coordinates of
 $Z_k=(Z_{1}^{(k)}, \ldots, Z_{n_k}^{(k)})^\top$ are exchangeable across the $n_k$ observations.
The model from \Cref{sec:twolayer} is a special case where $n_k = M$ for all $k\in[K]$.

For any given $\vec n \in \mathbb{N}^K$,
conditional on $\vec N= \vec n$, 
we define $\mathcal{G}_{\vec n}=\S_{n_1} \otimes \S_{n_2} \otimes \ldots \otimes \S_{n_K}$
as the direct product of
the permutation groups
$\S_{n_k}$ of the sets $[n_k]$. 
For any $g_{\vec n}\in \mathcal{G}_{\vec n}$, it holds that $g_{\vec n}=g_{n_1}\otimes g_{n_2}\ldots\otimes g_{n_K}$, where $g_{n_i}\in \S_{n_i},i\in [K]$. 
Choose
for all $k\in[K]$
the permutation actions  
$g_{n_k}\cdot Z_k$.
Then 
we have $g_{\vec n}\cdot Z=(g_{n_1}\cdot Z_1,\ldots, g_{n_K}\cdot Z_K)$. 
We let $U_{\vec n}$ be uniform measure over
$\mathcal{G}_{\vec n}$
and 
$\mZ = \mZ_{\vec n} =: \prod_{k=1}^K\mZ_0^{n_k}$,
$\tmZ =\tmZ_{\vec n} = \prod_{k=1}^K\tmZ_0^{n_k}$.

{\bf Transformations.} We also 
consider a
$\mathcal{G}_{\vec n}$-distributionally equivariant
map $V:\mathcal{Z}\rightarrow \tilde{\mathcal{Z}}$, 
 with respect to the actions 
 for which 
 $\rho(g)=g$ and $\tilde{\rho}(g)=g$ for all $g\in \mG_{\vec n}$,
i.e., 
$V(G_{\vec n}\cdot Z)=_d G_{\vec n}\cdot V(Z),\, \textnormal{ if } G_{\vec n}\sim U_{\vec n}$.

{\bf Properties of prediction sets.} We let $J_{\vec n} = \{(i,k): i\in[n_k], k\in [K]\} \setminus \{(n_K,K)\}$ 
denote the 
indices of the
 fully observed datapoints.
 In the unsupervised 
case,
the observation function has values $\obs(z)=(z_i^{(k)})_{(i,k)\in J_{\vec n}}$ for all $z$,
while in 
the supervised 
case,
it has values $\obs(z)=((z_i^{(k)})_{(i,k)\in J_{\vec n}}, x_{n_K}^{(K)})$ for all $z$.
In the unsupervised case,
we aim to predict $\psi_{\vec n}(\tz)=\tz_{n_K}^{(K)}$ and define, for all $\tz\in \tmZ$,
 \begin{align}\label{quantile_random}
 t_{\tz}=Q_{1-\alpha}\bigg(\frac{1}{K}\sum_{k=1}^{K}\frac{1}{n_k}\sum_{m=1}^{n_k}\delta_{\tz_{m}^{(k)}}\bigg).
 \end{align}
We also define the prediction set
\beq\label{T2}
T^{\mathrm\OurMethod}(\zo,\vec n) = \Big\{z: [V(z)]_{n_K}^{(K)}\le t_{V(z)},\, \o(z)= \zo\Big\}.
\eeq
In the supervised case, 
we define $V$ similarly 
to the definition from \Cref{sup_learning},
and $t$ and $T^{\mathrm\OurMethod}$ as above.
Then we have the following coverage guarantees. 
\begin{theorem}\label{prop_random}
In the setting described above in this section, we consider
 $\alpha\in[0,1]$,
$t_{\tz}$ defined in \eqref{quantile_random} 
and the observation function with $\obs(z)=((z_i^{(k)})_{(i,k)\in J_{K,M}},x_{M}^{(K)})$ for all $z\in\mZ$. 
For the prediction region defined in \eqref{T2}, we have
$P(Z\in T^{\mathrm\OurMethod}(\Zo,\vec N ))\ge 1-\alpha.$
When all $Z_{j}^{(k)},k\in [K],$ $j\in [N_k]$, are continuous random variables, we also have
\begin{align*}
P(Z\in T^{\mathrm\OurMethod}(\Zo,\vec N ))\le 1-\alpha+\E\max_{j\in [K]}\frac{1}{K\cdot N_j}.
\end{align*}
\end{theorem}

The proof of this Proposition is deferred to \Cref{app: prop_random_proof} of the Appendix.

\cite{lee2023distribution}
consider the closely related
problem of constructing a prediction set for 
the \emph{first} observation in
a new branch in the supervised learning regime 
an identical two-layer hierarchical model with random sample sizes.
This problem is distinct from the question of predicting the last unobserved outcome 
considered in \Cref{prop_random}.
However, our general framework includes their problem as a special case.
For simplicity, we will explain this in the unsupervised case.
We let the  observation function be $\obs(z)=((z_i^{(k)})_{k\in [K-1], i\in [n_k]})$. 
We then set
our prediction set in \eqref{T2} 
by taking
 $\tilde{z}_m^{(k)}$ in \eqref{quantile_random} as in \eqref{usls}
with $c=\infty$
and $\hat\sigma_k:=1$.

Then, our prediction region from  \eqref{T2} 
is equivalent to one for $(z_1^{(K)}, \ldots, z_{n_K}^{(K)})$, 
where the number of observations $n_K$ is not known.
Thus, our prediction region includes a union over the unknown values of $n_K\ge 0$.
The
induced prediction region for $z_1^{(K)}$ given by  the union of the 
projection of these prediction regions into their first coordinates is clearly a valid $1-\alpha$-coverage region.
Further, it is immediate that the union is included in the one with $n_K=1$, which becomes 
\beqs
\Bigg\{z_{1}^{(K)}: \tz_{1}^{(K)}\le 
Q_{1-\alpha}\bigg(\frac{1}{K}
\sum_{k=1}^{K-1}\frac{1}{n_k}\sum_{m=1}^{n_k}\delta_{\tz_{m}^{(k)}} + \frac{1}{K} \delta_{\tz_{1}^{(K)}} \bigg)\Bigg\}.
\eeqs
Up to changes of notation (such as our $K$ being their $K+1$), this recovers the HCP method of \cite{lee2023distribution}.

However, HCP does not aim to form predictions in the setting 
when there are multiple observations in the last branch. 
In particular, HCP
can lead to  wider prediction sets when $\mu_{P_k}$ differ a great deal across different $P_k,k\in [K]$, 
since the algorithm does not take the heterogeneity across different branches into account. 
We present a detailed simulation to compare 
the performance of these methods in \Cref{rss}.

\section{Simulations and Data Analysis Examples}

\subsection{Simulation Studies}
In this section, we provide simulations to corroborate the efficacy of our proposed approach, focusing on the two-layer hierarchical model from \Cref{sec:twolayer-detail}. 
Specifically, we conduct simulations under the scenarios of both unsupervised and supervised learning with both non-random and random instances of $\vec N$, respectively. 
We present the simulation results with a fixed $\vec N$ here, and we defer the results with random $\vec N$ to the appendix (\Cref{rss}).

{\bf Unsupervised Learning: Fixed Sample Size.}
We now present the simulation results 
with our proposed method in the context of unsupervised learning. 
We set the number of branches to $K=20$, where each branch has $M=15$ observations; 
but $Z_{M}^{(K)}$ is unobserved and needs to be predicted. 
We let $Z_{i}^{(k)},i\in [M]$ be sampled i.i.d.~from $\N(\mu_k, 0.5)$, with $k \in [K]$, 
and where $\mu_k,k\in [K]$ follow a normal distribution $\N(0, \sigma^2)$.

We consider $\sigma^2\in \{10, 2, 0.5, 0\}$. 
When $\sigma^2=0$, all location parameters are equal, reducing to the special case of full exchangeability. 
To construct prediction sets, we apply the method described in Section \ref{sec:meta_hetero}. 
In all numerical examples in this paper, we set $c= 2$, based on the quantiles of the standard Gaussian distribution. 
The simulation results are presented in Table \ref{tab: fix_N_unsup}.
We provide a summary and conclusions in the next subsection.

   \begin{table}[th]
	\begin{center}
		\def\arraystretch{1.2}
		\setlength\tabcolsep{4pt}
		\begin{tabular}{c||ccccc}
			\toprule
			&Method  & $\sigma^2=10$ & $\sigma^2=2$ & $\sigma^2=0.5$  & $\sigma^2=0$   \\
			\hline
   \multirow{8}{*}{$\alpha=0.05$} & Length: \OurMethod~& 2.050 (0.012) & 2.054 (0.015) & 2.088 (0.023)&  1.996 (0.014)\\ 
   & Coverage: \OurMethod~& 0.959 (0.020)  & 0.948 (0.024)  &0.950 (0.025) & 0.945 (0.025)\\
       & Length: Conformal &40.614 (0.818) & 7.948 (0.150) & 2.748 (0.025) & 1.974 (0.012)\\
           & Coverage: Conformal & 0.974 (0.025)  & 0.957 (0.023) & 0.951 (0.023) & 0.944 (0.025)\\
       & Length: Subsampling &44.254 (0.899) & 9.115 (0.190) & 3.122 (0.070) & 2.208 (0.049)\\
           & Coverage: Subsampling&  0.947 (0.026) & 0.959 (0.021) & 0.947 (0.026) & 0.946 (0.025)\\
           & Length: Single-Tree & Inf & Inf & Inf & Inf \\
           & Coverage: Single-Tree& 1.0 (0.0)  & 1.0 (0.0)  & 1.0 (0.0) & 1.0 (0.0) \\
    \hline
       \multirow{8}{*}{$\alpha=0.15$} & Length: \OurMethod~&1.496 (0.009) & 1.502 (0.010)  & 1.527 (0.010) & 1.465 (0.012) \\ 
   & Coverage: \OurMethod~&0.848 (0.033)  & 0.859 (0.036)  & 0.859 (0.037) & 0.845 (0.030)\\
       & Length: Conformal & 28.767 (0.645)   &5.827 (0.113)  & 2.020 (0.017) & 1.455 (0.007) \\
           & Coverage: Conformal &0.855 (0.035)  &0.852 (0.039) &0.847 (0.033) & 0.845 (0.028)\\
                  & Length: Subsampling &31.064 (0.685) & 6.365 (0.149) & 2.183 (0.047) & 1.539 (0.034)\\
           & Coverage: Subsampling&0.852 (0.035)  &0.849 (0.039)  &0.856 (0.035)  &0.839 (0.034) \\
           & Length: Single-Tree & 1.658 (0.038) &1.658 (0.029) &1.655 (0.032) & 1.649 (0.037) \\
           & Coverage: Single-Tree &0.868 (0.036)  &0.867 (0.029)  &0.874 (0.034)  &0.864 (0.029) \\
    \hline

			\bottomrule
		\end{tabular}
	\end{center}
	\caption{Prediction set length and coverage probability of our method and three benchmarks discussed in \Cref{sec:comp} with fixed $\vec N$ in unsupervised learning. We conduct 40 independent repetitions of the experiment. For each trial, we average prediction lengths and indicator values of coverage over 100 test datapoints (by generating the training, calibration and test data 100 times independently). 
 The reported values are the average length and coverage over these 40 trials, and the associated standard errors (in brackets).} 
	\label{tab: fix_N_unsup}
\end{table}

{\bf Supervised learning: Fixed sample size.}
We next study the simulation performance of our proposed method in a simple supervised learning example. 
Specifically, we let $\theta_k,k\in [20]$ be sampled from $\N(0,\sigma^2)$, 
where we consider $\sigma^2\in \{10, 2, 0.5, 0\}$.
We let $Y_{j}^{(k)},k\in[20],$ $j\in [30]$ be sampled from $Y_{j}^{(k)}=\theta_{k}X_{j}^{(k)}+\epsilon_{j}^{(k)}$,
where
$X_{j}^{(k)}\sim \mathrm{Unif}(-0.5,0.5)$ and $\epsilon_{j}^{(k)}\sim \N(0,0.5^2)$ for all $k\in [20],$ $j\in [30].$
For ease of computation, 
we conduct split conformal prediction where we split half of the data (15 datapoints) 
in each branch to fit $\hat\mu_k,k\in[20]$ via linear regression,
and thus also obtain a confidence band $\hat\sigma_k$ induced by the linear regression estimator. 
In addition, we also train $\hat\mu$ using all training data via linear regression.
We present the performance of our constructed prediction sets and of baseline methods 
in Table \ref{tab:fix_N_sup}.

    \begin{table}[ht]
	\begin{center}
		\def\arraystretch{1.2}
		\setlength\tabcolsep{4pt}
		\begin{tabular}{c||ccccc}
			\toprule
			&Method  & $\sigma^2=10$ & $\sigma^2=2$ & $\sigma^2=0.5$  & $\sigma^2=0$   \\
			\hline
   \multirow{8}{*}{$\alpha=0.05$} & Length: \OurMethod~&  2.048 (0.012) & 2.068 (0.014)  & 2.044 (0.011) &  1.991 (0.013) \\ 
   & Coverage: \OurMethod~&  0.952 (0.024) & 0.952 (0.023)  &0.950 (0.023) & 0.955 (0.019)\\
       & Length: Conformal & 12.365 (0.234) & 3.057 (0.034) & 2.053 (0.011) & 1.973 (0.012) \\
           & Coverage: Conformal & 0.945 (0.023)  & 0.949 (0.022) & 0.952 (0.023) & 0.953 (0.019) \\
       & Length: Subsampling & 14.091 (0.419) & 3.418 (0.099) & 2.239 (0.057) & 2.151 (0.043) \\
           & Coverage: Subsampling& 0.947 (0.023) & 0.953 (0.022)  & 0.954 (0.021)  & 0.953 (0.022)\\
           & Length: Single-Tree & Inf & Inf & Inf & Inf  \\
           & Coverage: Single-Tree& 1.0 (0.0)  &  1.0 (0.0) & 1.0 (0.0) & 1.0 (0.0) \\
    \hline
       \multirow{8}{*}{$\alpha=0.15$} & Length: \OurMethod~& 1.498 (0.009) & 1.452 (0.008)  & 1.495 (0.008) & 1.457 (0.008) \\ 
   & Coverage: \OurMethod~& 0.854 (0.030) & 0.851 (0.039) & 0.848 (0.035) & 0.851 (0.042) \\
       & Length: Conformal & 7.911 (0.155)  & 2.124 (0.022)  & 1.500 (0.008) & 1.445 (0.007) \\
           & Coverage: Conformal & 0.843 (0.035)  &0.856 (0.036) & 0.848 (0.033) & 0.851 (0.042) \\
                  & Length: Subsampling & 8.451 (0.222) & 2.233 (0.053) & 1.551 (0.028) & 1.500 (0.029) \\
           & Coverage: Subsampling& 0.851 (0.033) & 0.846 (0.039)  & 0.844 (0.037)  & 0.852 (0.046) \\
           & Length: Single-Tree & 1.646 (0.041) & 1.662 (0.036) & 1.646 (0.029) & 1.607 (0.039) \\
           & Coverage: Single-Tree& 0.867 (0.032) & 0.869 (0.034) & 0.856 (0.029) & 0.854 (0.037) \\
    \hline

			\bottomrule
		\end{tabular}
	\end{center}
	\caption{Prediction set length and coverage probability of our method and three benchmarks for supervised learning; same protocol as in Table \ref{tab: fix_N_unsup}.} 
 	\label{tab:fix_N_sup}
\end{table}

{\bf Summary of results.}
We now provide insights into the observed results.
In both unsupervised and supervised learning contexts, when $\sigma^2$ is large
and and the distributions corresponding to different branches are more dispersed, 
our prediction sets remain close to having optimal length---e.g., 
for $\alpha=0.05$, 
$2\cdot 1/2\cdot 1.96 = 1.96$,
as determined by the normal quantiles.
However, under these circumstances,
the conformal prediction and subsampling
 baselines
tend to yield significantly wider and less informative prediction sets.
Moreover, when $\alpha=0.05$,
the single-tree baseline 
does not yield an informative prediction set---specifically, it returns the entire real line as the prediction set---due to the limited sample size within the branch.

Even when $\sigma=0$, 
our method
results in prediction sets of comparable length to those generated by standard conformal inference; 
since our method essentially interpolates between within-branch and global distributions.
Furthermore, since we use data from other branches for calibration, the standard error of the length of our prediction set is smaller. 

\subsection{Empirical Data Analysis}
\subsubsection{Two-layer Hierarchical Model}
First, we provide an empirical data analysis example for \OurMethod~in a two-layer hierarchical model by
analyzing a sleep deprivation dataset \citep{belenky2003patterns,balkin2000effects}, where 18 drivers' reaction times after $0,\ldots,9$ nights of three hours of sleep restriction
are recorded.
This dataset has also been investigated in the  related work by \cite{dunn2022distribution} studying two-layer hierarchical models, 
and we follow their approach to define the covariates and responses. 

Specifically, the response variable $Y$ is the sleep-deprived reaction time, the covariate $X_1$ represents the number of days of sleep deprivation, and the covariate $X_2$ denotes the baseline reaction time on day zero, 
with normal sleep. 
For each of the 18 individuals, 
we observe nine triplets $(X_{1,j}^{(k)},X_{2,j}^{(k)},Y_{j}^{(k)}),$ $j\in[9],k\in [18]$. 
In our analysis, 
we model these nine triplets as drawn independently from a distribution $P_{k},k\in [18]$. 
We discuss the modelling asumptions in Section \ref{das}.
 
Next, we discuss the experimental setting. We repeat our experiments over 100 independent trials. For each trial, we randomly split the data into training, calibration, and test 
datasets independently 500 times, as follows:
For each train-calibration-test split, we first randomly select two-thirds of the datapoints
from every branch (i.e., six observations) as training data,
to fit models $\hat\mu_k$ using linear regression, and obtain associated confidence bands $\hat\sigma_{k}$, $k\in [18]$. 
 We then pool all the training data from branches
 to fit a linear model $\hat\mu$.
 Next, we randomly select one datapoint from
the remaining $3\times 18=54$ as a test datapoint, and we use the other 53 datapoints as a calibration set. 

Following the same procedure as in the simulation studies, 
we record the averaged
coverage indicators and lengths of prediction sets over 500 test points.  
The box plots 
 of averaged prediction set lengths and coverage probabilities are presented for 
 $\alpha=0.10$ over the 100 independent trials. 
The results are shown in Figure \ref{fig:length}.
We obtain slightly more variable results compared with the simulation studies, because the variances within branches are large (therefore the data is more noisy), and we also have less calibration data. Additionally, analogous results obtained at a significance level of $\alpha=0.20$ are included in the appendix.

 We also compare the performance of our method on this dataset with the benchmarks from \Cref{sec:comp}, 
 including standard conformal prediction 
 by training only one model $\hat\mu$ that applies to every branch, 
 and subsampling \citep{dunn2022distribution}. 
The prediction coverage probabilities obtained by our method and 
standard conformal prediction
are less conservative than those obtained by subsampling, even though all methods have valid coverage.
We expect that the repeated subsampling method from 
 \cite{dunn2022distribution}
could improve stability, but will still be conservative. 
Moreover, our method leads to tighter intervals than the other methods.
These results reinforce the advantages of our method compared to alternative approaches.

\begin{figure}[ht]
	\centering
	\includegraphics[width=1.00\textwidth]{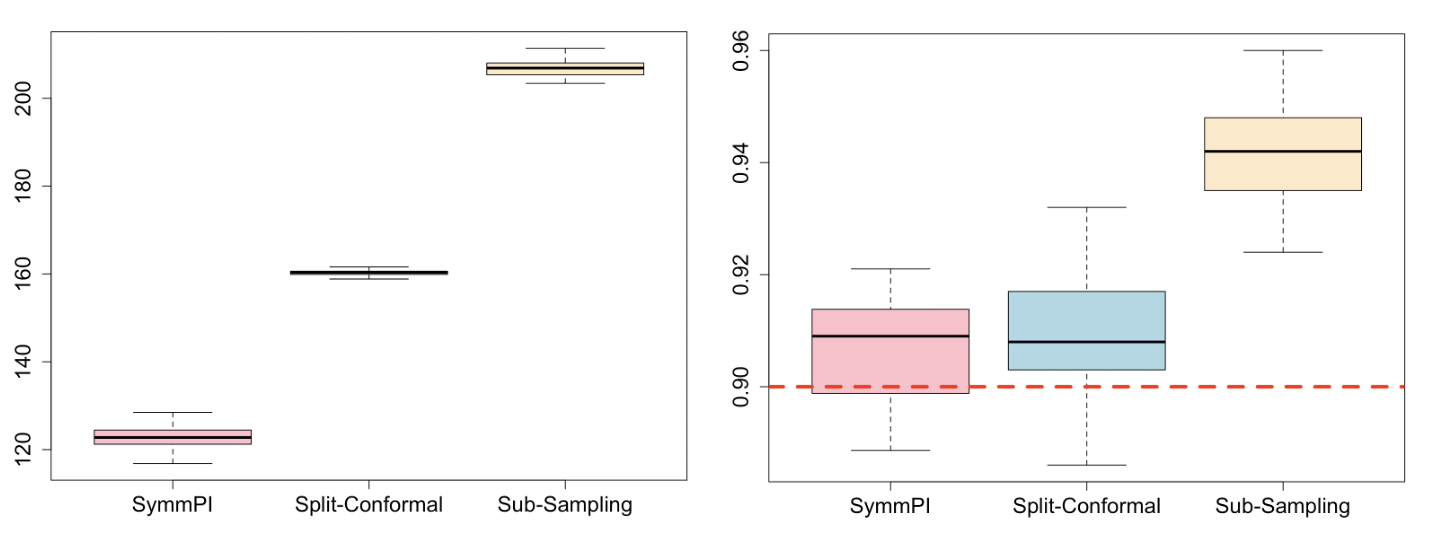}
	 \caption{Empirical data example: Prediction set lengths and coverage probabilities for various methods with level $\alpha=0.10$.  Left: Prediction set lengths; Right: Empirical coverage probabilities.}
           \label{fig:length}
\end{figure}

\subsubsection{Rotationally Invariant Image Classification}
{We present an additional empirical data example for \OurMethod~in a setting of rotationally invariant data. Specifically, we analyze the Multi-type Aircraft of Remote Sensing Images (MTARSI) dataset \citep{wu2020benchmark}, which comprises 9,385 remote sensing images showing top-down views of 20 aircraft types. These images have diverse backgrounds and varied poses. A visualization of some aircraft is shown in Figure \ref{fig_aircraft}.
  \begin{figure}[h]
  \centering
\includegraphics[width=0.65\textwidth]{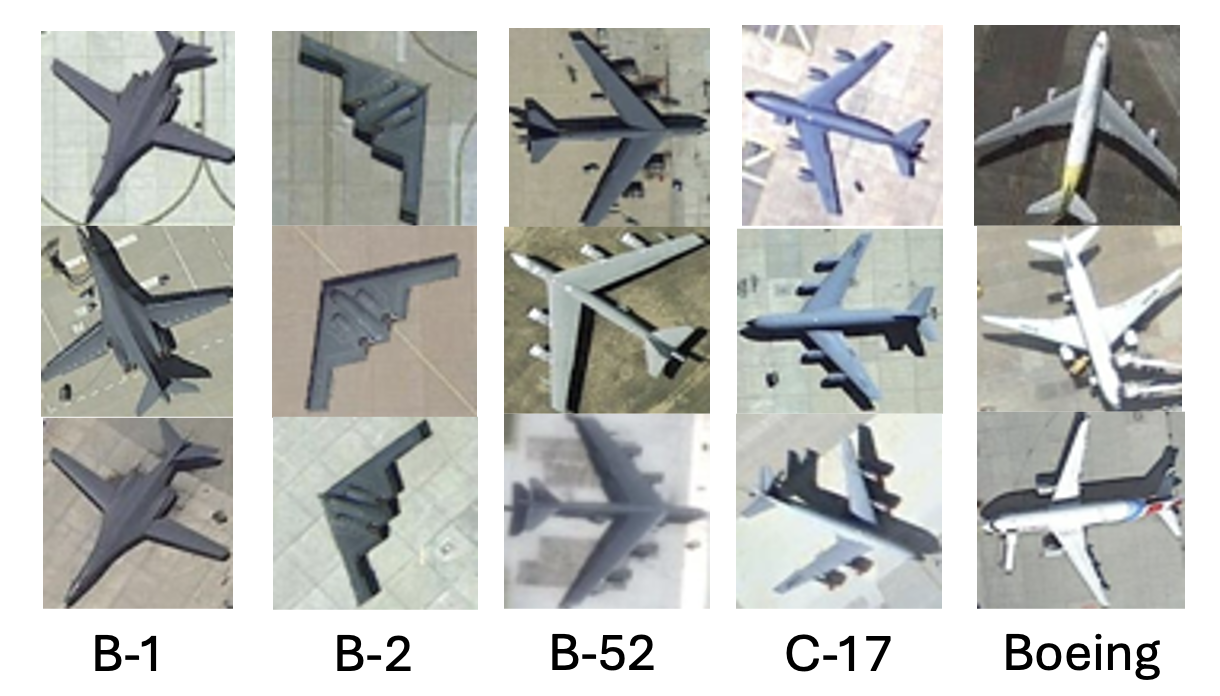}
\caption{\label{fig_aircraft} \small Illustration of aircraft types from the MTARSI dataset.}
\end{figure}}

{We aim to construct prediction sets for the labels of aircraft. To achieve this, we use a pre-trained 
ResNet18 \citep{he2016deep}
deep neural network model, which was trained using 3,000 images as described in \cite{mo2024ric}, serving as our classifier $\hat{p}:(y,x)\mapsto\hat{p}(y,x) \in [0,1]$, where $y$ denotes the label and $x$ represents the input features of the image. The remaining dataset is divided into two subsets: 1,000 images for the calibration data pool and the others for the test data pool.
The datapoints have approximate rotational distributional symmetry about the center of the image.}

{For each experiment, we randomly sample \(N\) images (with $N\in \{18, 36, 50, 80\}$ as shown in Table \ref{tab:data_aircraft}) from the entire calibration data pool as the calibration set. Additionally, we also randomly select 500 unlabeled test images from the test data pool.}

{For each test image, we use the method from \Cref{coex} to construct the prediction set, where we sample $2\times2$ orthogonal rotation matrices \(\mathcal{O} \sim U(O(2))\)---rotating the images around the center---20 times independently for the associated calibration and test images.
We record the averaged length and coverage of the prediction sets for $500$ test images, with $\alpha=0.05$ and $0.1$ respectively.
For comparison, we also construct standard conformal prediction sets without rotating the images. The overall testing procedure is further repeated 100 times independently, and we record the standard errors for both methods. The results are shown in Table \ref{tab:data_aircraft}.}

  \begin{table}[th]
	\begin{center}
		\def\arraystretch{1.2}
		\setlength\tabcolsep{4pt}
		\begin{tabular}{c||ccccc}
			\toprule
			&Method  & $N=18$ & $N=36$ & $N=50$  & $N=80$   \\
			\hline
   \multirow{4}{*}{$\alpha=0.05$} & Length: \OurMethod~& 2.901 (0.411) & 1.943 (0.149) &  1.329 (0.052)  & 1.221 (0.030) \\ 
   & Coverage: \OurMethod~& 0.972 (0.046)   & 0.959 (0.003)   &  0.958 (0.004) & 0.953 (0.003) \\
       & Length: Conformal & 20.00 (0.00)  & 3.004 (0.338)   & 1.951 (0.199)    & 1.260 (0.038) \\
           & Coverage: Conformal & 1.00 (0.00) &  0.977 (0.003)  & 0.964 (0.004)  & 0.954 (0.004) \\
    \hline
       \multirow{4}{*}{$\alpha=0.10$} & Length: \OurMethod~& 1.348 (0.067) & 1.057 (0.032)   & 1.004 (0.021)  &  0.979 (0.009) \\ 
   & Coverage: \OurMethod~& 0.930 (0.006)  & 0.914 (0.006)   &0.907 (0.007)  & 0.904 (0.006)  \\
       & Length: Conformal & 1.600 (0.129)   & 1.166 (0.053)   & 1.032 (0.040)  & 0.985 (0.013) \\
           & Coverage: Conformal & 0.940 (0.008)  & 0.925 (0.007) & 0.910 (0.009) & 0.905 (0.007) \\
    \hline

			\bottomrule
		\end{tabular}
	\end{center}
	\caption{Empirical data example (MTARSI dataset): Prediction set lengths and coverage probabilities for various methods with level $\alpha=0.05$ and $\alpha=0.10$, respectively. For each trial, we construct prediction sets using $N$ calibration images and average the prediction lengths and coverage indicator values over 500 independently sampled test images. The reported values represent the averaged length and coverage over 100 independent trials, with the associated standard errors provided in brackets. } 
	\label{tab:data_aircraft}
\end{table}

{From Table \ref{tab:data_aircraft}, we draw the following conclusions: Even when the size of the calibration sample is small \OurMethod~provides valid prediction sets with reasonable length and coverage. 
In contrast, in this case, standard conformal prediction can yield uninformative prediction sets (encompassing all labels).
Furthermore, even with larger calibration datasets for both methods, our approach consistently achieves more stable prediction intervals (somewhat smaller length and standard error) with comparable coverage probabilities compared to conformal prediction. }

{As a final remark, it is well-documented that classifiers trained on augmented datasets can achieve better classification accuracy compared to those trained on non-augmented datasets, see e.g., \citep{chen2020group} and references therein. 
Thus, if training with data augmentation, we expect that both  \OurMethod~and standard conformal prediction would achieve smaller prediction sets; but \OurMethod~would still have benefits during the calibration phase. 
In this experiment, to demonstrate performance, 
we focus on a pre-trained classifier.}

\section{Conclusion and Discussion}

We have presented a general methodology for predictive inference in arbitrary observation models satisfying distributional invariance.
We have illustrated that our methods have competitive performance in a two layer hierarchical model.
There are a number of intriguing directions for further research.
In some  examples, the data itself might not satisfy distributional invariance, but some transformation of the data---possibly dependent on unknown parameters---might do so.
Can one extend our methods to this setting, possibly leveraging ideas such as joint coverage regions \citep{dobriban2023joint}?
Moreover, is it possible to learn 
equivariant maps to enable improved predictive inference, as opposed to designing them as we did in this paper?
Studying these questions is expected to benefit the broad applicability of rigorous predictive inference methods.

\section*{Acknowledgements}
This work was supported in part by ARO W911NF-20-1-0080, ARO W911NF-23-1-0296,
NSF 2031895,
NSF 2046874, ONR N00014-21-1-2843, and the Sloan Foundation. 
We thank 
Yonghoon Lee,
Xiao Ma,
Matteo Sesia, 
Vladimir Vovk,
Yao Xie,
Sheng Xu, and Yuling Yan
for helpful discussion and feedback on earlier versions of the manuscript. We also thank Hanlin Mo for kindly providing their pretrained models and datasets.

{\small
\setlength{\bibsep}{0.2pt plus 0.3ex}
\bibliographystyle{plainnat-abbrev}
\bibliography{references}

\begin{thebibliography}{102}
\providecommand{\natexlab}[1]{#1}
\providecommand{\url}[1]{\texttt{#1}}
\expandafter\ifx\csname urlstyle\endcsname\relax
  \providecommand{\doi}[1]{doi: #1}\else
  \providecommand{\doi}{doi: \begingroup \urlstyle{rm}\Url}\fi

\bibitem[Anderson and Robinson(2001)]{anderson2001permutation}
M.~J. Anderson and J.~Robinson.
\newblock Permutation tests for linear models.
\newblock \emph{Australian \& New Zealand Journal of Statistics}, 43\penalty0
  (1):\penalty0 75--88, 2001.

\bibitem[Angelopoulos et~al.(2022)Angelopoulos, Bates, Fisch, Lei, and
  Schuster]{angelopoulos2022conformal}
A.~N. Angelopoulos, S.~Bates, A.~Fisch, L.~Lei, and T.~Schuster.
\newblock Conformal risk control.
\newblock \emph{arXiv preprint arXiv:2208.02814}, 2022.

\bibitem[Angelopoulos et~al.(2023)Angelopoulos, Bates,
  et~al.]{angelopoulos2021gentle}
A.~N. Angelopoulos, S.~Bates, et~al.
\newblock Conformal prediction: A gentle introduction.
\newblock \emph{Foundations and Trends{\textregistered} in Machine Learning},
  16\penalty0 (4):\penalty0 494--591, 2023.

\bibitem[Artin(2018)]{artin2018algebra}
M.~Artin.
\newblock \emph{Algebra}.
\newblock Pearson, 2018.

\bibitem[Babai(2016)]{babai2016graph}
L.~Babai.
\newblock Graph isomorphism in quasipolynomial time.
\newblock In \emph{Proceedings of the forty-eighth annual ACM symposium on
  Theory of Computing}, pages 684--697, 2016.

\bibitem[Balkin et~al.(2000)Balkin, Thome, Sing, Thomas, Redmond, Wesensten,
  Williams, Hall, Belenky, et~al.]{balkin2000effects}
T.~Balkin, D.~Thome, H.~Sing, M.~Thomas, D.~Redmond, N.~Wesensten, J.~Williams,
  S.~Hall, G.~Belenky, et~al.
\newblock Effects of sleep schedules on commercial motor vehicle driver
  performance.
\newblock Technical report, United States. Department of Transportation.
  Federal Motor Carrier Safety Safety Administration, 2000.

\bibitem[Barber et~al.(2023)Barber, Candes, Ramdas, and
  Tibshirani]{barber2023conformal}
R.~F. Barber, E.~J. Candes, A.~Ramdas, and R.~J. Tibshirani.
\newblock Conformal prediction beyond exchangeability.
\newblock \emph{The Annals of Statistics}, 51\penalty0 (2):\penalty0 816--845,
  2023.

\bibitem[Bates et~al.(2021)Bates, Angelopoulos, Lei, Malik, and
  Jordan]{bates2021distribution}
S.~Bates, A.~Angelopoulos, L.~Lei, J.~Malik, and M.~Jordan.
\newblock Distribution-free, risk-controlling prediction sets.
\newblock \emph{Journal of the ACM (JACM)}, 68\penalty0 (6):\penalty0 1--34,
  2021.

\bibitem[Bates et~al.(2023)Bates, Cand{\`e}s, Lei, Romano, and
  Sesia]{bates2023testing}
S.~Bates, E.~Cand{\`e}s, L.~Lei, Y.~Romano, and M.~Sesia.
\newblock Testing for outliers with conformal p-values.
\newblock \emph{The Annals of Statistics}, 51\penalty0 (1):\penalty0 149--178,
  2023.

\bibitem[Belenky et~al.(2003)Belenky, Wesensten, Thorne, Thomas, Sing, Redmond,
  Russo, and Balkin]{belenky2003patterns}
G.~Belenky, N.~J. Wesensten, D.~R. Thorne, M.~L. Thomas, H.~C. Sing, D.~P.
  Redmond, M.~B. Russo, and T.~J. Balkin.
\newblock Patterns of performance degradation and restoration during sleep
  restriction and subsequent recovery: A sleep dose-response study.
\newblock \emph{Journal of sleep research}, 12\penalty0 (1):\penalty0 1--12,
  2003.

\bibitem[Blum-Smith and Villar(2022)]{blum2022equivariant}
B.~Blum-Smith and S.~Villar.
\newblock Equivariant maps from invariant functions.
\newblock \emph{arXiv preprint arXiv:2209.14991}, 2022.

\bibitem[Cand{\`e}s et~al.(2023)Cand{\`e}s, Lei, and
  Ren]{candes2023conformalized}
E.~Cand{\`e}s, L.~Lei, and Z.~Ren.
\newblock Conformalized survival analysis.
\newblock \emph{Journal of the Royal Statistical Society Series B: Statistical
  Methodology}, 85\penalty0 (1):\penalty0 24--45, 2023.

\bibitem[Chatzipantazis et~al.(2023)Chatzipantazis, Pertigkiozoglou,
  Daniilidis, and Dobriban]{chatzipantazis2021learning}
E.~Chatzipantazis, S.~Pertigkiozoglou, K.~Daniilidis, and E.~Dobriban.
\newblock Learning augmentation distributions using transformed risk
  minimization.
\newblock \emph{Transactions on Machine Learning Research}, 2023.
\newblock URL \url{https://openreview.net/forum?id=LRYtNj8Xw0}.

\bibitem[Chen et~al.(2020)Chen, Dobriban, and Lee]{chen2020group}
S.~Chen, E.~Dobriban, and J.~H. Lee.
\newblock A group-theoretic framework for data augmentation.
\newblock \emph{Journal of Machine Learning Research}, 21\penalty0
  (245):\penalty0 1--71, 2020.

\bibitem[Chernozhukov et~al.(2018)Chernozhukov, Wuthrich, and
  Zhu]{Chernozhukov2018}
V.~Chernozhukov, K.~Wuthrich, and Y.~Zhu.
\newblock {Exact and Robust Conformal Inference Methods for Predictive Machine
  Learning With Dependent Data}.
\newblock In \emph{Proceedings of the 31st Conference On Learning Theory},
  2018.

\bibitem[Cohen and Welling(2016)]{cohen2016group}
T.~Cohen and M.~Welling.
\newblock Group equivariant convolutional networks.
\newblock In \emph{{I}nternational {C}onference on {M}achine {L}earning}, 2016.

\bibitem[Cox(2006)]{cox2006principles}
D.~R. Cox.
\newblock \emph{Principles of statistical inference}.
\newblock Cambridge University Press, 2006.

\bibitem[Dean and Verducci(1990)]{dean1990linear}
A.~Dean and J.~Verducci.
\newblock Linear transformations that preserve majorization, schur concavity,
  and exchangeability.
\newblock \emph{Linear Algebra and its Applications}, 127:\penalty0 121--138,
  1990.

\bibitem[Diaconis(1988)]{diaconis1988group}
P.~Diaconis.
\newblock \emph{Group Representations in Probability and Statistics}.
\newblock Institute of Mathematical Statistics, 1988.

\bibitem[Diestel and Spalsbury(2014)]{diestel2014joys}
J.~Diestel and A.~Spalsbury.
\newblock \emph{The joys of Haar measure}.
\newblock American Mathematical Society, 2014.

\bibitem[Dobriban(2022)]{dobriban2021consistency}
E.~Dobriban.
\newblock Consistency of invariance-based randomization tests.
\newblock \emph{The Annals of Statistics}, 50\penalty0 (4):\penalty0 2443 --
  2466, 2022.

\bibitem[Dobriban and Lin(2023)]{dobriban2023joint}
E.~Dobriban and Z.~Lin.
\newblock Joint coverage regions: Simultaneous confidence and prediction sets.
\newblock \emph{arXiv preprint arXiv:2303.00203}, 2023.

\bibitem[Dunn et~al.(2022)Dunn, Wasserman, and Ramdas]{dunn2022distribution}
R.~Dunn, L.~Wasserman, and A.~Ramdas.
\newblock Distribution-free prediction sets for two-layer hierarchical models.
\newblock \emph{Journal of the American Statistical Association}, pages 1--12,
  2022.

\bibitem[Eaton(1989)]{eaton1989group}
M.~L. Eaton.
\newblock Group invariance applications in statistics.
\newblock In \emph{Regional conference series in Probability and Statistics},
  1989.

\bibitem[Eden and Yates(1933)]{eden1933validity}
T.~Eden and F.~Yates.
\newblock On the validity of {F}isher's z test when applied to an actual
  example of non-normal data.
\newblock \emph{The Journal of Agricultural Science}, 23\penalty0 (1):\penalty0
  6--17, 1933.

\bibitem[Einbinder et~al.(2022)Einbinder, Romano, Sesia, and
  Zhou]{einbinder2022training}
B.-S. Einbinder, Y.~Romano, M.~Sesia, and Y.~Zhou.
\newblock Training uncertainty-aware classifiers with conformalized deep
  learning.
\newblock \emph{{A}dvances in {N}eural {I}nformation {P}rocessing {S}ystems},
  2022.

\bibitem[Ernst(2004)]{ernst2004permutation}
M.~D. Ernst.
\newblock Permutation methods: a basis for exact inference.
\newblock \emph{Statistical Science}, 19\penalty0 (4):\penalty0 676--685, 2004.

\bibitem[Finzi et~al.(2020)Finzi, Stanton, Izmailov, and
  Wilson]{pmlr-v119-finzi20a}
M.~Finzi, S.~Stanton, P.~Izmailov, and A.~G. Wilson.
\newblock Generalizing convolutional neural networks for equivariance to lie
  groups on arbitrary continuous data.
\newblock In \emph{Proceedings of the 37th {I}nternational {C}onference on
  {M}achine {L}earning}, 2020.

\bibitem[Fisch et~al.(2021)Fisch, Schuster, Jaakkola, and
  Barzilay]{fisch2021few}
A.~Fisch, T.~Schuster, T.~Jaakkola, and R.~Barzilay.
\newblock Few-shot conformal prediction with auxiliary tasks.
\newblock In \emph{{I}nternational {C}onference on {M}achine {L}earning}, 2021.

\bibitem[Fisher(1935)]{fisher1935design}
R.~A. Fisher.
\newblock \emph{The design of experiments}.
\newblock Oliver and Boyd, 1935.

\bibitem[Folland(2016)]{folland2016course}
G.~B. Folland.
\newblock \emph{A course in abstract harmonic analysis}.
\newblock CRC Press, 2016.

\bibitem[Fukushima(1980)]{fukushima1980neocognitron}
K.~Fukushima.
\newblock Neocognitron: A self-organizing neural network model for a mechanism
  of pattern recognition unaffected by shift in position.
\newblock \emph{Biological cybernetics}, 36\penalty0 (4):\penalty0 193--202,
  1980.

\bibitem[Fulton and Harris(2013)]{fulton2013representation}
W.~Fulton and J.~Harris.
\newblock \emph{Representation theory: a first course}, volume 129.
\newblock Springer Science \& Business Media, 2013.

\bibitem[Gammerman et~al.(1998)Gammerman, Vovk, and
  Vapnik]{gammerman1998learning}
A.~Gammerman, V.~Vovk, and V.~Vapnik.
\newblock Learning by transduction.
\newblock In \emph{Proceedings of the Fourteenth conference on Uncertainty in
  artificial intelligence}, 1998.

\bibitem[Geisser(2017)]{geisser2017predictive}
S.~Geisser.
\newblock \emph{Predictive inference: an introduction}.
\newblock Chapman and Hall/CRC, 2017.

\bibitem[Gilmer et~al.(2017)Gilmer, Schoenholz, Riley, Vinyals, and
  Dahl]{gilmer2017neural}
J.~Gilmer, S.~S. Schoenholz, P.~F. Riley, O.~Vinyals, and G.~E. Dahl.
\newblock Neural message passing for quantum chemistry.
\newblock In \emph{{I}nternational {C}onference on {M}achine {L}earning}, 2017.

\bibitem[Giri(1996)]{giri1996group}
N.~C. Giri.
\newblock \emph{Group invariance in statistical inference}.
\newblock World Scientific, 1996.

\bibitem[Good(2006)]{good2006permutation}
P.~I. Good.
\newblock \emph{Permutation, parametric, and bootstrap tests of hypotheses}.
\newblock Springer Science \& Business Media, 2006.

\bibitem[Gross(1996)]{gross1996role}
D.~J. Gross.
\newblock The role of symmetry in fundamental physics.
\newblock \emph{Proceedings of the National Academy of Sciences}, 93\penalty0
  (25):\penalty0 14256--14259, 1996.

\bibitem[Guan(2023{\natexlab{a}})]{guan2023conformal}
L.~Guan.
\newblock A conformal test of linear models via permutation-augmented
  regressions.
\newblock \emph{arXiv preprint arXiv:2309.05482}, 2023{\natexlab{a}}.

\bibitem[Guan(2023{\natexlab{b}})]{guan2023localized}
L.~Guan.
\newblock Localized conformal prediction: A generalized inference framework for
  conformal prediction.
\newblock \emph{Biometrika}, 110\penalty0 (1):\penalty0 33--50,
  2023{\natexlab{b}}.

\bibitem[Guan and Tibshirani(2022)]{guan2022prediction}
L.~Guan and R.~Tibshirani.
\newblock Prediction and outlier detection in classification problems.
\newblock \emph{Journal of the Royal Statistical Society Series B: Statistical
  Methodology}, 84\penalty0 (2):\penalty0 524--546, 2022.

\bibitem[Gui et~al.(2023+)Gui, Hore, Ren, and Barber]{gui2022conformalized}
Y.~Gui, R.~Hore, Z.~Ren, and R.~F. Barber.
\newblock Conformalized survival analysis with adaptive cutoffs.
\newblock \emph{Biometrika}, 2023+.

\bibitem[He et~al.(2016)He, Zhang, Ren, and Sun]{he2016deep}
K.~He, X.~Zhang, S.~Ren, and J.~Sun.
\newblock Deep residual learning for image recognition.
\newblock In \emph{Proceedings of the IEEE conference on computer vision and
  pattern recognition}, pages 770--778, 2016.

\bibitem[Hemerik and Goeman(2018)]{hemerik2018exact}
J.~Hemerik and J.~Goeman.
\newblock Exact testing with random permutations.
\newblock \emph{Test}, 27\penalty0 (4):\penalty0 811--825, 2018.

\bibitem[Hoeffding(1952)]{hoeffding1952large}
W.~Hoeffding.
\newblock The large-sample power of tests based on permutations of
  observations.
\newblock \emph{The Annals of Mathematical Statistics}, pages 169--192, 1952.

\bibitem[Kaur et~al.(2022)Kaur, Jha, Roy, Park, Dobriban, Sokolsky, and
  Lee]{kaur2022idecode}
R.~Kaur, S.~Jha, A.~Roy, S.~Park, E.~Dobriban, O.~Sokolsky, and I.~Lee.
\newblock idecode: In-distribution equivariance for conformal
  out-of-distribution detection.
\newblock In \emph{Proceedings of the AAAI Conference on Artificial
  Intelligence}, 2022.

\bibitem[Kennedy(1995)]{kennedy1995randomization}
F.~E. Kennedy.
\newblock Randomization tests in econometrics.
\newblock \emph{Journal of Business \& Economic Statistics}, 13\penalty0
  (1):\penalty0 85--94, 1995.

\bibitem[Kuchibhotla(2020)]{kuchibhotla2020exchangeability}
A.~K. Kuchibhotla.
\newblock Exchangeability, conformal prediction, and rank tests.
\newblock \emph{arXiv preprint arXiv:2005.06095}, 2020.

\bibitem[LeCun et~al.(1989)LeCun, Boser, Denker, Henderson, Howard, Hubbard,
  and Jackel]{lecun1989backpropagation}
Y.~LeCun, B.~Boser, J.~S. Denker, D.~Henderson, R.~E. Howard, W.~Hubbard, and
  L.~D. Jackel.
\newblock Backpropagation applied to handwritten zip code recognition.
\newblock \emph{Neural computation}, 1\penalty0 (4):\penalty0 541--551, 1989.

\bibitem[Lee et~al.(2023)Lee, Barber, and Willett]{lee2023distribution}
Y.~Lee, R.~F. Barber, and R.~Willett.
\newblock Distribution-free inference with hierarchical data.
\newblock \emph{arXiv preprint arXiv:2306.06342}, 2023.

\bibitem[Lehmann and Stein(1949)]{lehmann1949theory}
E.~L. Lehmann and C.~Stein.
\newblock On the theory of some non-parametric hypotheses.
\newblock \emph{The Annals of Mathematical Statistics}, 20\penalty0
  (1):\penalty0 28--45, 1949.

\bibitem[Lei(2014)]{lei2014classification}
J.~Lei.
\newblock Classification with confidence.
\newblock \emph{Biometrika}, 101\penalty0 (4):\penalty0 755--769, 2014.

\bibitem[Lei and Wasserman(2014)]{lei2014distribution}
J.~Lei and L.~Wasserman.
\newblock Distribution-free prediction bands for non-parametric regression.
\newblock \emph{Journal of the Royal Statistical Society: Series B (Statistical
  Methodology)}, 76\penalty0 (1):\penalty0 71--96, 2014.

\bibitem[Lei et~al.(2013)Lei, Robins, and Wasserman]{lei2013distribution}
J.~Lei, J.~Robins, and L.~Wasserman.
\newblock Distribution-free prediction sets.
\newblock \emph{Journal of the American Statistical Association}, 108\penalty0
  (501):\penalty0 278--287, 2013.

\bibitem[Lei et~al.(2015)Lei, Rinaldo, and Wasserman]{lei2015conformal}
J.~Lei, A.~Rinaldo, and L.~Wasserman.
\newblock A conformal prediction approach to explore functional data.
\newblock \emph{Annals of Mathematics and Artificial Intelligence}, 74\penalty0
  (1):\penalty0 29--43, 2015.

\bibitem[Lei et~al.(2018)Lei, G’Sell, Rinaldo, Tibshirani, and
  Wasserman]{lei2018distribution}
J.~Lei, M.~G’Sell, A.~Rinaldo, R.~J. Tibshirani, and L.~Wasserman.
\newblock Distribution-free predictive inference for regression.
\newblock \emph{Journal of the American Statistical Association}, 113\penalty0
  (523):\penalty0 1094--1111, 2018.

\bibitem[Li et~al.(2022)Li, Ji, Dobriban, Sokolsky, and Lee]{li2022pac}
S.~Li, X.~Ji, E.~Dobriban, O.~Sokolsky, and I.~Lee.
\newblock Pac-wrap: Semi-supervised pac anomaly detection.
\newblock In \emph{Proceedings of the 28th ACM SIGKDD Conference on Knowledge
  Discovery and Data Mining}, 2022.

\bibitem[Liang et~al.(2022)Liang, Sesia, and Sun]{liang2022integrative}
Z.~Liang, M.~Sesia, and W.~Sun.
\newblock Integrative conformal p-values for powerful out-of-distribution
  testing with labeled outliers.
\newblock \emph{arXiv preprint arXiv:2208.11111}, 2022.

\bibitem[Liang et~al.(2023)Liang, Zhou, and Sesia]{liang2023conformal}
Z.~Liang, Y.~Zhou, and M.~Sesia.
\newblock Conformal inference is (almost) free for neural networks trained with
  early stopping.
\newblock In \emph{{I}nternational {C}onference on {M}achine {L}earning}, 2023.

\bibitem[Lunde et~al.(2023)Lunde, Levina, and Zhu]{lunde2023conformal}
R.~Lunde, E.~Levina, and J.~Zhu.
\newblock Conformal prediction for network-assisted regression.
\newblock \emph{arXiv preprint arXiv:2302.10095}, 2023.

\bibitem[Lyle et~al.(2019)Lyle, Kwiatkowksa, and Gal]{lyle2016analysis}
C.~Lyle, M.~Kwiatkowksa, and Y.~Gal.
\newblock An analysis of the effect of invariance on generalization in neural
  networks.
\newblock In \emph{{I}nternational {C}onference on {M}achine {L}earning
  Workshop on Understanding and Improving Generalization in Deep Learning},
  2019.

\bibitem[Mo and Zhao(2024)]{mo2024ric}
H.~Mo and G.~Zhao.
\newblock Ric-cnn: rotation-invariant coordinate convolutional neural network.
\newblock \emph{Pattern Recognition}, 146:\penalty0 109994, 2024.

\bibitem[Munkres(2019)]{munkres2019topology}
J.~R. Munkres.
\newblock \emph{Topology}.
\newblock Pearson Education, 2019.

\bibitem[Nachbin(1976)]{nachbin1976haar}
L.~Nachbin.
\newblock \emph{The Haar Integral}.
\newblock R. E. Krieger Publishing Company, 1976.

\bibitem[Papadopoulos et~al.(2002)Papadopoulos, Proedrou, Vovk, and
  Gammerman]{papadopoulos2002inductive}
H.~Papadopoulos, K.~Proedrou, V.~Vovk, and A.~Gammerman.
\newblock Inductive confidence machines for regression.
\newblock In \emph{European {C}onference on {M}achine {L}earning}. Springer,
  2002.

\bibitem[Park et~al.(2020)Park, Bastani, Matni, and Lee]{Park2020pac}
S.~Park, O.~Bastani, N.~Matni, and I.~Lee.
\newblock {PAC} confidence sets for deep neural networks via calibrated
  prediction.
\newblock In \emph{International Conference on Learning Representations}, 2020.

\bibitem[Park et~al.(2022{\natexlab{a}})Park, Dobriban, Lee, and
  Bastani]{park2021pac}
S.~Park, E.~Dobriban, I.~Lee, and O.~Bastani.
\newblock {PAC} prediction sets under covariate shift.
\newblock In \emph{International Conference on Learning Representations},
  2022{\natexlab{a}}.

\bibitem[Park et~al.(2022{\natexlab{b}})Park, Dobriban, Lee, and
  Bastani]{park2022pac}
S.~Park, E.~Dobriban, I.~Lee, and O.~Bastani.
\newblock {PAC} prediction sets for meta-learning.
\newblock In \emph{{A}dvances in {N}eural {I}nformation {P}rocessing
  {S}ystems}, 2022{\natexlab{b}}.

\bibitem[Pesarin(2001)]{pesarin2001multivariate}
F.~Pesarin.
\newblock \emph{Multivariate permutation tests: with applications in
  biostatistics}.
\newblock Wiley, 2001.

\bibitem[Pesarin and Salmaso(2010)]{pesarin2010permutation}
F.~Pesarin and L.~Salmaso.
\newblock \emph{Permutation tests for complex data: theory, applications and
  software}.
\newblock John Wiley \& Sons, 2010.

\bibitem[Pesarin and Salmaso(2012)]{pesarin2012review}
F.~Pesarin and L.~Salmaso.
\newblock A review and some new results on permutation testing for multivariate
  problems.
\newblock \emph{Statistics and Computing}, 22\penalty0 (2):\penalty0 639--646,
  2012.

\bibitem[Pitman(1937)]{pitman1937significance}
E.~J. Pitman.
\newblock Significance tests which may be applied to samples from any
  populations.
\newblock \emph{Supplement to the Journal of the Royal Statistical Society},
  4\penalty0 (1):\penalty0 119--130, 1937.

\bibitem[Qiu et~al.(2022)Qiu, Dobriban, and Tchetgen]{qiu2022distribution}
H.~Qiu, E.~Dobriban, and E.~T. Tchetgen.
\newblock Prediction sets adaptive to unknown covariate shift.
\newblock \emph{Journal of the Royal Statistical Society: Series B (to appear),
  arXiv preprint arXiv:2203.06126}, 2022.

\bibitem[Ramdas et~al.(2023)Ramdas, Barber, Cand{\`e}s, and
  Tibshirani]{ramdas2023permutation}
A.~Ramdas, R.~F. Barber, E.~J. Cand{\`e}s, and R.~J. Tibshirani.
\newblock Permutation tests using arbitrary permutation distributions.
\newblock \emph{Sankhya A}, 85\penalty0 (2):\penalty0 1156--1177, 2023.

\bibitem[Robinson(2011)]{robinson2011symmetry}
M.~Robinson.
\newblock \emph{Symmetry and the standard model}.
\newblock Springer, 2011.

\bibitem[Romano et~al.(2019)Romano, Patterson, and
  Candes]{romano2019conformalized}
Y.~Romano, E.~Patterson, and E.~Candes.
\newblock Conformalized quantile regression.
\newblock In \emph{{A}dvances in {N}eural {I}nformation {P}rocessing
  {S}ystems}, 2019.

\bibitem[Romano et~al.(2020)Romano, Sesia, and
  Candes]{romano2020classification}
Y.~Romano, M.~Sesia, and E.~Candes.
\newblock Classification with valid and adaptive coverage.
\newblock \emph{{A}dvances in {N}eural {I}nformation {P}rocessing {S}ystems},
  2020.

\bibitem[Sadinle et~al.(2019)Sadinle, Lei, and Wasserman]{Sadinle2019}
M.~Sadinle, J.~Lei, and L.~Wasserman.
\newblock {Least Ambiguous Set-Valued Classifiers With Bounded Error Levels}.
\newblock \emph{Journal of the American Statistical Association}, 114\penalty0
  (525):\penalty0 223--234, 2019.

\bibitem[Saunders et~al.(1999)Saunders, Gammerman, and
  Vovk]{saunders1999transduction}
C.~Saunders, A.~Gammerman, and V.~Vovk.
\newblock Transduction with confidence and credibility.
\newblock In \emph{IJCAI}, 1999.

\bibitem[Scheffe and Tukey(1945)]{scheffe1945non}
H.~Scheffe and J.~W. Tukey.
\newblock Non-parametric estimation. {I}. {V}alidation of order statistics.
\newblock \emph{The Annals of Mathematical Statistics}, 16\penalty0
  (2):\penalty0 187--192, 1945.

\bibitem[Schwichtenberg(2018)]{schwichtenberg2018physics}
J.~Schwichtenberg.
\newblock \emph{Physics from symmetry}.
\newblock Springer, 2018.

\bibitem[Sesia et~al.(2023{\natexlab{a}})Sesia, Favaro, and
  Dobriban]{sesia2022conformal}
M.~Sesia, S.~Favaro, and E.~Dobriban.
\newblock Conformal frequency estimation using discrete sketched data with
  coverage for distinct queries.
\newblock \emph{Journal of Machine Learning Research}, 24\penalty0
  (348):\penalty0 1--80, 2023{\natexlab{a}}.

\bibitem[Sesia et~al.(2023{\natexlab{b}})Sesia, Favaro, and
  Dobriban]{sesia2023conformal}
M.~Sesia, S.~Favaro, and E.~Dobriban.
\newblock Conformal frequency estimation using discrete sketched data with
  coverage for distinct queries.
\newblock \emph{Journal of Machine Learning Research}, 24\penalty0
  (348):\penalty0 1--80, 2023{\natexlab{b}}.

\bibitem[Si et~al.(2023)Si, Park, Lee, Dobriban, and Bastani]{si2023pac}
W.~Si, S.~Park, I.~Lee, E.~Dobriban, and O.~Bastani.
\newblock {PAC} prediction sets under label shift.
\newblock \emph{arXiv preprint arXiv:2310.12964}, 2023.

\bibitem[Tibshirani et~al.(2019)Tibshirani, Foygel~Barber, Candes, and
  Ramdas]{tibshirani2019conformal}
R.~J. Tibshirani, R.~Foygel~Barber, E.~Candes, and A.~Ramdas.
\newblock Conformal prediction under covariate shift.
\newblock \emph{{A}dvances in {N}eural {I}nformation {P}rocessing {S}ystems},
  32, 2019.

\bibitem[Tornier(2020)]{tornier2020haar}
S.~Tornier.
\newblock Haar measures.
\newblock \emph{arXiv preprint arXiv:2006.10956}, 2020.

\bibitem[Tukey(1947)]{tukey1947non}
J.~W. Tukey.
\newblock Non-parametric estimation {II}. {S}tatistically equivalent blocks and
  tolerance regions--the continuous case.
\newblock \emph{The Annals of Mathematical Statistics}, 18\penalty0
  (4):\penalty0 529--539, 1947.

\bibitem[Tukey(1948)]{tukey1948nonparametric}
J.~W. Tukey.
\newblock Nonparametric estimation, {III}. {S}tatistically equivalent blocks
  and multivariate tolerance regions--the discontinuous case.
\newblock \emph{The Annals of Mathematical Statistics}, 19\penalty0
  (1):\penalty0 30--39, 1948.

\bibitem[Villar et~al.(2021)Villar, Hogg, Storey-Fisher, Yao, and
  Blum-Smith]{villar2021scalars}
S.~Villar, D.~W. Hogg, K.~Storey-Fisher, W.~Yao, and B.~Blum-Smith.
\newblock Scalars are universal: Equivariant machine learning, structured like
  classical physics.
\newblock \emph{{A}dvances in {N}eural {I}nformation {P}rocessing {S}ystems},
  34:\penalty0 28848--28863, 2021.

\bibitem[Villar et~al.(2023)Villar, Yao, Hogg, Blum-Smith, and
  Dumitrascu]{villar2022dimensionless}
S.~Villar, W.~Yao, D.~W. Hogg, B.~Blum-Smith, and B.~Dumitrascu.
\newblock Dimensionless machine learning: Imposing exact units equivariance.
\newblock \emph{Journal of Machine Learning Research}, 24\penalty0
  (109):\penalty0 1--32, 2023.

\bibitem[Vovk(2013)]{Vovk2013}
V.~Vovk.
\newblock {Conditional validity of inductive conformal predictors}.
\newblock In \emph{Asian {C}onference on {M}achine {L}earning}, 2013.

\bibitem[Vovk et~al.(2005)Vovk, Gammerman, and Shafer]{vovk2005algorithmic}
V.~Vovk, A.~Gammerman, and G.~Shafer.
\newblock \emph{Algorithmic learning in a random world}.
\newblock Springer Science \& Business Media, 2005.

\bibitem[Vovk et~al.(2022)Vovk, Gammerman, and Shafer]{vovk2022algorithmic}
V.~Vovk, A.~Gammerman, and G.~Shafer.
\newblock \emph{Algorithmic Learning in a Random World}.
\newblock Springer Nature, 2022.

\bibitem[Vovk et~al.(1999)Vovk, Gammerman, and Saunders]{vovk1999machine}
V.~Vovk, A.~Gammerman, and C.~Saunders.
\newblock Machine-learning applications of algorithmic randomness.
\newblock In \emph{{I}nternational {C}onference on {M}achine {L}earning}, 1999.

\bibitem[Wald(1943)]{Wald1943}
A.~Wald.
\newblock {An Extension of Wilks' Method for Setting Tolerance Limits}.
\newblock \emph{The Annals of Mathematical Statistics}, 14\penalty0
  (1):\penalty0 45--55, 1943.

\bibitem[Wang et~al.(2024)Wang, Li, and Yu]{wang2024conformal}
B.~Wang, F.~Li, and M.~Yu.
\newblock Conformal causal inference for cluster randomized trials:
  model-robust inference without asymptotic approximations.
\newblock \emph{arXiv preprint arXiv:2401.01977}, 2024.

\bibitem[Weiler et~al.(2018)Weiler, Hamprecht, and Storath]{WeilerHS18}
M.~Weiler, F.~A. Hamprecht, and M.~Storath.
\newblock Learning steerable filters for rotation equivariant cnns.
\newblock In \emph{2018 {IEEE} Conference on Computer Vision and Pattern
  Recognition}, pages 849--858. {IEEE} Computer Society, 2018.

\bibitem[Wijsman(1990)]{wijsman1990invariant}
R.~A. Wijsman.
\newblock \emph{Invariant measures on groups and their use in statistics}.
\newblock IMS, 1990.

\bibitem[Wilks(1941)]{Wilks1941}
S.~S. Wilks.
\newblock {Determination of Sample Sizes for Setting Tolerance Limits}.
\newblock \emph{The Annals of Mathematical Statistics}, 12\penalty0
  (1):\penalty0 91--96, 1941.

\bibitem[Wu et~al.(2020)Wu, Wan, Wang, Tan, Zou, Li, and Chen]{wu2020benchmark}
Z.-Z. Wu, S.-H. Wan, X.-F. Wang, M.~Tan, L.~Zou, X.-L. Li, and Y.~Chen.
\newblock A benchmark data set for aircraft type recognition from remote
  sensing images.
\newblock \emph{Applied Soft Computing}, 89:\penalty0 106132, 2020.

\bibitem[Xu et~al.(2018)Xu, Li, Tian, Sonobe, Kawarabayashi, and
  Jegelka]{xu2018representation}
K.~Xu, C.~Li, Y.~Tian, T.~Sonobe, K.-i. Kawarabayashi, and S.~Jegelka.
\newblock Representation learning on graphs with jumping knowledge networks.
\newblock In \emph{{I}nternational {C}onference on {M}achine {L}earning}, 2018.

\end{thebibliography}
}

\section{Supplementary Material}
\label{supp}
{\bf Additional notation and definitions.}
In the supplementary material we will use the following additional notation and definitions.
If two sets $A,B$ are in a bijection, we write $A\cong B$.
For a group $\mG$ acting on a set $\mZ$ by an action $\rho$, 
the stabilizer of an element $z\in \mZ$ is the set $\{g\in \mG: \rho(g)z = z\}$.
The stabilizer is a subgroup of $\mG$.
For two sets $S,S_0$, a map $V:S\to S_0$, and a subset $S' \subseteq S_0$, we denote by $V^{-1}(S')$ the preimage of $S'$ under $V$.

\subsection{Discussion of Equivariance Properties}
\label{disp}

In this section, we discuss the key properties of
deterministic equivariance,
distributional equivariance, and
distributional invariance, 
which form the building blocks of our theory.
We shed light on  the connections between these properties, and on their connections with various classical topics in statistics and the mathematical sciences.
These discussions are not strictly needed for understanding the description of our methods.

\subsubsection{Deterministic Equivariance}
\label{dispd}

Recall that the 
deterministic $\mG$ equivariance
condition \eqref{e} requires that 
for all  $z\in\mZ$ and all $g\in \mG$,
$V(\rg z)=  \trg V(z)$.
Equivalently, one can require that for all $\mZ$-valued random variables, $V(\rg Z)=  \trg V(Z)$.
Given $z \in \mZ$, and $\tilde z \in \tmZ$,
let $\mH_z,\tilde \mH_{\tilde z}$ be their stabilizers with respect to $\rho,\tir$, respectively.
We have the following result:

\begin{proposition}[Characterizing equivariant maps]\label{chardeteq}
The functions $V$
satisfying
the deterministic equivariance condition  \eqref{e} can be described as follows: partition $\mZ$ into orbits, 
collecting representatives in a set $R$, so that
we 
have the disjoint union
$\mZ = \cup_{z\in R} O_z$.
For each orbit representative $z\in R$, 
choose $V(z)$ such that  $\mH_z$ is a subgroup of $ \tilde\mH_{V(z)}$. 
For all 
$z' \in O_z$, choose any $g\in \mG$ such that
$z' = \rg z$, and define $V(z') = \trg V(z)$.     
\end{proposition}
\begin{proof}
First, we show that the procedure from the statement leads to equivariant functions.  
Indeed, if 
 for some $g'\in \mG$,
$\rg z = \rho(g')z$, 
then the definition of $V$ leads to $V(\rg z) = V(\rho(g')z)$, so $\trg V(z) = \tilde\rho(g') V(z)$. This shows that we must have $g(g')^{-1} \in  \tilde\mH_{V(z)}$. 
Now, $\rg z = \rho(g')z$ shows that we have 
$g(g')^{-1} \in \mH_{z}$; and since by assumption 
$\mH_z$ is a subgroup of $ \tilde\mH_{V(z)}$,
the required condition $g(g')^{-1} \in  \tilde\mH_{V(z)}$ follows, showing that the above procedure always leads to equivariant functions.  

Second, we show that all equivariant functions satisfy these conditions.
For any $h\in \mH_z$, we have $V(z) = V(\rh z) = \trh V(z)$. 
Thus, $\mH_z\subset \tilde\mH_{V(z)}$, and since $\mH_z,\tilde\mH_{V(z)}$ are both subgroups of $\mG$, 
it follows that  $\mH_z$ is a subgroup of $ \tilde\mH_{V(z)}$.
Moreover, if 
$z' = \rg z$, by equivariance, we have $V(z') = \trg V(z)$; showing that equivariant functions satisfy these conditions.     
\end{proof}

\subsubsection{Distributional Equivariance}
\label{deqs}

We next characterize distributional equivariance. 
Consider a fixed $z\in \mZ$.
 Recall that
 for $v \in \mZ$,
 $O_v = \{\rg \cdot v: g\in \mG\} $ is the orbit of $v$ under $\mG$.
Denote the distribution of $\rG v$ 
over the orbit $O_v$ with the sigma-algebra generated by the intersection of the sigma-algebra over $\mZ$ with $O_v$
when $G\sim U$
by $\mu_v$.
Similarly, for $v\in \tmZ$, 
define $\tilde  O_v = \{\trg  \cdot v: g\in \mG\} $.
We have the following result:

\begin{proposition}[Characterizing distributionally equivariant maps]\label{diem}
    The functions $V$
satisfying
the distributional equivariance condition  \eqref{e} can be described as follows: 
for $P$-almost every $z$, we have that 
for any measurable set $\tilde S \subset \tO_{V(z)}$
and
any $g\in \mG$,
    \beq\label{dec}
    \mu_{z}(V^{-1}(\tilde S))= \mu_{z}(V^{-1}(\trg  \tilde S)).
    \eeq
\end{proposition}

    Informally, this means that the preimages of the elements of the orbit of $V(z)$ under $\mG$ have the same density in the original orbit of $z$.
    Or, ``orbits map to orbits, in the natural measure-preserving way".
    In comparison, as shown in \Cref{dispd}, deterministic $\mG$-equivariance requires a more restrictive specific one-to-one correspondence within each orbit. 
    
\begin{proof}
If 
\eqref{d} holds, we can condition on $Z=z$  to deduce that $V(\rG z)=_d \trG V(z)$ for $P$-almost every $z$.
Now, $V(\rG z)$ is distributed as the pushforward 
$V \# \mu_z$ of $\mu_z$ under $V$, 
while $\trG V(z) \sim \tilde \mu_{V(z)}$, 
so \eqref{d} holds iff
    $$V \# \mu_z = \tilde \mu_{V(z)}.$$
    
    By the definition of a pushforward measure, for any measurable set $\tilde S \subset \tO_{\tilde z}$, 
    $V \# \mu_z(\tilde S) = \mu_{z}(V^{-1}(\tilde S)).$
    Hence, the above means 
    $\tilde \mu_{V(z)}(\tilde S)= \mu_{z}(V^{-1}(\tilde S)).$
    Since for $G\sim U$, and any $g\in \mG$,
    $\trg\trG V(z) = _d \trG V(z)$,
    it follows that
    $\tilde \mu_{V(z)}(\trg \tilde S)= \tilde \mu_{V(z)}(\tilde S)$, and thus the above implies
    that for any $g\in \mG$, \eqref{dec} holds.      
\end{proof}

For a finite group $\mG$, 
we obtain the following simpler result.
Given $z \in \mZ$, and $\tilde z \in \tmZ$,
let $\mH_z,\tilde \mH_{\tilde z}$ be their stabilizers with respect to $\rho,\tir$, respectively.

\begin{corollary}[Equivariance in finite groups]\label{ef}
If $\mG$ is finite, then
given $z \in \mZ$ and $\tilde z \in \tmZ$,
and the actions $\rho,\tir$ of $\mG$ on $\mZ,\tmZ$,
the existence of an 
equivariant map $V:O_z\to \tO_{\tz}$ such that $V(z) = \tz$ is characterized as follows:
    \benum
    \item Deterministic equivariance is equivalent to $\mH_z$ being a subgroup of $ \tilde\mH_{\tz}$, 
     $\mH_z \leq \tilde \mH_{\tz}$.
     \item Distributional equivariance is equivalent to 
    the size of the group $\mH_z$ dividing the size of the group $\tilde \mH_{\tz}$, i.e., $|\mH_z| \mid |\tilde \mH_{\tz}| $.
     
    \eenum
\end{corollary}

Clearly, the condition for equivariant maps is stricter. 

\begin{proof}
The first part follows from Proposition \ref{chardeteq}.
For the second part,
since $\{\tz\} \in \tO_{\tz}$
is a measurable subset of $\tO_{\tz}$,
and since
for $S\subset O_v$ and
$G$ distributed uniformly on the finite group $\mG$,
we have
$\mu_v(S) = P(\rho(G)v\in S) = |S|/|O_v|$,
the condition \eqref{dec} 
characterizing distributional equivariance becomes that  $r(\tz):=|V^{-1}(\trg  \tilde z)| \in \mathbb{N}_{>0}$ does not depend on $g $. 
Now, from the orbit-stabilizer theorem \citep[][Proposition 6.8.4]{artin2018algebra}, 
we have that $\mG/\mH_{z} \cong O_z$, $\mG/\tilde \mH_{z} \cong \tO_{\tz}$.
Thus, $|\mG|/|\mH_{z}| = r(\tz) |\mG|/|\tilde \mH_{\tz}|$, and hence $|\tilde \mH_{\tz}| = r(\tz) |\mH_{z}|$, and so the size of the group $\mH_{z}$ divides the size of the group $\tilde \mH_{\tz}$. 
\end{proof}

As an example
for a finite group $\mG$,
consider any 
sets of representatives $\mathcal{R}, \tilde{\mathcal{R}}$ of the orbits of the action of $\rho$ on $\mZ$, and $\tir$ on $\tmZ$, 
an arbitary injective map $r:\mathcal{R}\to \tilde{\mathcal{R}}$,
and any 
collection of one-to-one maps $V_z:O_z\to \tilde O_{r(z)}$.
Then $V$ defined as $V(z') = V_z(z')$ when $z'\in O_z$ is distributionally equivariant.
For concreteness,
let
$\mZ =\mZ' = \{0,1,\ldots, j\}$
and 
$\mG = \mathbb{Z}_{j+1} =  (\{0,1,\ldots, j\},+)$
acting via addition 
$g\cdot z = g+z $ modulo $j+1$.
Then 
deterministic equivariance requires
$V(g+z) = g+V(z)$  modulo $j+1$, for all $g,z$.
This means
$V(z+g) = z+V(g)$, so $g+V(z) = z+V(g)$, so
with $a = V(0)-0$, we have
$V(z) = z  + a$ for all $z$.
In contrast, distributional equivariance requires that
$V(G+z) =_d G+V(z)$  modulo $j+1$, when $G\sim U$, and for all $z$.
It is clear that any function $V:\mZ\to \mZ'$ satisfies this.

For the special case of the symmetric group where $\mG = \S_n$, and for the permutation actions $\rho, \tir$ acting on $\mZ = \mZ_0^n$, \cite{dean1990linear} have provided a sufficient condition for a transform $V$ to preserve equivariance. 
Their condition states that for any 
$z\in \mZ_0^n$,
and any 
$g'\in \S_n$,
there is $g\in \S_n$ such that 
$g'V(z) = V(gz)$.
Our result recovers theirs in this special case.
Indeed, our condition Corollary \ref{diem}
in this case
states that the preimage of the orbit 
$\{gV(z):g\in S_n\}$ 
under $V$
equals the orbit
$\{gz:g\in S_n\}$, which matches their condition.

\subsubsection{Distributional Invariance}
\label{dip}

The distributional invariance $Z =_d \rG Z$, when $G\sim U$, is equivalent to the distribution of $Z|\{Z\in O_z\}$ being $\mu_z$, for $P$-almost every $z\in \mZ$.
    Thus, in this case, the orbits (or measurably chosen representatives), 
    are a sufficient statistic for the  distribution of $Z$. 
    Thus distributional invariance is an example of conditional ancillarity.

    The key advantage of distributional invariance is that it is preserved under compositions.
    In contrast, this is not always convenient for conditionally ancillarity. Suppose $Z\sim P$, $p\in \mP$, is conditionally ancillary given $A$, where $A:\mZ\to \mA$ is a map, 
    i.e., 
    for almost every $a$ we have that 
    $Z|A=a$ has the same distributions for all $P\in \mP$.
    Now, for a map $V:\mZ\to\tmZ$, we ask when $V(Z)$ is conditionally ancillary given $A$. In general, this is not ensured, because $V$ may map different preimages $A^{-1}(a)$, $a\in \mA$ to the same value, and hence the resulting conditional distribution may mix the distributions of  $Z|A=a$ with the---possibly $P$-dependent---distributions of $Z$ and $A$.
    
    In contrast, a key advantage of group invariance is that it does not have this restriction. 
    Different orbits may map into the same orbit, and the resulting distribution is still uniform.
    In the end, this enables the development of broader classes of architectures that preserve invariance and finally power our methods.
 
\subsection{Distribution Shift with Non-symmetric Algorithm}\label{non-symmetric_algorithm}

This section is devoted to studying predictive inference in cases where a non-symmetric algorithm is employed, 
even if we have a distribution shift and so $Z \neq_d \rho(G) Z$. We present a novel algorithm along with theoretical coverage guarantees, wherein possibly distinct weights are assigned to various members within a representative set $S$, and the function $V$ may not necessarily be distributionally equivariant.

For a given group $\mathcal{G}$ and a function $\psi$, we consider the induced functions $\mathcal{F}=\{\ell_g \mid g\in \mathcal{G}\}$, where each $\ell_g$ 
is defined as $\ell_g(\tilde{z})= \psi(\tilde{\rho}(g)\tilde{z})$ for all $\tilde{z}$. {Recall that we defined in the main text that $\mH=\{g\in S_{n+1}: \psi(\tilde{\rho}(g)\tilde{z})=\psi (\tilde{z}) \textnormal{ for all } \tilde{z}\in \tmZ\}$ and the cosets are $\mathcal{G}/\mathcal{H}$.}
Here we focus on a simplified scenario where the set $\mathcal{F}$ is finite and can be represented as $\mathcal{F}=\{\psi_1,\ldots,\psi_{|\mathcal{F}|}\}$. 
In this case, 
in our non-symmetric algorithm,
we sample the cosets (or, equivalently, a set of representatives denoted as $S=\{g_1,\ldots,g_{|\mathcal{F}|}\}$) 
from a distribution $\Gamma_{S}$ 
with probabilities given by $(w_1,\ldots,w_{|\mathcal{F}|})$. 
The procedure can also be extended to continuous groups using an approach similar to that discussed earlier. 

The weights in $\Gamma_{S}$ can be arbitrary but 
are 
typically 
chosen, aiming to minimize the coverage gap due to the distribution shift. 
For example, when considering the group $G=S_{n+1}$ and the function $\psi(z)=z_{n+1}$ for all $z$, 
{we will see that the representatives are in a one-to-one correspondence with the coordinates $z_1, \ldots, z_{n+1}$, and the reduction in coverage can be related to the total variation distance between $Z$ and $Z$ with coordinates $j$ and $n+1$ swapped, for $j=1,\ldots,n+1$.}
In this case, the selection of weights is guided by the 
characteristics of the data, such as time series analysis and change point detection \citep{barber2023conformal}. 
A common strategy is to allocate greater weights to elements that are closer in time to the unobserved {$(n+1)$-st data point}. 

Next, we formally define our prediction set when we do not assign equal weights to the representatives of cosets or we do not have a distributionally equivariant map $V$.
We first sample a representative $g$ according to the probability measure $\Gamma_{S}$ and we construct our prediction set as follows:
\begin{align}\label{non-asym-alg}
T^{\mathrm{non-sym}}(\zo)=\Bigg\{z:\psi(\tilde{\rho}(g)V(\rho^{-1}(g)z))\le Q_{1-\alpha}\bigg(\sum_{j=1}^{|\mathcal{F}|}w_j\delta_{\psi_j(V(\rho^{-1}(g)z))}\bigg),\obs(z)=\zo\Bigg\}.
\end{align}

We next provide a coverage property for this prediction set. We let $$ \Delta^w:=\sum_{i=1}^{|\mathcal{F}|}w_i 
 \mathrm{TV}\big(\nu_i(V(\rho^{-1}(g_i)Z)),\nu_i(V(Z))\big),$$
 where $\nu_i(x)=\psi_i(x)-Q_{1-\alpha}(\sum_{j=1}^{|\mathcal{F}|}w_j\delta_{\psi_j(x)})$ for all $x$.
In addition, we let $F^w_{\tz}$ be the c.d.f.~of the random variable $\psi(\rho(g)  \tz),\, g\sim \Gamma_S$.
Furthermore, we let $F^{w'}_{\tz}$ be the probability it places on individual points, i.e., for $x\in \R$, 
$F^{w'}_{\tz}(x) = F^w_{\tz}(x)-F^{w-}_{\tz}(x)$, 
where
$F^{w-}_{\tz}(x) = \lim_{y\to x, y<x} F^w_{\tz}(y) \ge 0$. We then define $t^w_{\tilde{z}}:=Q_{1-\alpha}(\sum_{j=1}^{|\mathcal{F}|}w_j\delta_{\psi_j(\tilde{z})}).$ for all $\tz$.
See \Cref{pfthm_asy} for the proof of the result below.
\begin{theorem}\label{thm_asy}
Regardless of the weights 
$(w_1,\ldots,w_{|\mathcal{F}|})$, 
and regardless of whether $Z$ is distributionally invariant and $V$ is distributionally equivariant, 
the prediction set from \eqref{non-asym-alg} satisfies the coverage bound
 \begin{align*}
-\Delta^w\le P\Big(Z\in T^{\mathrm{non-sym}}(\Zo )\Big)-(1-\alpha) \le \E[F^{w'}_{V(Z)}(t^w_{V(Z)})]+\Delta^w.
\end{align*}
\end{theorem}

When $Z$ is distributionally invariant over $\mathcal{G}$, the coverage lower bound reduces to $1-\alpha$.   If this distributional invariance property does not hold, we see that the prediction set now relies on the set of representatives $S$ that we choose. Therefore, in order to minimize the coverage gap, we suggest choosing each $g_i$ from each coset that minimizes the difference between $Z$ and $\rho^{-1}(g_i)Z$. For the case in \cite{barber2023conformal}, the suggested $g_i\in \S_{n+1}$ are permuting the $i$th and the $(n+1)$st entry of $Z$.
While it would be desirable to compare the coverage of the symmetric and non-symmetric algorithms, in general, this does not seem straightforward.
Therefore, we leave this to future work.

Our coverage conclusion in \Cref{thm_asy} reduces to the conclusion of non-exchangeable conformal prediction in Theorem
2 and 3 presented in \cite{barber2023conformal} 
when we let $\mathcal{G}=\S_{n+1}$ and $\psi(\tilde{z})=e_{n+1}^\top \tilde{z}$ and train a non-symmetric non-conformity score in terms of the input order of the data. We present the relevant details as an example, illustrated below:
\begin{example}[Non-exchangeable conformal prediction]
To recover the results of 
\cite{barber2023conformal}, we consider $Z=(Z_1,\ldots,Z_{n+1})^\top$, where we have $Z_i=(X_i,Y_i)$ with $X_i$ being the covariate, $Y_i$ being the response.

Let $0\le w_1\le w_2\le \ldots \le w_n\le w_{n+1}  = 1$, and $\tw_i=w_i/(\sum_{j=1}^{n+1} w_j)$, $i\in [n+1]$. 
Denote by $g_j$ the transposition exchanging $j$ and $n+1$, keeping other indices fixed, 
so that $g_jz:=z^{(j)}=(z_1,\ldots$,  
$z_{j-1},z_{n+1},z_{j+1},\ldots,z_j)^\top$, where we exchange $z_{n+1} \leftrightarrow z_{j}$.
{Letting $\mH=\{g\in S_{n+1}: e_{n+1}^{\top}(g\cdot\tilde{z})=e_{n+1}^\top \tilde{z} \textnormal{ for all } \tilde{z}\in \tmZ\}$ be the set of permutations fixing the last coordinate,}
note that $g_j$, $j\in [n+1]$ is a set of representatives of $\mG/\mH$. 

We construct the prediction set in the same way as \eqref{non-asym-alg} by letting $\psi(x)=e_{n+1}^\top x$, $\mathcal{G}=S_{n+1}$, and taking $Z_i  = (X_i,Y_i)$ for all $i\in[n+1]$, and $V(z) = (R(z_1),\ldots, R(z_{n+1}))^\top$, where $R(z_i) = |y_i- \hmu^z(x_i)|$,  for all $i\in[n+1]$, 
with $\hmu^z:\mZ\to \R$ being non-symmetric in the sense that $\hmu^{gz}\neq  \hmu^z$  for some $g\in \S_{n+1}$.
Then, 
since $\textnormal{TV}\Big(\nu_j(V(Z)),\nu_j(V(g_j^{-1}Z))\Big) \le \textnormal{TV}\Big(V(Z),V(g_j^{-1}Z)\Big)$ for all $i\in [n+1]$,
Theorem \ref{thm_asy}
implies that the coverage probability is lower bounded as 
     \begin{align*}
    P_Z\Big(Z_{n+1} \in T^{\tw,n+1}(Z_{1:n})\Big)\ge 1-\alpha-\sum_{j=1}^{n}\tw_j\textnormal{TV}\Big(V(g_j^{-1}Z),V(Z)\Big).
\end{align*}
When $V$ is injective and
$(Z_1,\ldots, Z_{n+1})$ has a continuous distribution, Theorem \ref{thm_shift} implies that 
 \begin{align*}
    P_Z\Big(Z_{n+1} \in T^{\tw,n+1}(Z_{1:n})\Big)\le 1-\alpha+\tw_{n+1}+\sum_{j=1}^{n}\tw_j\textnormal{TV}\Big(V(g_j^{-1}Z),V(Z)\Big).
\end{align*}

Then, these results match with the conclusions in Theorems 2b \& 3b of \cite{barber2023conformal}.
\end{example}

Next, we discuss the example of a two-layer tree.

\begin{example}[Non-exchangeable conformal prediction for two-layer trees]
Continuing the example from \Cref{g}, we consider predicting $Z_{M}^{(K)}$ in the final branch of the two-layer tree depicted in Figure \ref{fig:tree2}, left panel. 
In the following, we will adopt the notations used there. 
However, we allow that $g\cdot \Gamma \neq _d \Gamma$: for instance, branch growth can be time-ordered and
nodes within each branch $\Gamma_k,k\in [K]$ can have stronger dependence 
compared to nodes across different branches. In such a scenario, it is appropriate to assign greater weights, denoted as $w_i^{(k)},i\in [M],k\in [K]$, to the leaves of the last branch relative to those in earlier branches.

For example, we can pick $0\le w_{1}^{(1)}= \cdots= w_{1}^{(K)}\le\cdots \le w_M^{(1)}= \cdots =w_{M}^{(K)}$ and let $\tilde{w}_i^{(k)}=w_i^{(k)}/(\sum_{j=1,k=1}^{M,K}w_j^{(k)}).$ We let $g_s^{(q)},s\in [K],q\in [M]$ be the set of representatives of $\mathcal{G}/\mathcal{H},$ and $g_s^{(q)}z:=z^{(s,q)}=(R,\Gamma_1,\cdots,\Gamma_{q-1},\Gamma_K^{s},\Gamma_{q+1},\cdots,\Gamma_{K-1},\Gamma_q),$ where $$\Gamma_K^{s}=(C_K,Z_{1}^{(K)},\cdots,Z_{s-1}^{(K)},Z_{M}^{(K)},Z_{s+1}^{(K)},\cdots,Z_{s}^{(K)}).$$ 
Here we first exchange the $q$-th branch with the $K$-th branch and then permute the last entry of the current $q$-th branch (the original $K$-th branch) with its $s$-th entry.  
We keep $V(z)=|z|$ for all $z\in \mZ$, 
as we did in \Cref{g}. Therefore,
Theorem \ref{thm_asy}
implies that the coverage probability is lower bounded as 
     \begin{align*}
    P_Z\Big(Z_{M}^{(K)} \in T^{\tw}(Z_{i}^{(k)},i\in [M],k\in [K] \setminus(M,K))\Big)\ge 1-\alpha-\sum_{i=1,k=1}^{M,K\setminus(M,K)}\tw_i^{(k)}\textnormal{TV}\Big(|(g_i^{(k)})^{-1}Z|,|Z|\Big).
\end{align*}
When 
$Z=(R,\Gamma_1,\cdots,\Gamma_K)$ has a continuous distribution, Theorem \ref{thm_asy} implies that 
 \begin{align*}
    &P_Z\Big(Z_{M}^{(K)} \in T^{\tw}(Z_{i}^{(k)},i\in [M],k\in [K] \setminus(M,K))\Big)\\
    &\le 1-\alpha+\tw_{M}^{(K)}+\sum_{i=1,k=1}^{M,K\setminus(M,K)}\tw_i^{(k)}\textnormal{TV}\Big(|\big(g_i^{(k)}\big)^{-1}Z|,|Z|\Big).
\end{align*}

\end{example}

 \subsection{Coarsening Approach for Predictive Inference on Graphs}
 \label{coarse}
From a practical perspective, if the graph is large, 
we can coarsen it. 
For instance, we can use a hierarchical graph clustering method to cluster nodes. 
We keep edges between the clustered nodes with associated multiplicities; thus we obtain a graph with multiple edges allowed between nodes.
The automorphism group of the new graph can then be viewed as a subgroup of the original graph that fixes the vertices within each clustered node.

We can apply our method to the clustered graph 
by considering the values of clusters for which node observations are missing as also missing.
This construction leads to prediction sets at the cluster-level, which we can view as prediction sets for any symmetric function of the node values, such as for their sum. 
If there is only one missing observation per cluster, then we can back out prediction sets for the individual nodes; 
e.g., if using the sum to aggregate, by subtracting the sum of the labeled nodes in the cluster. 
If there are several nodes with missing values in a cluster, then we obtain a prediction set for their sum. 

For the tree-structured graphical model from
\Cref{g}
we can predict at a cluster after coarsening the graph as follows.
We can 
define the clusters as the branches
 $C_k, Z_{1}^{(k)},\ldots,Z_{M}^{(k)}$, for $k\in [K]$.
We can re-write the data as $\Gamma=(R,\Gamma_1^\top,\ldots,\Gamma_K^\top)^\top$, where we treat every branch $\Gamma_k:=(C_K,Z_1^{(K)},\ldots,Z_{M}^{(K)})^\top, k\in [K]$ as one component. 
Our goal is to predict the sum of values of the last branch.
We let $\psi(z)=|(0, 0_{(K-1)\cdot (M+1)}^\top,1_{M+1}^\top)^\top z|$, so the quantile is $t^{(2)}_{z}:=Q_{1-\alpha}(|1^\top \gamma_1|,\ldots,|1^\top \gamma_K|)$,
where $\gamma_k$ are the realized values of $\Gamma_k$, $k\in[K]$,
and the prediction set is given by
\beqs
T^{\mathrm\OurMethod}(\zo) = \{z: |1^\top \gamma_K|\le t^{(2)}_{\tilde{z}},\, \o(z)= \zo \} .
\eeqs

\subsection{Graph Neural Network Construction for Unsupervised Learning}
Here we give two 
graph neural network constructions 
for unsupervised learning: 
a simpler one that is easier to understand,
and a more sophisticated one that has better performance by interpolating with standard conformal prediction.

\subsubsection{GNN for Unsupervised Learning}\label{sec:GNN_unsup}
We design a fixed GNN architecture
where the neighborhood structure of the graph is as in \Cref{fig:tree2}:
\benum
\item
 We let the sample means over the second-layer nodes be proxy statistics, 
setting 
$$\tilde P^{(k,0)} = {M}^{-1}\sum_{i=1}^{M}Z_i^{(k)},\, k=1,\ldots, K.$$
More generally, any permutation invariant functions over $Z_i^{(k)},i\in [M]$ are a valid choice.
We also set the proxy $\tilde{P}^0$ for the root note to be zero, and let $Z_{i}^{(k,0)} =Z_{i}^{(k)}$.
\item We apply a modified two-layer MPGNN. In the first layer, we have the following.
\benum
\item
We update the leaves as $Z_{i}^{(k,1)}=|Z_i^{(k,0)}-\tilde{P}^{(k,0)}|$.

\item We update the proxy variables as
$\tilde P^{(k,1)}=
\sum_{i\in \N(k)}\bigl(\tilde P^{(k,0)}- Z_i^{(k,0)}\bigr)^2/(M-1),$
which corresponds to \eqref{u} with $V_0(x)=y$, $V_{11}(x,y)=(x-y)^2/(M-1)$ and $V_{12}(\tilde{P}^{(k,0)},\tilde{P}^0)=0$. 
\eenum
In the second layer of the MPGNN, we update the leaves as follows:
$$Z_{i}^{(k,2)}=
\frac{Z_{i}^{(k,1)}}
{\sqrt{\tilde{P}^{(k,1)}}}
=
\frac{\left|Z_i^{(k)}-{\sum_{i=1}^M Z_{i}^{(k)}}/{M}\right|}
{\sqrt{\frac{1}{M-1}\sum_{j\in \N(k)}\biggl(Z_j^{(k)}-\frac{\sum_{j=1}^M Z_{j}^{(k)}}{M}\biggr)^2}}.$$
This corresponds to \eqref{u} with $V_1(x,y)=x/\sqrt{y},V_0(x,y)=y$,

\item We construct prediction sets 
for $Z_{M}^{K}$
by following the method in \Cref{sec:predict_region}.

\eenum

\subsubsection{Interpolation with Conformal Inference}

Here, for every variable we have two channels. 
\benum
\item Step 0: (Initialization) 
\benum
\item 
We let 
$\bar{Z}=\frac{1}{KM}\sum_{k=1}^K\sum_{i=1}^M Z_i^{(k,0)}$
and $\hat\sigma$ the sample standard deviation of the full sample.
We set 
the zeroth-layer node  $\tilde{P}^{0}$ 
to be\footnote{Alternatively, we can achieve the same effect by initializing zeroth- and first-layer nodes as two-dimensional zero vectors first, and updating them through message passing from the leaves step by step (passing the the sample mean and standard errors).}
$(\bar{Z},\hat\sigma)^\top$. 
\item  For first-layer nodes $\tilde{P}^{(k,0)}$, we initialize them as $(\frac{1}{M}\sum_{i=1}^M Z_i^{(k,0)},\hat\sigma_k)^\top$, where $\hat\sigma_k$ is the standard error of the $\{Z_i^{(k,0)}\}_{i=1}^M.$
\item 
For leaves, we augment all features $Z_i^{(k,0)}$ with a constant $1$; as the second channel is not used and only kept so that all nodes have two features. We denote $X_i^{(k,0)}=(Z_i^{(k,0)},1)^\top.$

\eenum

\item  Step 1: (Update) 

We update $\tilde{P}^{(k,0)}$ to $\tilde{P}^{(k,1)}$ for all $k\in [K]$ via $\tilde{P}^{(k,0)}$ and $\tilde{P}^{0}$ as \begin{align*}
\tilde{P}^{(k,1)}=V^{(1)}(\tilde{P}^{(k,0)},\tilde{P}^{0})&=\bigg(\bigg|\frac{\tilde{P}^{(k,0)}(1)-\tilde{P}^{0}(1)}{\tilde{P}^{(k,0)}(2)/\sqrt{M}}\bigg|,\tilde{P}^{(k,0)}(2)\bigg)^\top
=\bigg(\bigg|\frac{\bar{Z}_k-\bar{Z}}{\hat\sigma_k/\sqrt{M}}\bigg|,\hat\sigma_k\bigg)^\top.
\end{align*}
We keep other variables---the zeroth layer node  $\tilde{P}^{(1)}$
and leaves $X_i^{(k,1)},\forall k\in[K],i\in [M]$---unchanged:
$\tilde{P}^{(1)}=\tilde{P}^{(0)}$ and $X_i^{(k,1)}=X_i^{(k,0)},\forall k\in[K],i\in [M]$.
\item  Step 2: (Update) 

Again, we only update $\tilde{P}^{(k,1)}$, via
\begin{align*} \tilde{P}^{(k,2)}=V^{(2)}_0\bigg(\tilde{P}^{(k,1)},\sum_{i=1}^{M}V^{(2)}_{11}(\tilde{P}^{(k,1)},X_i^{(k,1)})+V^{(2)}_{12}(\tilde{P}^{(k,1)},\tilde{P}^{(1)})\bigg).
\end{align*}
Here we take $V^{(2)}_{11}(\tilde{P}^{(k,1)},X_i^{(k,1)})=Z_{i}^{(k,1)}\mathbb{K}(\tilde{P}^{(k,1)}(1))/M$, 
for some kernel $\mathbb{K}$. 
Therefore, $\sum_{i=1}^{M}V^{(2)}_{11}(\tilde{P}^{(k,1)},X_i^{(k,1)})$ $=\bar{Z}_k\mathbb{K}(\tilde{P}^{(k,1)}(1)).$
We set 
$\mathbb{K}(x)=I(|x|\le c)$ with some constant $c$, so that 
\begin{align*}
    \sum_{i=1}^{M}V_{11}^{(2)}(\tilde{P}^{(k,1)},X_i^{(k,1)})=\bar{Z}_kI\bigg(\bigg|\frac{\bar{Z}_k-\bar{Z}}{\hat\sigma_k/\sqrt{M}}\bigg|\le c\bigg).
\end{align*}
For $V_{12}^{(2)}$ we let $V_{12}^{(2)}(\tilde{P}^{(k,1)},\tilde{P}^{(1)})=\bar{Z}(1-\mathbb{K}(\tilde{P}^{(k,1)}(1)))=\bar{Z}I\bigg(\bigg|\frac{\bar{Z}_k-\bar{Z}}{\hat\sigma_k/\sqrt{M}}\bigg|>c\bigg)$
Therefore, 
\begin{align}\label{het}
\tilde{P}^{(k,2)}=\Bigg(\bar{Z}_kI\bigg(\bigg|\frac{\bar{Z}_k-\bar{Z}}{\hat\sigma_k/\sqrt{M}}\bigg|\le c\bigg)+\bar{Z}I\bigg(\bigg|\frac{\bar{Z}_k-\bar{Z}}{\hat\sigma_k/\sqrt{M}}\bigg|>c\bigg),\hat\sigma_k\Bigg).
\end{align}
We keep $X_i^{(k,2)}=X_i^{(k,1)}$ and $\tilde{P}^{(2)}=\tilde{P}^{(1)}.$

\item Step 3: (Update) 

We finally update all $X_i^{(k,3)}$, for all $i\in[M],k\in [K]$. We let 
\begin{align*}
X_i^{(k,3)}&=V_0^{(3)}(X_i^{(k,2)},V_1^{(3)}(X_i^{(k,2)},\tilde{P}^{(k,2)}))\\
&=\Bigg(\Bigg|Z_{i}^{(k)}-\left[\bar{Z}I\left(\frac{|\bar{Z}_k-\bar{Z}|}{\hat\sigma_k/\sqrt{M}}\le c\right)
+\bar{Z}_kI\left(\frac{|\bar{Z}_k-\bar{Z}|}{\hat\sigma_k/\sqrt{M}}> c\right)\right]
\bigg|/\hat\sigma_k,1\Bigg)^\top.
\end{align*}
We take $v=(v_1;\ldots;v_{n+1})^\top\in \mathbb{R}^{2\cdot(n+1)}$ and let $v_{n+1}=(1,0)^\top.$
\eenum

\subsection{Simulation with Random Sample Sizes}
\label{rss}
In this simulation, we let the number of branches be $K=20$. For every branch, we let $N_k\sim \mathrm{Unif}(\{10,20\}), k\in [K]$ independently in the setting of unsupervised learning. 
The other settings are the same as for the fixed sample size regime presented in the main text. 
The final statistics are:
$$\Bigg|Z_{i}^{(k)}-\left[\bar{Z}I\left(\frac{|\bar{Z}_k-\bar{Z}|}{\hat\sigma_k/\sqrt{M}}\le c\right)
+\bar{Z}_kI\left(\frac{|\bar{Z}_k-\bar{Z}|}{\hat\sigma_k/\sqrt{M}}> c\right)\right]
\bigg|/\hat\sigma_k,$$ 
for $i\in [M],k\in [K]$, where $\bar{Z}=\frac{1}{K}\sum_{k=1}^{K}\frac{1}{N_k}\sum_{i=1}^{N_k}Z_{i}^{(k)}$, for unsupervised learning.

For supervised learning, we also consider 
a setting very similar to the fixed sample size regime presented in the main text. 
The only difference is that we sample $N_k\sim \mathrm{Unif}(\{20,40\}),k\in [K]$ independently for each branch and use the first half of all branches as the training data. 
The output statistics are the same as \eqref{stat_1} for supervised learning except that the datapoints used for training $\hat\mu_k,k\in [K]$ and $\hat\mu$ are of random sizes.

In addition to the methods used in the main text, we also compare with hierarchical conformal prediction (HCP) \citep{lee2023distribution}.
The results are presented in Tables \ref{tab:random_N_unsup} and \ref{tab: random_N_sup}, respectively.
The conclusions are identical to those from the experiments from the main text.
Here HCP performs well, similarly to our method.
However, of course our method is more general and applicable to any symmetry group, not just to the hierarchical setting, 
as described in the main text.

    \begin{table}[t]
	\begin{center}
		\def\arraystretch{1.2}
		\setlength\tabcolsep{4pt}
		\begin{tabular}{c||ccccc}
			\toprule
			&Estimator  & $\sigma^2=10$ & $\sigma^2=2$ & $\sigma^2=0.5$  & $\sigma^2=0$   \\
			\hline
   \multirow{10}{*}{$\alpha=0.05$} & Length: \OurMethod~& 2.076 (0.012) & 2.097 (0.016)  & 2.138 (0.224) & 2.019 (0.012) \\ 
   & Coverage: \OurMethod~& 0.957 (0.022)  & 0.950 (0.023) &0.949 (0.023) & 0.952 (0.022) \\
          & Length: HCP & 41.076 (0.865) & 8.064 (0.142) & 2.774 (0.030) & 1.996 (0.010) \\
           & Coverage: HCP & 0.948 (0.027)  & 0.955 (0.017)  &0.945 (0.022) & 0.952 (0.024) \\
       & Length: Conformal & 38.222 (0.716) & 7.859 (0.138) & 2.740 (0.031)& 1.977 (0.009) \\
           & Coverage: Conformal & 0.945 (0.025) & 0.950 (0.019)  & 0.943 (0.023) & 0.947 (0.025)  \\
                                        & Length: Subsampling & 44.377 (0.958) & 9.149 (0.168) & 3.129 (0.067) & 2.220 (0.048) \\
           & Coverage: Subsampling& 0.949 (0.026) &0.954 (0.017)  & 0.949 (0.021)  & 0.951 (0.025) \\
           & Length: Single-Tree & Inf & Inf & Inf & Inf  \\
           & Coverage: Single-Tree & 1.0 (0.0)  & 1.0 (0.0)  & 1.0 (0.0) & 1.0 (0.0) \\
    \hline
       \multirow{10}{*}{$\alpha=0.15$} & Length: \OurMethod~& 1.516 (0.010) & 1.524 (0.014)   & 1.547 (0.014) & 1.470 (0.012)\\ 
   & Coverage: \OurMethod~& 0.844 (0.033)  & 0.842 (0.038) &0.850 (0.033) & 0.857 (0.039) \\
          & Length: HCP &28.771 (0.612) &5.835 (0.127) &2.030 (0.017)  & 1.451 (0.008) \\
           & Coverage: HCP &0.849 (0.039)  & 0.837 (0.030) &0.848 (0.032) & 0.853 (0.040)\\
       & Length: Conformal &28.006 (0.668)   & 5.798 (0.125) & 2.016 (0.016) & 1.450 (0.007)\\
           & Coverage: Conformal &0.846 (0.041)  & 0.837 (0.033) & 0.843 (0.030) & 0.853 (0.038)\\
                             & Length: Subsampling & 30.821 (0.670) &6.347 (0.148) & 2.172 (0.042) & 1.549 (0.028)\\
           & Coverage: Subsampling&0.848 (0.039)  &0.838 (0.031)   & 0.848 (0.034) & 0.863 (0.035)\\
           & Length: Single-Tree &1.758 (0.052) &1.770 (0.053) & 1.749 (0.056)  & 1.738 (0.047)\\
           & Coverage: Single-Tree&0.872 (0.031)  & 0.877 (0.031)  & 0.874 (0.031) & 0.878 (0.031) \\
    \hline

			\bottomrule
		\end{tabular}
	\end{center}
	\caption{Prediction set length and coverage probability of our method and three benchmarks for unsupervised learning with a random sample size; same protocol as in Table \ref{tab: fix_N_unsup}.} 
	\label{tab:random_N_unsup}
\end{table}

    \begin{table}[t]
	\begin{center}
		\def\arraystretch{1.2}
		\setlength\tabcolsep{4pt}
		\begin{tabular}{c||ccccc}
			\toprule
			&Estimator  & $\sigma^2=10$ & $\sigma^2=2$ & $\sigma^2=0.5$  & $\sigma^2=0$   \\
			\hline
   \multirow{10}{*}{$\alpha=0.05$} & Length: \OurMethod~& 2.080 (0.021) & 2.105 (0.024)  & 2.067 (0.012) & 2.013 (0.016) \\ 
   & Coverage: \OurMethod~& 0.955 (0.021) &  0.942 (0.024) & 0.950 (0.020) & 0.951 (0.019) \\
          & Length: HCP & 12.618 (0.023)  & 3.098 (0.037) & 2.071 (0.013)  &  1.993 (0.015) \\
           & Coverage: HCP & 0.953 (0.023)  & 0.952 (0.024) & 0.948 (0.023) & 0.950 (0.019) \\
       & Length: Conformal & 12.353 (0.248)  & 3.051 (0.038) & 2.051 (0.012) & 1.976 (0.014) \\
           & Coverage: Conformal & 0.948 (0.024)   & 0.948 (0.025) & 0.945 (0.023) & 0.948 (0.020) \\
                                        & Length: Subsampling & 12.821 (0.618) & 3.152 (0.091) & 2.069 (0.058) & 1.990 (0.051) \\
           & Coverage: Subsampling& 0.923 (0.025)  & 0.919 (0.032) & 0.927 (0.028)  & 0.932 (0.022) \\
           & Length: Single-Tree & Inf & Inf & Inf & Inf  \\
           & Coverage: Single-Tree& 1.0 (0.0)  &  1.0 (0.0) & 1.0 (0.0) & 1.0 (0.0) \\
    \hline
       \multirow{10}{*}{$\alpha=0.15$} & Length: \OurMethod~& 1.512 (0.008) & 1.531 (0.009)  &1.507 (0.008) & 1.463 (0.008) \\ 
   & Coverage: \OurMethod~&0.854 (0.037)  & 0.847 (0.038)  & 0.857 (0.033) & 0.852 (0.034) \\
          & Length: HCP & 7.986 (0.153) &  2.143 (0.024) & 1.510 (0.008) & 1.451 (0.008)  \\
           & Coverage: HCP & 0.855 (0.032) & 0.857 (0.031)  & 0.858 (0.034) & 0.854 (0.036) \\
       & Length: Conformal &7.890 (0.163)   & 2.125 (0.025) & 1.500 (0.008) &  1.444 (0.008) \\
           & Coverage: Conformal & 0.848 (0.033) & 0.853 (0.033) & 0.855 (0.033) & 0.851 (0.037) \\
                             & Length: Subsampling & 8.414 (0.293) &2.229 (0.054) &  1.565 (0.035) & 1.479 (0.032) \\
           & Coverage: Subsampling& 0.841 (0.037)  & 0.843 (0.035)  & 0.846 (0.038) & 0.836 (0.035)\\
           & Length: Single-Tree & 1.732 (0.046) & 1.745 (0.051) & 1.727 (0.043) & 1.689 (0.050) \\
           & Coverage: Single-Tree& 0.876 (0.031)  &  0.878 (0.026) & 0.881 (0.028) & 0.867 (0.038)\\
    \hline

			\bottomrule
		\end{tabular}
	\end{center}
	\caption{Prediction set length and coverage probability of our method and three benchmarks for supervised learning with a random sample size; same protocol as in Table \ref{tab: fix_N_unsup}.} 
	\label{tab: random_N_sup}
\end{table}

\subsection{Additional Information about Empirical Data Example}

\subsubsection{Discusson of Modelling Assumptions}
\label{das}

Finally, we discuss modelling assumptions, and in particular the applicability 
of 
exchangeability assumptions to the empirical dataset. 
The covariate $X_1$ is the number of days lacking sleep. 
Within a branch, this covariate has a time trend, taking values $0,1,\ldots,9$.
Thus, its values for one datapoint are not exchangeable with the values for other data points within the same branch. 
However, even if the covariates are not exchangeable, the  values  $\epsilon_k=(\epsilon_{1}^{(k)},\ldots,\epsilon_{9}^{(k)})^\top$ 
of the random noise
can be assumed exchangeable. 
Therefore,
when we fit $\hat\mu_k,k\in [18]$ accurately, 
the residuals (or non-conformity scores) are ``nearly" exchangeable. 
This is expected to be a sufficient assumption for analyzing this dataset. 
In fact, when $\epsilon_i^{(k)}$ is independent of $X_i^{(k)}$, 
some distributions of $\epsilon_i^{(k)}$ (such as the Gaussian distribution) 
admit an upper bound $TV(Y_i^{(k)}-\hat\mu_k(X_i^{(k)}), Y_i^{(k)}-\mu_k^*(X_i^{(k)}))\le C \mathbb{E}_{X_i^{(k)}}[|\hat\mu_k(X_i^{(k)})-\mu_k^*(X_i^{(k)}))|]$. 
Hence, if we estimate $\mu_k^*$ accurately, then the empirical residuals $Y_i^{(k)}-\hat\mu_k(X_i^{(k)})$ will be approximately exchangeable.

We remark that \cite{dunn2022distribution} also make a similar exchangeability assumption. 
However, they train a common $\hat\mu$ 
for
 all branches. 
For a given $k\in [18]$, the residuals $|Y_i^{(k)}-\hat\mu(X_i^{(k)})|,i\in [9]$ may not necessarily be exchangeable, since $\hat\mu$ could be very different from $\mu_k^*$ and thus these residuals can be strongly affected by $X_i^{(k)}$.

\subsubsection{Additional Empirical Data Plots}
In this section, we present the additional plots on the length of prediction sets and coverage probability via various methods with level $\alpha=0.20$ in the following figure \ref{fig:cov}.
\begin{figure}[ht]
	\centering
	\includegraphics[width=1.00\textwidth]{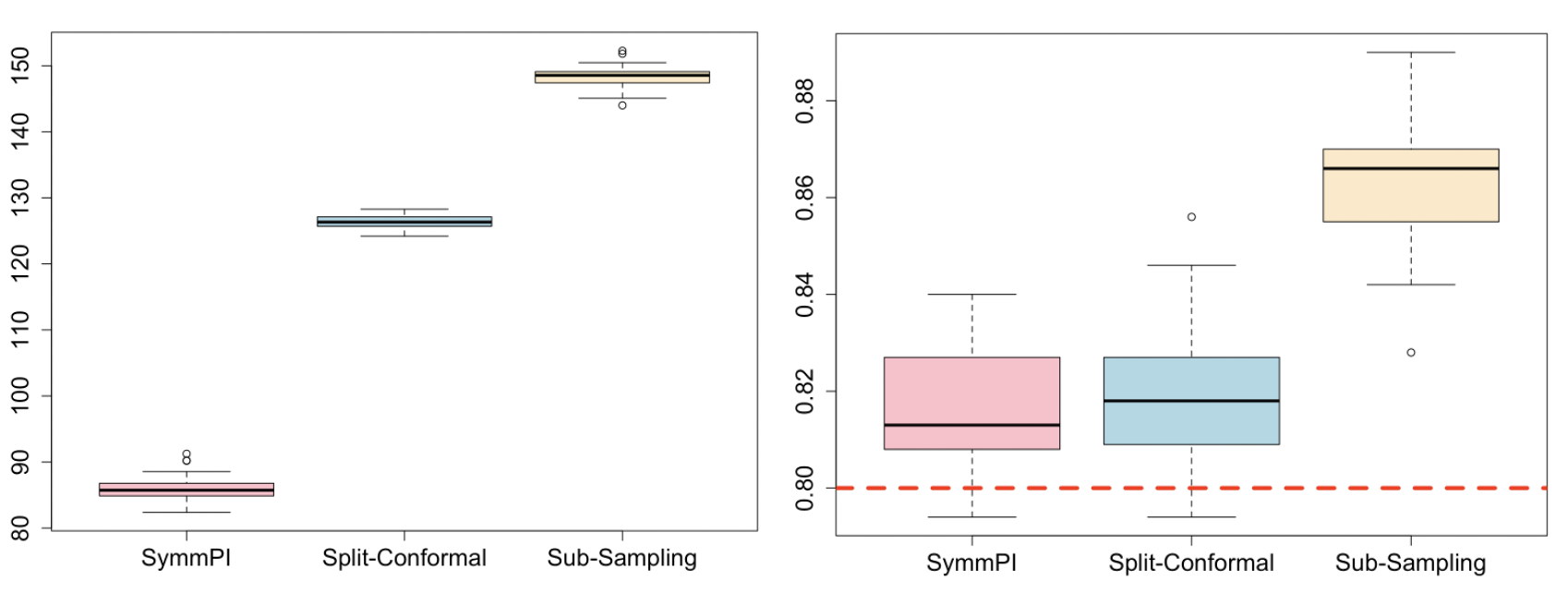}
	 \caption{Prediction set lengths and coverage probabilities for various methods with level $\alpha=0.20$.  Left: Prediction set lengths; Right: Empirical coverage probabilities.}
           \label{fig:cov}
\end{figure}


\section{Proofs}
\subsection{Proof of Theorem \ref{c1}}
\label{pfc1}

\begin{proof}
Since $\o(z)= \zo$ by definition,
\begin{align}\label{dp}
P(Z\in T_r(\Zo ))
&=
P(
\psi(V(Z)) < t_{V(Z)}
\textnormal{ or }
\psi(V(Z))=t_{V(Z)},\,
U'  < \delta_{V(Z)}).
\end{align}
Now, for all $\tz \in \tmZ$, by the definitions of $t$ from \eqref{t}
and $\delta$ from \eqref{del},
we have
\begin{align*}
P_G (\psi(\trG  \tz) < t_{\tz}
\textnormal{ or }
\psi(\trG\tz)=t_{\tz},\,
U'  < \delta_{\tz}) 
= F^-_{\tz}(t_{\tz})  + F'_{\tz}(t_{\tz}) \delta_{\tz} 
= 1-\alpha.
\end{align*}
Hence, letting $\tz = V(z)$, we have
for all $z\in \mZ$ that
$ P_G (\psi(\trG  V(z)) < t_{V(z)}
\textnormal{ or }
\psi(\trG V(z))=t_{V(z)},\,
U'  < \delta_{V(z)}) = 1-\alpha$.
Therefore, 
using that
$Z$ is $\mG$-distributionally-invariant
and
$t$ is $\mG$-invariant,
\eqref{dp} equals
\begin{align*}
& P_{G,Z}(
\psi(\trG V(Z)) < t_{\trG V(Z)}
\textnormal{ or }
\psi(\trG V(Z))=t_{\trG V(Z)},\,
U'  < \delta_{\trG V(Z)})\\
&= P_{G,Z}(
\psi(\trG V(Z)) < t_{V(Z)}
\textnormal{ or }
\psi(\trG V(Z))=t_{V(Z)},\,
U'  < \delta_{V(Z)})\\
&=
\E_{Z}[P_{G}
(\psi(\trG V(Z)) < t_{V(Z)}
\textnormal{ or }
\psi(\trG V(Z))=t_{V(Z)},\,
U'  < \delta_{V(Z)}))]
= 1-\alpha.
\end{align*}
This proves the first relation in \eqref{lb}.
The second relation follows since $T_r(\zo) \subset T^{\mathrm\OurMethod}(\zo)$.
Next, for $\tz \in \tmZ$,
by the definitions of $t$ from \eqref{t}
and of $F'$,
$$ 
P_G (\psi(\trG  \tz) \le t_{\tz}) 
= F_{\tz}(t_{\tz}) 
\le 
1-
\alpha+ F'_{\tz}(t_{\tz}).$$
Hence, 
as above
\begin{align}
P(Z\in T^{\mathrm\OurMethod}(\Zo ))
&=
P(\psi( V(Z))\le t_{V(Z)},\, \o(Z)= \zo)
=P(\psi( V(Z))\le t_{V(Z)})\label{defp}\\
&=
\E_{Z}[P_{G}(\psi(\trG  V(Z)) \le t_{V(Z)})]
\le
1-\alpha+ \E_{Z}[F'_{V(Z)}(t_{V(Z)})],\nonumber
\end{align}
proving the third relation in \eqref{lb}.

\end{proof}

\subsection{Proof of Proposition \ref{ub}}
\label{pub}

Let $\mF_{\tz}=\{\psi(\trg \tz),\, g\in \mG\}$ be the set of values of $\psi(\trg \tz)$, $g\in \mG$.
If
$\mH_{\tz}$ is a subgroup of $\mG$, 
clearly for any $c \in \mF_{\tz}$, there is a unique coset $\mC = \mH_{\tz} \cdot g$, for some $g \in \mG$, such that for all $g' \in \mC$, $\psi(\trgp\tz)=c$.
Thus, the size of each coset is $|\mH_{\tz}|$.
Thus, $F'_{\tz}(x) \le |\mH_{\tz}|/|\mG|$ for all $x\in \R$, 
which implies that 
$    P(Z\in T^{\mathrm\OurMethod}(\Zo ))\le 1-\alpha+\E|\mH_{V(Z)}|/|\mG|.$
Moreover, if $\mH_{\tz'}$ 
does not depend on $\tz'$, 
then, $\mH$ is clearly a group.
In this case, 
$F'_{V(Z)}(x) \le |\mH|/|\mG|$ for all $x\in \R$,
almost surely for $Z\sim P$;
and hence
$    P(Z\in T^{\mathrm\OurMethod}(\Zo ))\le 1-\alpha+|\mH|/|\mG|.$

\subsection{Proof of Proposition \ref{sample_g}}
\label{pfsample_g}
The proof of Proposition \ref{sample_g} is derived using the same method as outlined in Theorem 4.2 of \cite{dobriban2023joint}. Therefore, we omit the details.

\subsection{Proof of Theorem \ref{thm_shift}}
\label{pfthm_shift}

\begin{proof}
By \eqref{defp},
we have
\begin{align*}
&P_Z \Big(Z\in T^{\mathrm\OurMethod}(\Zo )\Big) 
=\bigg\{P_Z \Big(\psi( V(Z)) \le t_{V(Z)}\Big)-P_{G,Z}\Big(\psi(\trG  V(Z))  \le t_{V(Z)}\Big)\bigg\}\\
&+P_{G,Z}\Big(\psi(\trG  V(Z))  \le t_{V(Z)}\Big)
=\mathbf{(i)}+\mathbf{(ii)}.
\end{align*}

For the first term, by definition,
\begin{align*}
|\mathbf{(i)}|&= \Bigg|\int_{G}\int_{Z} 
\bigg[I\Big( \psi(V(Z))\le t_{V(Z)} \Big)- I\Big( \psi(\tilde{\rho}(G)V(Z))\le t_{V(Z)} \Big) \bigg]
dP(Z) dP(G)\Bigg|
\\& =\Bigg|\int_{Z} 
\E_{G'\sim U(\mG/\mH)}
\bigg[I\Big( \psi(V(Z))\le t_{V(Z)} \Big)- I\Big( \psi(\tilde{\rho}(G')V(Z))\le t_{V(Z)} \Big)\bigg] dP(Z)\Bigg| \le \Delta.
\end{align*}
In the last line, we have used that 
$I\Big( \psi(V(z))\le t_{V(z)} \Big) = I(\nu(\tz)\le 0)$
and
$I\Big( \psi(\tilde{\rho}(G')V(z))\le t_{V(z)} \Big) = I(\nu(\tilde{\rho}(G')\tz)\le 0)$; which also uses that $t$ is $\mG$-invariant.

For the second term, by the definition of $t_{\tz}$ from \eqref{t}, we have $\mathbf{(ii)}\le \alpha.$
Combining the upper bounds on $\mathbf{(i)}$ and $\mathbf{(ii)}$, we obtain the lower bound in \eqref{ds}.
Next, 
following the same proof procedure as in Theorem \ref{c1}, for term $\mathbf{(ii)}$, we obtain that
\begin{align*}
    P_{G,Z}\Big(\psi(\trG  V(Z))  \le t_{V(Z)}\Big)\le 1-\alpha+\E[F'_{V(Z)}(t_{V(Z)})].
\end{align*}
Combining this with the upper bound of $\mathbf{(i)}$, 
the upper bound in \eqref{ds} follows.
\end{proof}

\subsection{Proof of Propositions \ref{meta_unsup} and \ref{meta_sup}}
\label{pfm}
The proof of Propositions \ref{meta_unsup} and \ref{meta_sup} follows directly from our construction of the deterministic equivariant massage-passing graph neural network and the distributional invariance of the input variable $Z$. Therefore, we omit the details.

\subsection{Proof of Theorem \ref{prop_random}}\label{app: prop_random_proof}
\begin{proof}
Due to the data-generating process and the definition of $t_{\tz}$ in \eqref{quantile_random}, 
for the given $\vec n=(n_1,\ldots,n_K)^\top $, 
it holds that $t_{\tz}=t_{g_{\vec n}\tz}$, for any 
$g_{\vec n} \in \mG_{\vec n}=\S_{n_1}\otimes\S_{n_2}\otimes \ldots \otimes \S_{n_K}$.
According the definition of $\mathcal{G}_{\vec n}$, and $t_{(\cdot)}$, 
recalling $\psi_{\vec n}(\tz)=\tz_{n_K}^{(K)}$ for all $\tz\in \tmZ$,
it holds 
conditionally on  $\vec n$
that 
\begin{align*}
&P_{Z} \Big(\psi_{\vec n}(V(Z)) >  t_{V(Z)}\Big) 
= P_{G_{\vec n},Z}\Big(\psi_{\vec n}(G_{\vec n} \cdot V(Z))  >  t_{V(Z)}\Big)\\
&\qquad
=\E_{Z}\E_{G_{\vec n}}\Big[I\Big(\psi_{\vec n}(G_{\vec n}\cdot  V(Z))>t_{V(Z)}\Big)\Big]
=
\E_{Z}\bigg[\frac{1}{n_K}\sum_{j=1}^{n_K}I([V(Z)]_{j}^{(K)} > t_{V(Z)})\bigg].
\end{align*}

We next define $\tilde{\psi}(x)=e_{K}^\top x$, 
where $e_K\in \mathbb{R}^K$ with the $K$-th entry being unity and others being zero, 
for all $x\in \mathbb{R}^K$. 
Furthermore, for the given $n$, we let $h:\mathcal{Z}^*\rightarrow \in \mathbb{R}^K$ by defining its $m$-th output entry as $[h(Z_1,\ldots,Z_K)]_m=\frac{1}{n_m}\sum_{j=1}^{n_m}I(V(Z_{j}^{(m)})>t_{V(Z)}), m\in [K]$.

Next, we take the randomness of $Z,\vec N$ into consideration. 
By the exchangeability of $(Z_1,\ldots,Z_K)^\top$ and the definition of $t_{(\cdot)}$, 
it holds that $h(Z_1,\ldots, Z_K)$ $=_d G' h(Z_1,\ldots,Z_K)$ for $G'\sim \mathrm{Unif}(\S_K)$. 
Therefore, according to the definition of $\tilde{\psi}$ and $h(\cdot)$, it holds that 
\begin{align*}
\E_{Z}\bigg[\frac{1}{N_K}\sum_{j=1}^{N_K}I([V(Z)]_{j}^{(K)} > t_{V(Z)})\bigg]&=\E_{Z}\bigg[\tilde{\psi}(h(Z))\bigg]
=\E_{Z,G'}\bigg[\tilde{\psi}(G'\cdot h(Z))\bigg]
\\&= \E_{Z}\bigg[\frac{1}{K}\sum_{k=1}^{K}\frac{1}{N_k}\sum_{j=1}^{N_k}I([V(Z)]_{j}^{(k)} >  t_{V(Z)})\bigg]
\le \alpha.
\end{align*}

Therefore, we obtain
\begin{align*}
P\Big(Z\in T^{\mathrm\OurMethod}(\Zo,\vec N)\Big)
&=
P_{N,Z}(\psi_{\vec N}(V(Z))\le t_{V(Z)},\, \o(Z)= \Zo)\ge 1- \alpha.
\end{align*}
The upper bound can be proved in a similar way as the corresponding upper bound in the proof of Theorem \ref{c1}, and we omit the details. 
\end{proof}

\subsection{Proof of Theorem \ref{thm_asy}}
\label{pfthm_asy}

By the definition of the prediction set in \eqref{non-asym-alg}, we have 
\begin{align*}
&P\bigg(\psi(\tir(G)V(\rho^{-1}(G)Z)> Q_{1-\alpha}\Big(\sum_{j=1}^{|\mathcal{F}|}w_j\delta_{\psi_j(V(\rho^{-1}(G)Z))}\Big) \bigg)\\&\qquad=\sum_{i=1}^{|\mathcal{F}|}w_i P\bigg(\psi(\tilde{\rho}(G)V(\rho^{-1}(G)Z)> Q_{1-\alpha}\Big(\sum_{j=1}^{|\mathcal{F}|}w_j\delta_{\psi_j(V(\rho^{-1}(G)Z))}\Big)\Big{|}G=g_i\bigg).
\end{align*}
Since $G$ is independent of $Z$, we obtain
\begin{align*}
&\sum_{i=1}^{|\mathcal{F}|}w_i P\bigg(\psi(\rho(G)V(\rho^{-1}(G)Z)> Q_{1-\alpha}\Big(\sum_{j=1}^{|\mathcal{F}|}w_j\delta_{\psi_j(V(\rho^{-1}(G)Z))}\Big)\Big{|}G=g_i\bigg)\\&\qquad =\sum_{i=1}^{|\mathcal{F}|}w_i P\bigg(\psi_i(V(\rho^{-1}(g_i)Z)> Q_{1-\alpha}\Big(\sum_{j=1}^{|\mathcal{F}|}w_j\delta_{\psi_j(V(\rho^{-1}(g_i)Z))}\Big)\bigg)
\\ &\qquad\le \sum_{i=1}^{|\mathcal{F}|}w_i P\bigg(\psi_i(V(Z))> Q_{1-\alpha}\Big(\sum_{j=1}^{|\mathcal{F}|}w_j\delta_{\psi_j(V(Z))}\Big)\bigg)+\sum_{i=1}^{|\mathcal{F}|}w_i 
 \mathrm{TV}\big(\nu_i(V(\rho^{-1}(g_i)Z)),\nu_i(V(Z))\big),
\end{align*}
where $\nu_i(x)=\psi_i(x)-Q_{1-\alpha}(\sum_{j=1}^{|\mathcal{F}|}w_j\delta_{\psi_j(x)}).$ 
The inequality above holds since we can write $$P\bigg(\psi_i(V(Z))> Q_{1-\alpha}\Big(\sum_{j=1}^{|\mathcal{F}|}w_j\delta_{\psi_j(V(Z))}\Big)\bigg)=I\big(\nu_i(V(Z))>0\big)$$ and $$I\bigg(\psi_i(V(\rho^{-1}(g_i)Z)> Q_{1-\alpha}\Big(\sum_{j=1}^{|\mathcal{F}|}w_j\delta_{\psi_j(V(\rho^{-1}(g_i)Z))}\Big)\bigg)=I\big(\nu_i(V(\rho^{-1}(g_i)Z))>0\big).$$

Subsequently, we prove upper and lower bounds for term $$\sum_{i=1}^{|\mathcal{F}|}w_i P\bigg(\psi_i(V(Z))> Q_{1-\alpha}\Big(\sum_{j=1}^{|\mathcal{F}|}w_j\delta_{\psi_j(V(Z))}\Big)\bigg).$$
For the upper bound, this term is at most $\alpha$ by definition.
For the lower bound,
following an argument similar to the proof of
Theorem \ref{c1}, we obtain
$$\sum_{i=1}^{|\mathcal{F}|}w_i P\bigg(\psi_i(V(Z))\le  Q_{1-\alpha}\Big(\sum_{j=1}^{|\mathcal{F}|}w_j\delta_{\psi_j(V(Z))}\Big)\bigg)\le 1-\alpha+\E_Z[F^{w'}_{V(Z)}(t^w_{V(Z)})].$$
We then conclude the proof.

\end{document}